\documentclass[%
twocolumn,
superscriptaddress,
nofootinbib,
amsmath,amssymb,
aps,
prx
]{revtex4-2}

\usepackage{microtype}
\usepackage{graphicx}
\usepackage{booktabs}
\usepackage{multirow}
\usepackage{placeins}
\usepackage[T1]{fontenc}
\usepackage{ytableau}

\usepackage{hyperref}
\hypersetup{
  colorlinks=true,
  linkcolor=blue,
  citecolor=blue,
  urlcolor=blue
}

\usepackage{amsmath}
\usepackage{amssymb}
\usepackage{mathtools}
\usepackage{amsthm}

\makeatletter
\newsavebox{\@brx}
\newcommand{\llangle}[1][]{\savebox{\@brx}{\(\m@th{#1\langle}\)}%
  \mathopen{\copy\@brx\kern-0.5\wd\@brx\usebox{\@brx}}}
\newcommand{\rrangle}[1][]{\savebox{\@brx}{\(\m@th{#1\rangle}\)}%
  \mathclose{\copy\@brx\kern-0.5\wd\@brx\usebox{\@brx}}}
\makeatother

\usepackage[dvipsnames]{xcolor}

\newcommand{\ket}[1]{|#1\rangle}
\newcommand{\bra}[1]{\langle#1|}

\newcommand{\SO}{\mathrm{SO}}
\newcommand{\SU}{\mathrm{SU}}
\newcommand{\Sym}{\mathfrak{S}}

\DeclareMathOperator{\Tr}{Tr}
\DeclareMathOperator{\rank}{rank}

\usepackage[capitalize,noabbrev]{cleveref}

\theoremstyle{plain}
\newtheorem{theorem}{Theorem}
\newtheorem{proposition}[theorem]{Proposition}

\newtheorem{observation}[theorem]{Observation}
\theoremstyle{definition}

\theoremstyle{remark}

\begin{document}

\title{HyQuRP: Hybrid quantum-classical neural network with rotational and permutational equivariance}

\author{Semin Park}
\email{s.park.semin@gmail.com}
\affiliation{School of Integrated Technology, Yonsei University, Seoul 03722, Korea}

\author{Chae-Yeun Park}
\email{chae-yeun@yonsei.ac.kr}
\affiliation{School of Integrated Technology, Yonsei University, Seoul 03722, Korea}
\affiliation{Department of Quantum Information, Yonsei University, Seoul 03722, Korea}
\affiliation{BK21 Graduate Program in Intelligent Semiconductor Technology, Seoul 03722, Korea}

\date{\today}

\begin{abstract}
Group-equivariant quantum machine learning has emerged as a promising paradigm by incorporating symmetry into quantum models.
However, constructing models simultaneously equivariant to both rotational and permutational symmetries in a principled manner remains a bottleneck.
In this work, we develop a general framework for dual-equivariant gates under rotations and permutations and analyze the dimension of the resulting gate space using group representation theory. 
Based on this, we introduce HyQuRP, a hybrid quantum-classical neural network with dual equivariance. 
On 3D point cloud classification benchmarks in the sparse-point regime, HyQuRP outperforms strong classical and quantum baselines.
For example, when six subsampled points are used, HyQuRP ($\sim$1.5K parameters) achieves 76.13\% accuracy on the 5-class ModelNet benchmark, compared with 72.54\%, 71.09\%, and 71.03\% for Tensor Field Network, PointNet, and PointMamba with similar parameter counts.
These results highlight HyQuRP's strong data efficiency and suggest the potential of equivariant quantum machine learning approaches in symmetry-sensitive tasks.
\end{abstract}

\maketitle

\section{Introduction}

Embedding group equivariance into a neural network is one of the most successful principles for designing machine learning models that process data whose features are invariant under group transformations~\cite{cohen2016group}.
A convolutional layer~\cite{lecun1989backpropagation,krizhevsky2012imagenet} is a classic example that preserves translational symmetry in the input, which is relevant for image data whose characteristics are invariant under such transformations.
Recently, neural networks equivariant under more advanced symmetries, such as the symmetric group $\Sym_n$~\cite{zaheer2017deep,lee2019set}, Euclidean group $E(3)$~\cite{satorras2021n,du2022se,batzner20223}, the special orthogonal group $\SO(3)$~\cite{esteves2018learning,deng2021vector,kim2024continuous}, and the Lorentz group $\SO(1,3)$~\cite{bogatskiy2020lorentz}, have been proposed.
These studies report higher classification accuracy and better sample efficiency with fewer parameters across various datasets.

Quantum machine learning (QML) is an emerging area that combines quantum
computing with machine learning. Broadly, QML approaches can be grouped into parameterized quantum circuit models, which learn through trainable quantum gates~\citep{farhi2018qnn,mitarai2018qcl,schuld2020circuitcentric,abbas2021power}, and quantum kernel methods, which use quantum feature maps in kernel-based classification~\citep{havlicek2019qefs,schuld2021supervised}. However, across a range of standard classification tasks, existing QML models have not yet convincingly surpassed strong classical baselines. Several empirical studies report that their performance often lags behind well-tuned classical models~\citep{kuros2022trafficsign,bowles2024better,schnabel2025quantumkernel}, including simple architectures such as MLPs or CNNs, in capacity-controlled comparisons. This persistent gap suggests the need for stronger inductive biases in QML.

Equivariant QML is emerging as a particularly promising direction for narrowing this gap. Indeed, recent equivariant QML models have shown performance gains by incorporating symmetries such as permutations, reflections, $\SU(2)$ rotations, and Lorentz transformations~\cite{meyer2023exploiting,West2023ReflectionEquivariantQNN,le2025symmetry,larocca2022group,nguyen2024theory,jahin2025lorentz,east2026all}. Yet the broader challenge remains open: how to build QML models that accommodate multiple symmetries simultaneously without relying on ad hoc constructions.
This challenge becomes particularly acute in settings where rotational and permutational symmetries must be preserved at the same time, including 3D point clouds, molecular structures, and crystalline materials, among others.
However, in the standard qubit setting where these symmetries are imposed globally, their joint realization is obstructed by the Schur--Weyl duality, which makes such dual-equivariant gates essentially trivial. 
Addressing this obstruction calls for a general treatment of how both symmetries can be incorporated simultaneously and lays the foundation for an equivariant QML architecture that preserves both the rotational and permutational symmetries.

In this paper, we first examine the challenges of equivariant QML under both rotations and permutations. 
We then develop a general framework for constructing and analyzing gates that are equivariant under both symmetries. This framework overcomes the previous bottleneck by replacing qubit-wise permutations with block-wise permutations.
Based on this, we introduce HyQuRP, a hybrid quantum-classical neural network with rotational and permutational equivariance. HyQuRP consists of four quantum stages and a classical head, each carefully designed to preserve both the rotation and permutation of input data.
As a concrete testbed, we evaluate HyQuRP on standard 3D point cloud benchmarks in a sparse-point regime to assess its data efficiency. Across two benchmarks, HyQuRP outperforms classical and quantum state-of-the-art (SOTA) models that are tuned to have similar parameter counts.

The main contributions of this work are summarized as follows:

\begin{itemize}
    \item We introduce a general construction of dual-equivariant gates under rotations and permutations and characterize the dimension of the corresponding gate space, laying a principled foundation beyond ad hoc constructions.
    
    \item We propose HyQuRP, a hybrid quantum-classical architecture that performs classification invariant to global rigid rotations and permutations of the input points. All stages of HyQuRP are organized within a unified representation-theoretic framework.
    
    \item We conduct experiments against strong classical and quantum 3D point cloud baseline models and find that HyQuRP generally surpasses them.
\end{itemize}

\section{Preliminaries}
This section reviews mathematical foundations for the design of group-equivariant quantum circuits.
We only present several significant results here, while providing a pedagogical review in Appendix~\ref{app:group_theory}.

\subsection{Group equivariance and invariance}
Let $G$ be a group and let 
$\rho_X : G \to \mathrm{GL}(X)$ and $\rho_Y : G \to \mathrm{GL}(Y)$ 
be representations of $G$ on $X$ and $Y$, respectively.
A function $f : X \to Y$ is called $G$-invariant if
$f\bigl(\rho_X[g] \cdot x\bigr) = f(x)$ and $G$-equivariant if
$f\bigl(\rho_X[g] \cdot x\bigr) = \rho_Y[g] \cdot f(x)$, for all $g \in G$ and $x \in X$.

A common approach to constructing invariant networks is to compose $G$-equivariant layers with invariant aggregation operations:
\begin{equation}
f_{\mathrm{inv}}(x)
= \mathrm{Agg}\bigl(f_{\mathrm{eq}}(x)\bigr)
\end{equation}
where $f_{\mathrm{eq}}$ preserves group structure and $\mathrm{Agg}$ removes it, yielding overall $G$-invariance. 

For a quantum circuit $C$, the equivariance condition translates to
\begin{align}
R(g)C \ket{\psi} = CR(g) \ket{\psi},
\end{align}
where $R(g)$ is a group representation of $g \in G$ to a unitary operator, and $\ket{\psi}$ is a quantum state.
To satisfy this condition for arbitrary $\ket{\psi}$, we need
\begin{align}
R(g) C R(g)^\dagger = C
\end{align}
for all $g \in G$, which implies that $C$ is an element of the commutant algebra.

Typically, for QML models, one can further extract invariant outputs from a $G$-equivariant circuit using $G$-invariant input states and observables. We will discuss such details in Sec.~\ref{sec:Method}.

\subsection{Schur--Weyl duality}

Let \(V \cong \mathbb{C}^d\) and \(W := V^{\otimes n}\). There are two natural actions on \(W\): 
the action of \(\mathrm{GL}(V)\), given by 
\begin{align}
    \rho : \mathrm{GL}(V) \to \mathrm{GL}(W), \qquad \rho(T)=T^{\otimes n},
\end{align}  
and the action of the symmetric group $\mathfrak S_n$, given by 
$\Pi : \mathfrak S_n \to \mathrm{GL}(W)$,
where $\Pi(\sigma)$ permutes the tensor factors according to \(\sigma \in \mathfrak S_n\), i.e., 
it transforms a product state $\ket{i_1,i_2,\cdots,i_n}$ to 
\begin{align}
    \Pi(\sigma) \ket{i_1,i_2\cdots,i_n} = \ket{i_{\sigma^{-1}(1)},i_{\sigma^{-1}(2)},\cdots, i_{\sigma^{-1}(n)}}.
\end{align}
Thus, for any $T \in \mathrm{GL}(V)$ and $\sigma \in \Sym_n$, $T^{\otimes n} \Pi(\sigma) =\Pi(\sigma)T^{\otimes n}$.

Schur--Weyl duality states that \(W\) admits the simultaneous decomposition
\begin{align}
W
\cong
\bigoplus_{\lambda \vdash n,\; \ell(\lambda)\le d}
\mathcal{V}^{\lambda}\otimes \mathfrak{S}^{\lambda},
\end{align}
where \(\mathcal{V}^{\lambda}\) and \(\mathfrak{S}^{\lambda}\) denote the irreducible representations of \(\mathrm{GL}(V)\) and \(\mathfrak{S}_n\), respectively.
The direct sum is over all partitions of $n$, $\lambda \vdash n$, whose length $\ell(\lambda)$ is not greater than $d$.
For example, $\lambda \in \{(4), (3,1), (2,2), (2, 1, 1), (1,1,1,1)\}$ are valid partitions of $n=4$ whose lengths are $1, 2, 2, 3, 4$, respectively  
(see Appendices~\ref{app:symmetric_group} and \ref{app:schur_weyl_duality} for further details).

In the \(n\)-qubit setting, where \(W=(\mathbb{C}^2)^{\otimes n}\), the action \(\rho\) above can naturally be restricted to the subgroup \(\mathrm{SU(2)}\subset \mathrm{GL}(\mathbb{C}^2)\), which implements single-qubit state transformations. 

Consequently, the condition \(\ell(\lambda)\le 2\) restricts \(\lambda\) to two-part partitions of the form \((n-k,k)\)  and \(\mathcal{V}^{\lambda}\) is identified with the irreducible representation of $\mathrm{SU(2)}$ which is nothing more than a spin-$J$ space with total spin $J=n/2-k$. Hence,
\begin{align}
\label{eq:schur-weyl}
(\mathbb{C}^2)^{\otimes n}
\cong
\bigoplus_{k=0}^{\lfloor n/2\rfloor}
J^{(n/2-k)} \otimes \mathfrak{S}^{(n-k,k)},
\end{align}
where $J^{(n/2-k)}$ is a $(n-2k+1)$-dimensional vector space representing total spin $J=n/2-k$. 
This decomposition will be used in Sec.~\ref{sec:theory} to analyze the constraints on \(W=(\mathbb{C}^2)^{\otimes n}\) arising from these two natural actions.

\subsection{SU(2) equivariant gates}
\label{sec:su2-equiv-gates}

Single-qubit gates are naturally represented as elements of SU(2), which is a double cover of SO(3).
This fundamental relationship allows quantum circuits to be rotation-equivariant in 3D space through the covering map $\SU(2) \rightarrow \mathrm{SO}(3)$, enabling direct processing of geometric data while preserving rotational symmetries. 
As discussed above, symmetric group actions commute with $\SU(2)$ transformations.
This commutation property makes permutations inherently $\SU(2)$-equivariant, providing a foundation for constructing more general equivariant operations.

We extend this concept by introducing the space of generalized permutations, $\mathbb{C}[\Sym_n]$, defined as:
\begin{equation}
\mathbb{C}[\Sym_n] = \Bigl\{ \sum_{\sigma \in \Sym_n} c_\sigma \sigma \Bigr\},
\end{equation}
and its representation
\begin{align}
\Pi(\mathbb{C}[\Sym_n]) = \Bigl\{ \sum_{\sigma \in \Sym_n} c_\sigma \Pi(\sigma) \Bigr\},
\end{align}
where each $c_\sigma \in \mathbb{C}$ is a complex coefficient and the sum is over all permutations. 
For $P \in \Pi(\mathbb{C}[\Sym_n])$, its exponential map is defined as 
\begin{align}
\exp(P) = \sum_{k=0}^\infty \frac{1}{k!}\Bigl( \sum_{\sigma \in \Sym_n} c_\sigma \Pi(\sigma)  \Bigr)^k \in \Pi(\mathbb{C}[\Sym_n]),
\end{align}
since $\mathbb{C}[\Sym_n]$ is closed under addition and multiplication.

Remarkably, this construction is not merely sufficient but also necessary: leveraging Schur's lemma, Schur--Weyl duality, and various tools from representation theory, one can show that every $\mathrm{SU}(2)$-equivariant unitary operator admits such a representation (see Appendix~\ref{app:generalized_permutation} for details).

\subsection{Group twirling}
The twirling formula provides a systematic method for projecting arbitrary operators onto the subspace of group-equivariant operators through group averaging. For a group $G$ acting on a Hilbert space with unitary representation $V[g]$ and operator $A$, the twirling operation is defined as
\begin{equation}
\mathcal{T}_G[A] = \frac{1}{|G|} \sum_{g \in G} V[g] A V[g]^{\dagger},
\end{equation}
for a finite group and 
\begin{equation}
\mathcal{T}_G[A] = \int d\mu(g) \, V[g] A V[g]^{\dagger}
\end{equation}
for a continuous group (compact group) with the normalized Haar measure $\mu(g)$.

By construction, the twirled operators $T_G[A]$ commute with all representation matrices, i.e., commute with $V[g]$ for all $g \in G$, ensuring that they are group equivariant (see Appendix~\ref{app:group_twirling} for a proof).
This averaging procedure is fundamental for constructing group-equivariant neural network layers that respect the underlying symmetries of data.

\section{Theory of dual-equivariant quantum circuits from pair permutations}
\label{sec:theory}
Although models that achieve dual equivariance under both rotations and permutations are crucial across a range of geometric and physical settings, their principled construction has remained a bottleneck. In this section, we first explain why formulating dual-equivariant models is difficult under Schur--Weyl duality when we consider permutations between all qubits. 
We then show that meaningful dual-equivariant operators can be constructed when the permutation is restricted to a \textit{subgroup} of the symmetric group, and find a general form of such dual-equivariant operators.
Choosing the pair permutation as a subgroup, where each block contains two qubits, and those blocks are allowed to be permuted without changing the internal order, we quantify the expressive size of this gate class by computing the dimension of the corresponding generator space. 
Detailed proofs of the results given in this section are provided in Appendix~\ref{app:proof}.

\subsection{Constraints on dual equivariance}

Let \(\mathfrak{S}_n\) be the symmetric group on \(n\) symbols, and let \(\Pi : \mathfrak{S}_n \to \mathrm{GL}(W)\) be its permutation representation on
\(W = (\mathbb{C}^2)^{\otimes n}\). We also write \(\rho(U)=U^{\otimes n}\) for \(U\in \mathrm{SU}(2)\), and regard this as the global \(\mathrm{SU}(2)\)-action on \(W\).

By the Schur--Weyl decomposition in Eq.~\ref{eq:schur-weyl}, \(W\) decomposes into irreducible $\mathrm{SU(2)}$- and $\mathfrak{S}_n$-modules. Using this decomposition, we first show that an operator equivariant under any permutation and $\SU(2)$ group transformation must act trivially within a subspace. 
\begin{observation}
\label{obv:joint_su2_sn_equiv_gate}
If an operator $M$ satisfies
\begin{align}
[M,\rho(U)] = 0
\; \forall U \in \mathrm{SU}(2),\;
\text{and}\;
[M,\Pi(\sigma)] = 0
\; \forall \sigma \in \mathfrak{S}_n,
\label{eq:joint_equivariance_prop}
\end{align}
then
\begin{align}
M
=
\bigoplus_{k=0}^{\lfloor n/2\rfloor}
c_k\,
I_{J^{(n/2-k)}}
\otimes
I_{\Sym^{(n-k,k)}}
\;
\text{for some } c_k\in\mathbb{C}.
\label{eq:joint_equivariance_scalar_blocks}
\end{align}
In particular, if $M$ is unitary, then $c_k=e^{i\theta_k}$ for some $\theta_k\in\mathbb{R}$.
\end{observation}

\begin{proof}[Proof sketch]
For each \(k\), consider the summand
\(
J^{(n/2-k)}\otimes \Sym^{(n-k,k)}
\)
in the decomposition of \(W\).
By Schur's lemma, the condition \([M,\rho(U)]=0\) forces \(M\) to act as
\(I\otimes B_k\) on this block, whereas \([M,\Pi(\sigma)]=0\) forces it to act
as \(C_k\otimes I\). Hence, the block must be scalar. Summing over \(k\) gives
the claimed decomposition, and unitarity implies \(|c_k|=1\).
See Appendix~\ref{app:proof_joint_su2_sn_equiv_gate} for detailed proof.
\end{proof}

The above observation tells us that dual-equivariant gates are fully characterized by $\lfloor n/2 \rfloor + 1$ free parameters, thus occupy only an exponentially small subset of the full operator space acting on the Hilbert space, $W=(\mathbb{C}^2)^{\otimes n}$.

We next show that a vector that is invariant under any permutation and $\SU(2)$ group transformation must be a zero vector.

\begin{observation}
\label{obv:no_joint_fixed_vectors}
Define
\begin{equation}
W^{\mathrm{SU}(2)}
:=
\{v\in W : \rho(U)v=v \text{ for all } U\in \mathrm{SU}(2)\}
\end{equation}
and
\begin{equation}
W^{\mathfrak{S}_n}
:=
\{v\in W : \Pi(\sigma)v=v \text{ for all } \sigma\in \mathfrak{S}_n\}.
\end{equation}
Then
\begin{equation}
W^{\mathrm{SU}(2)} \cap W^{\mathfrak{S}_n} = \{0\}.
\end{equation}
\end{observation}

\begin{proof}
Using the Schur-Weyl duality, a vector fixed by both
\(\mathrm{SU}(2)\) and \(\mathfrak S_n\) would have to lie in a summand whose
\(\mathrm{SU}(2)\)-factor and \(\mathfrak S_n\)-factor are both trivial.
However, the former occurs only for \(k=\frac n2\), whereas the latter occurs
only for \(k=0\). Hence, no such nonzero summand exists, so
\(
W^{\mathrm{SU}(2)} \cap W^{\mathfrak S_n}=\{0\}.
\)
See Appendix~\ref{app:proof_no_joint_fixed_vectors} for detailed proof.
\end{proof}
Obviously, since the zero vector is not normalizable, it cannot be a valid quantum state.

Taken together, the above two observations show that, under the permutation action of \(\mathfrak{S}_n\) and the global \(\mathrm{SU}(2)\)-action on \(W\), imposing both symmetries simultaneously is highly restrictive: jointly equivariant operators have severely limited expressive power, and no nontrivial jointly invariant pure state exists. Consequently, these two constraints constitute a bottleneck in constructing a model that is simultaneously invariant under both actions.

\subsection{General structure of dual-equivariant operators}
\label{sec:4.3.1}
In the previous subsection, we showed that an operator that commutes with both the $\Sym_n$ and $\SU(2)$ groups must act trivially on each subspace, thereby limiting its applicability.
Instead, we now consider an operator that commutes only with elements within $H \leq \Sym_n$.
The motivation for considering a subgroup will become clear in the next subsection.

We begin by defining the set of all representations of generalized permutations,
\begin{equation}
\mathcal{A}_n := \Pi\bigl(\mathbb{C}[\mathfrak{S}_n]\bigr).
\end{equation}
Here, $\mathbb{C}[\mathfrak{S}_n] = \{\sum_{\sigma \in \mathfrak{S}_n} c_\sigma \sigma \}$ is a group algebra generated by $\Sym_n$.
For a subgroup $H \le \mathfrak{S}_n$, let \begin{align}
\mathbb{C}[\mathfrak{S}_n]^H
:=
\left\{
X \in \mathbb{C}[\mathfrak{S}_n]
\ \middle|\
\sigma X \sigma^{-1}=X,
\,
\forall \sigma \in H
\right\}, 
\label{eq:group_alg_H_inv}
\end{align}
and
\begin{align}
\mathcal{A}_n^H
:=
\left\{
M \in \mathcal{A}_n
\;\middle|\;
\Pi(\sigma) M \Pi(\sigma)^\dagger = M,\, \forall \sigma \in H
\right\}.
\end{align}
For any operator $M$, the twirling map over $H$ is given by \begin{equation} \label{eq:general-twirl} \mathcal{T}_H[M] := \frac{1}{|H|} \sum_{\sigma \in H} \Pi(\sigma)\, M\, \Pi(\sigma)^\dagger. \end{equation}

We first show that twirling $\mathcal{A}_n$ over \(H\) projects precisely onto the \(H\)-invariant subspace $\mathcal{A}_n^H$.

\begin{proposition}
\label{pro:twir_equal}
For any subgroup $H \le \mathfrak{S}_n$, one has
\begin{equation}
\mathcal{T}_H[\mathcal{A}_n] = \mathcal{A}_n^H.
\end{equation}
\end{proposition}

\begin{proof}
As shown in Appendix~\ref{app:group_twirling}, one always has \(\mathcal{T}_H[\mathcal{A}_n] \subseteq \mathcal{A}_n^H.\)

Conversely, if $M \in \mathcal{A}_n^H$, then
\begin{equation}
\mathcal{T}_H[M]
=
\frac{1}{|H|}
\sum_{\sigma \in H}
\Pi(\sigma)\, M\, \Pi(\sigma)^\dagger
=
\frac{1}{|H|}
\sum_{\sigma \in H} M
=
M.
\end{equation}
Hence $M = \mathcal{T}_H[M]$, which shows that $M$ lies in the image of $\mathcal{T}_H$. Therefore,
\begin{equation}
\mathcal{A}_n^H \subseteq \mathcal{T}_H[\mathcal{A}_n].
\end{equation}
It follows that
\begin{equation}
\mathcal{T}_H[\mathcal{A}_n] = \mathcal{A}_n^H.
\end{equation}
\end{proof}

Building on Proposition~\ref{pro:twir_equal}, we now introduce a generalized class of unitary operators that are equivariant under both global $\mathrm{SU}(2)$ rotations and the permutation action of a subgroup $H \le \mathfrak{S}_n$.

\begin{theorem}
\label{thm:twir_a}
A unitary operator $Q$ commutes with both the global $\mathrm{SU}(2)$ action and the permutation action of $H$ if and only if there exists a skew-Hermitian operator $A \in \mathcal{A}_n$ such that
\begin{equation}
\label{eq:dual-equiv-general-2}
Q = \exp(\mathcal{T}_H[A]).
\end{equation}
\end{theorem}

\begin{proof}
Suppose first that $Q$ commutes with both the global $\mathrm{SU}(2)$ rotations and the permutation action of $H$.
By the characterization established in Sec.~\ref{sec:su2-equiv-gates}, commutation with the global $\mathrm{SU}(2)$ action implies that
$Q \in \exp(\mathcal{A}_n)\subset \mathcal{A}_n$.
Since moreover
\begin{align}
\Pi(\sigma)\,Q\,\Pi(\sigma)^\dagger = Q
\qquad (\forall \sigma \in H),
\end{align}
we have $Q \in \mathcal{A}_n^H$.
Since $Q$ is unitary, its spectral decomposition
$Q = \sum_j e^{i\theta_j} P_j$ has projections $P_j \in \mathcal{A}_n^H$.
Setting $K := \sum_j i\theta_j P_j$ gives a skew-Hermitian element
$K \in \mathcal{A}_n^H$ with $e^K = Q$.
By Proposition~\ref{pro:twir_equal}, $\mathcal{A}_n^H = \mathcal{T}_H[\mathcal{A}_n]$,
so there exists $A \in \mathcal{A}_n$ such that $K = \mathcal{T}_H[A]$; in fact, $A$ can be chosen to be skew-Hermitian (e.g., $A := K$).
Exponentiating gives
\begin{equation}
Q = \exp(\mathcal{T}_H[A]),
\end{equation}
which proves one direction.

Conversely, suppose that $Q = \exp(\mathcal{T}_H[A])$ for some skew-Hermitian $A \in \mathcal{A}_n$.
Since $\mathcal{T}_H[A] \in \mathcal{A}_n^H$, it commutes with both the global $\mathrm{SU}(2)$ action and the permutation action of $H$.
Hence so does its exponential $Q$, proving the converse.
\end{proof}

Theorem~\ref{thm:twir_a} can be viewed as a generalization of Observation~\ref{obv:joint_su2_sn_equiv_gate}, where the full symmetric group \(\mathfrak{S}_n\) is replaced by an arbitrary subgroup \(H\). From this viewpoint, taking \(H\) to be a proper subgroup of \(\mathfrak{S}_n\) relaxes the symmetry constraints and can therefore enlarge the expressive power of the resulting equivariant operators.

\subsection{The \(\Sym_{\mathrm{pair}}\)-invariant subspace and its dimension}
We now specialize the arbitrary subgroup $H$ to the subgroup of $\Sym_{2N}$ preserving a fixed block structure on the wire indices.
Specifically, we group the wire indices into disjoint pairs,
\begin{equation}
B_\ell := \{2\ell,\,2\ell+1\}, \qquad \ell=0,\dots,N{-}1.
\end{equation}
We then define the pair-permuting subgroup
\begin{equation}
\Sym_{\mathrm{pair}}
:= \bigl\{\sigma: \{B_\ell\} \rightarrow \{B_\ell\}| \text{$\sigma$ is bijective}\bigr\} \leq \Sym_{2N},
\end{equation}
which acts by permuting the pairs as rigid blocks while preserving the internal structure of each pair.

Since every operator in $\mathcal{A}_{2N}$ is already $\mathrm{SU}(2)$-equivariant, finding an operator that is also equivariant under $\Sym_{\rm pair}$ can be used to construct dual-equivariant gates with a proper encoding.

For each $\sigma \in \Sym_{\mathrm{pair}}$, there exists a unique permutation
$h_\sigma \in \mathfrak{S}_N$ such that
\begin{equation}
\sigma(B_\ell)=B_{h_\sigma(\ell)}
\qquad (\ell=0,\dots,N{-}1).
\end{equation}
Equivalently, each $h\in\mathfrak{S}_N$ determines a unique element $\widetilde{h}\in\Sym_{\mathrm{pair}}$ by permuting the pair blocks according to $h$ while preserving the internal ordering within each block.

Now let the cycle type of $h \in \mathfrak{S}_N$ be
\begin{equation}
\lambda
=
1^{m_1(\lambda)}2^{m_2(\lambda)}3^{m_3(\lambda)}\cdots \vdash N.
\end{equation}
Each $r$-cycle of $h$ permutes $r$ pair blocks, and therefore induces two
disjoint $r$-cycles on the $2N$ wire indices: one on the first wires of the
corresponding pairs and one on the second wires. Hence the cycle type of
$\widetilde{h}$ in $\mathfrak{S}_{2N}$ is
\begin{equation}
\widetilde{\lambda}
=
1^{2m_1(\lambda)}2^{2m_2(\lambda)}3^{2m_3(\lambda)}\cdots \vdash 2N.
\end{equation}

With this choice of subgroup, Theorem~\ref{thm:twir_a} shows that every dual-equivariant unitary operator can be represented in the form
\begin{equation}
Q = \exp\!\bigl(\mathcal{T}_{\mathfrak{S}_{\mathrm{pair}}}[A]\bigr),
\qquad
A \in \mathcal{A}_{2N}.
\end{equation}

We first record the dimension formula for the pair-symmetrized subspace in the group algebra.
\begin{proposition}
\label{pro:dim_group}
The dimension of $\mathbb{C}[\mathfrak{S}_{2N}]^{\mathfrak{S}_{\mathrm{pair}}}$ is
\begin{equation}
\label{eq:pair-dim-formula}
\sum_{\lambda \vdash N}
\prod_{r \ge 1}
r^{\,m_r(\lambda)}
\frac{(2m_r(\lambda))!}{m_r(\lambda)!}.
\end{equation}
\end{proposition}
In Appendix~\ref{app:proof_dim_group}, we prove the Proposition using Burnside's lemma and counting the number of elements in $\Sym_{2N}$ that are invariant under $\Sym_{\rm pair}$.
We provide a detailed example for $N=2$ presenting all basis elements in the next subsection.

However, to quantify the expressive size of dual-equivariant gates, the relevant quantity is the dimension of the invariant operator space, rather than that of its group algebra counterpart. The two do not coincide because, for $W=(\mathbb{C}^2)^{\otimes 2N}$,
only $\mathfrak{S}_{2N}$-irreducible representations labeled by
two-row partitions $\mu=(2N-k,k)$ appear in the Schur--Weyl
decomposition; accordingly, components of
$\mathbb{C}[\mathfrak{S}_{2N}]^{\mathfrak{S}_{\mathrm{pair}}}$
lying in irreducible sectors with \(\ell(\mu)\ge 3\) vanish under the representation and therefore
do not contribute to $\mathcal{A}_{2N}^{\mathfrak{S}_{\mathrm{pair}}}$. We now make this distinction explicit by giving the dimension formula for $\mathcal{A}_{2N}^{\mathfrak{S}_{\mathrm{pair}}}$.

\begin{proposition}
\label{thm:dim_opr}
The dimension of $\mathcal{A}_{2N}^{\mathfrak{S}_{\mathrm{pair}}}$ is
\begin{equation}
\sum_{\lambda\vdash N}
\frac{1}{z_\lambda}
\sum_{k=0}^{N}
\chi^{(2N-k,k)}(\widetilde{\lambda})^2,
\end{equation}
where \(z_\lambda := \prod_{r\ge 1} r^{m_r(\lambda)} m_r(\lambda)!\), and $\chi^{(2N-k,k)}$ is the irreducible character of $\mathfrak S_{2N}$.
\end{proposition}

In Appendix~\ref{app:proof_dim_opr}, we prove the Proposition by expressing the dimension as the trace of the twirling map and computing it blockwise under the Schur--Weyl decomposition.
The irreducible character $\chi^{(2N-k,k)}$ can be efficiently computed using, e.g., the Murnaghan–Nakayama rule.

Importantly, we do not introduce $\mathfrak{S}_{\mathrm{pair}}$ merely as a convenient subgroup of $\mathfrak{S}_{2N}$; rather, it is dictated by the pairwise architecture of HyQuRP, where symmetry is imposed at the level of fixed two-qubit blocks. From this perspective, this subsection establishes the general form of dual-equivariant gates relevant to HyQuRP and clarifies the extent of their achievable expressivity.

\subsection{Example}
\label{subsec:exam}
We illustrate Proposition~\ref{pro:dim_group} and Proposition~\ref{thm:dim_opr} through the smallest nontrivial case, taking \(N=2\).
We consider the space $W=(\mathbb{C}^2)^{\otimes 4}$
with pair blocks $B_0=\{0,1\}, \quad B_1=\{2,3\}$.
In this case, the pair-permuting subgroup is
\begin{align}
\mathfrak{S}_{\mathrm{pair}}=\{e,\tau\}\cong \mathfrak{S}_2,
\qquad
\tau := (0\,2)(1\,3),
\end{align}
which swaps the two pairs as rigid blocks. Therefore, for each $\pi \in \mathfrak{S}_4$, the $\mathfrak{S}_{\mathrm{pair}}$-orbit average under conjugation is
\(
\frac{1}{2}\bigl(\pi + \tau \pi \tau^{-1}\bigr).
\)

We now determine the \(\mathfrak{S}_{\mathrm{pair}}\)-orbits in \(\mathfrak{S}_4\). Since \(|\mathfrak{S}_{\mathrm{pair}}|=2\), each orbit has size either
\(1\) or \(2\). By Proposition~\ref{pro:dim_group}, we have
\begin{align}
\dim \mathbb{C}[\mathfrak{S}_4]^{\mathfrak{S}_{\mathrm{pair}}}
=
\sum_{\lambda\vdash 2}
\prod_{r\ge 1}
r^{m_r(\lambda)}
\frac{(2m_r(\lambda))!}{m_r(\lambda)!}.
\end{align}
For the two partitions of \(2\),
\begin{align}
\lambda=(1,1):\quad
1^2\frac{(2\cdot 2)!}{2!}=12,
\end{align}
and
\begin{align}
\lambda=(2):\quad
2^1\frac{(2\cdot 1)!}{1!}=4,
\end{align}
and therefore
\begin{equation}
\dim \mathbb{C}[\mathfrak{S}_4]^{\mathfrak{S}_{\mathrm{pair}}}=12+4=16.
\end{equation}
The \(24\) elements of \(\mathfrak{S}_4\) decompose into \(16\) \(\mathfrak{S}_{\mathrm{pair}}\)-orbits, which give rise to the corresponding basis elements listed in Table~\ref{tab:4qubit_spair_orbits}.
\begin{table}[htbp]
\centering
\small
\setlength{\tabcolsep}{2pt}
\renewcommand{\arraystretch}{1.3}
\resizebox{\columnwidth}{!}{%
\begin{tabular}{@{}c|c|c|c@{}}
\hline
\multicolumn{2}{c|}{Orbit size $1$} & \multicolumn{2}{c}{Orbit size $2$} \\
\hline
Orbit & Basis & Orbit & Basis \\
\hline
$\{e\}$ & $e$
& $\{(0\,1),(2\,3)\}$ & $\frac12\bigl((0\,1)+(2\,3)\bigr)$ \\

$\{(0\,2)\}$ & $(0\,2)$
& $\{(0\,3),(1\,2)\}$ & $\frac12\bigl((0\,3)+(1\,2)\bigr)$ \\

$\{(1\,3)\}$ & $(1\,3)$
& $\{(0\,1\,2),(0\,2\,3)\}$ & $\frac12\bigl((0\,1\,2)+(0\,2\,3)\bigr)$ \\

$\{(0\,1)(2\,3)\}$ & $(0\,1)(2\,3)$
& $\{(0\,2\,1),(0\,3\,2)\}$ & $\frac12\bigl((0\,2\,1)+(0\,3\,2)\bigr)$ \\

$\{(0\,2)(1\,3)\}$ & $(0\,2)(1\,3)$
& $\{(1\,2\,3),(0\,1\,3)\}$ & $\frac12\bigl((1\,2\,3)+(0\,1\,3)\bigr)$ \\

$\{(0\,3)(1\,2)\}$ & $(0\,3)(1\,2)$
& $\{(1\,3\,2),(0\,3\,1)\}$ & $\frac12\bigl((1\,3\,2)+(0\,3\,1)\bigr)$ \\

$\{(0\,1\,2\,3)\}$ & $(0\,1\,2\,3)$
& $\{(0\,1\,3\,2),(0\,2\,3\,1)\}$ & $\frac12\bigl((0\,1\,3\,2)+(0\,2\,3\,1)\bigr)$ \\

$\{(0\,3\,2\,1)\}$ & $(0\,3\,2\,1)$
& $\{(0\,2\,1\,3),(0\,3\,1\,2)\}$ & $\frac12\bigl((0\,2\,1\,3)+(0\,3\,1\,2)\bigr)$ \\
\hline
\end{tabular}
}
\caption{\textbf{Basis elements of $\mathbb{C}[\mathfrak{S}_4]^{\mathfrak{S}_{\mathrm{pair}}}$.}
The $16$ orbit-averaged basis elements are grouped by orbit size.}
\label{tab:4qubit_spair_orbits}
\end{table}

We now consider the images of the basis elements in
\(\mathbb{C}[\mathfrak{S}_4]^{\mathfrak{S}_{\mathrm{pair}}}\)
under the operator-algebra map
\begin{equation}
\mathcal{A}_4^{\mathfrak{S}_{\mathrm{pair}}}
\;=\;
\Pi\!\bigl(\mathbb{C}[\mathfrak{S}_4]^{\mathfrak{S}_{\mathrm{pair}}}\bigr)
\subseteq \mathrm{End}\!\bigl((\mathbb{C}^2)^{\otimes 4}\bigr).
\end{equation}
The dimension decreases from the group algebra level to the operator level. To make this concrete, let
\begin{equation}
X_{\mathcal O}
\;:=\;
\frac{1}{2}\bigl(\pi + \tau \pi \tau^{-1}\bigr),
\end{equation}
for each \(\mathfrak{S}_{\mathrm{pair}}\)-orbit \(\mathcal O\subset \mathfrak{S}_4\), and write
\begin{equation}
\widehat X_{\mathcal O}
:=
\Pi(X_{\mathcal O})
=
\mathcal{T}_{\mathfrak{S}_{\mathrm{pair}}}\!\bigl[\Pi(\pi)\bigr],
\qquad
\pi\in \mathcal O.
\end{equation}
We can then compute, for instance,
\begin{equation}
\label{eq:4qubit-linear-dependence-example}
\begin{aligned}
\widehat X_{\{(0\,3\,2\,1)\}}
&=
-\widehat X_{\{e\}}
+\widehat X_{\{(0\,2)\}}
+\widehat X_{\{(1\,3)\}}
+\widehat X_{\{(0\,1)(2\,3)\}} \\
&\quad
-\widehat X_{\{(0\,2)(1\,3)\}}
+\widehat X_{\{(0\,3)(1\,2)\}}
-\widehat X_{\{(0\,1\,2\,3)\}}
\end{aligned}
\end{equation}
which illustrates that linear independence at the level of the group algebra does not necessarily persist after passing to the operator level.
By Proposition~\ref{thm:dim_opr}, we have
\begin{equation}
\dim\!\bigl(\mathcal{A}_{4}^{\mathfrak{S}_{\mathrm{pair}}}\bigr)
=
\sum_{\lambda \vdash 2}
\frac{1}{z_\lambda}
\sum_{k=0}^{2}
\chi^{(4-k,k)}(\widetilde{\lambda})^2.
\end{equation}
For the two partitions of \(2\), we have
\begin{equation}
z_{(1,1)} = 1^2\,2! = 2,
\qquad
z_{(2)} = 2^1\,1! = 2.
\end{equation}
Moreover,
\begin{equation}
\widetilde{(1,1)}
=
1^{2\cdot 2}
=
1^4,
\qquad
\widetilde{(2)}
=
2^{2\cdot 1}
=
2^2.
\end{equation}
Since only the two-row partitions \((4)\), \((3,1)\), and \((2,2)\) occur, we obtain
\begin{equation}
\dim\!\bigl(\mathcal{A}_4^{\mathfrak{S}_{\mathrm{pair}}}\bigr)
=
\frac{1}{2}
\sum_{k=0}^{2}\chi^{(4-k,k)}(1^4)^2
+
\frac{1}{2}
\sum_{k=0}^{2}\chi^{(4-k,k)}(2^2)^2.
\end{equation}
Using the character values
\begin{equation}
\chi^{(4)}(1^4)=1,\qquad
\chi^{(3,1)}(1^4)=3,\qquad
\chi^{(2,2)}(1^4)=2,
\end{equation}
and
\begin{equation}
\chi^{(4)}(2^2)=1,\qquad
\chi^{(3,1)}(2^2)=-1,\qquad
\chi^{(2,2)}(2^2)=2,
\end{equation}
we find
\begin{equation}
\dim\!\bigl(\mathcal{A}_4^{\mathfrak{S}_{\mathrm{pair}}}\bigr)
=
\frac12(1^2+3^2+2^2)
+
\frac12(1^2+(-1)^2+2^2)
=
10.
\end{equation}
Hence the \(16\) orbit averaged basis elements collapse to a \(10\)-dimensional operator family, as shown in Table~\ref{tab:4qubit_spair_operator_basis}. The theoretical framework established in this section forms the basis for the HyQuRP architecture presented in the next section.

We also compute the dimension of the invariant algebra and its representation in Appendix~\ref{app:extension_to_arbitrary_block} when each block has size $b$ instead of the size $2$, as discussed in this section. We show that such an extension may provide better expressivity when $N$ is large enough.

\begin{table}[t]
\centering
\small
\renewcommand{\arraystretch}{1.3}
\setlength{\tabcolsep}{5pt}
\resizebox{\columnwidth}{!}{%
\begin{tabular}{@{}c|c|c|c@{}}
\hline
Basis op. &
Expression &
Basis op. &
Expression \\ \hline

\(\widehat X_1\)  & \(\Pi(e)\) 
& \(\widehat X_6\)  & \(\Pi((0\,3)(1\,2))\) \\

\(\widehat X_2\)  & \(\Pi((0\,2))\) 
& \(\widehat X_7\)  & \(\Pi((0\,1\,2\,3))\) \\

\(\widehat X_3\)  & \(\Pi((1\,3))\) 
& \(\widehat X_8\)  & \(\Pi\!\left(\tfrac12\bigl((0\,1)+(2\,3)\bigr)\right)\) \\

\(\widehat X_4\)  & \(\Pi((0\,1)(2\,3))\) 
& \(\widehat X_9\)  & \(\Pi\!\left(\tfrac12\bigl((0\,3)+(1\,2)\bigr)\right)\) \\

\(\widehat X_5\)  & \(\Pi((0\,2)(1\,3))\) 
& \(\widehat X_{10}\) & \(\Pi\!\left(\tfrac12\bigl((0\,1\,2)+(0\,2\,3)\bigr)\right)\) \\ \hline
\end{tabular}
}
\caption{\textbf{Basis elements of \(\mathcal{A}_4^{\mathfrak{S}_{\mathrm{pair}}}\).}
Explicit expressions for a basis of \(\mathcal{A}_4^{\mathfrak{S}_{\mathrm{pair}}}\).}
\label{tab:4qubit_spair_operator_basis}
\end{table}

\section{Method: HyQuRP}
\label{sec:Method}

\begin{figure*}[t]
\centering
\includegraphics[width=1.0\linewidth]{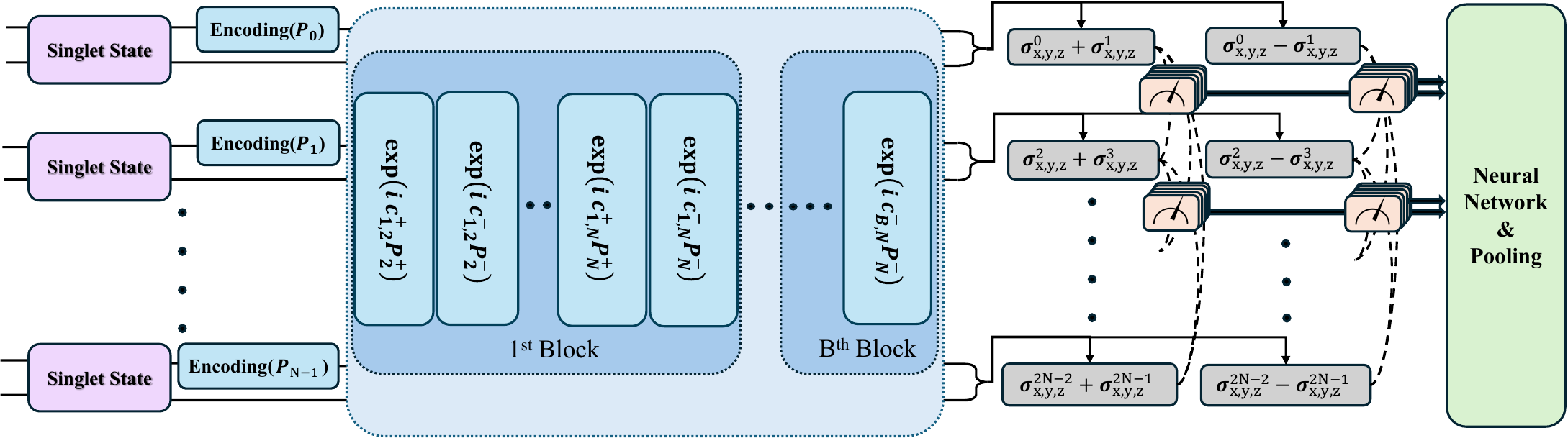}
\caption[HyQuRP Pipeline]{
\textbf{HyQuRP Pipeline.}
A hybrid pipeline maps a set of $N$ 3D points to a $K\times1$ score vector, where $K$ is the number of classes.
The quantum circuit is constructed as follows:
(i) Initialize each qubit pair in the singlet state.
(ii) Encode points on the even-indexed qubits (0, 2, 4, \ldots) while preserving the pair structure.
(iii) Apply twirling over cycles of length $2$ through $N$, under appropriate constraints, resulting in two effective gate types for each cycle length; we then repeat this block $B$ times to form a quantum network.
(iv) For every unordered pair $\langle i,j \rangle$, measure two types of Hamiltonians ($H^+_{\langle i, j\rangle}$ and $H^-_{\langle i, j\rangle}$), giving a total $2\tbinom{N}{2}$ expectation values.
(v) Send the features to a classical head (pre-pooling neural layers, permutation-invariant pooling, post-pooling classifier) to produce the score vector.
Each color indicates how each component behaves under group actions on the input point set: (Purple) rotation- and permutation-invariant, (Blue) rotation- and permutation-equivariant, (Gray) rotation-invariant and permutation-equivariant, (Green) permutation-invariant readout and classifier. Dashed rounded boxes indicate repetition by $B$.
}
\label{fig:hyqurp-overview}
\end{figure*}

Our framework is designed for 3D geometric data with underlying rotational and permutational symmetries. Specifically, we apply it to 3D point cloud classification. Given an input point-set (a row of a point cloud dataset) $\mathsf{P} = \{\mathbf{p}_{i} \in \mathbb{R}^3\}_{i=0}^{N-1}$, our method produces invariant class predictions via the HyQuRP network.
Here, $N$ denotes the number of points, and we set the quantum register size to $n = 2N$ qubits (two qubits per point).

As shown in Fig.~\ref{fig:hyqurp-overview}, we decompose HyQuRP into five stages: 
(i) singlet state initialization, 
(ii) selective geometric quantum encoding,
(iii) a quantum network with rotation and permutation equivariance,
(iv) Hamiltonian expectation value measurement, and
(v) a classical network.
At a high level, stages~(i)–(iv) form the quantum backbone and stage~(v) is a classical head; together they are mathematically constrained to respect the relevant group actions, so that the overall architecture implements a classifier that is invariant to global rigid rotations and permutations of the $N$ input points.
We describe this construction in the following subsections and provide all additional details and proofs for each component in Appendix~\ref{app:quantum-network}.

\subsection{Singlet state initialization}

The quantum processing initiates with the preparation of the $n$-qubit register in the generalized singlet configuration, implementing the tensor product state
\begin{equation}
\ket{\psi_{\text{initial}}} = \bigotimes_{i=0}^{N-1} \frac{1}{\sqrt{2}}(\ket{01}_{2i,2i+1} - \ket{10}_{2i,2i+1}),
\end{equation}
where each adjacent qubit pair is the entangled Bell singlet state. This initial state can be easily prepared on a quantum computer, and is $\SU(2)$ invariant as $U\otimes U (\ket{01} - \ket{10}) = \ket{01} - \ket{10}$ for any $U \in \SU(2)$.

\subsection{Selective geometric quantum encoding}

We encode each 3D point $\mathbf{p}\in\mathbb{R}^3$ using a unitary operator $E(\mathbf{p})$ defined as
\begin{equation}
E(\mathbf{p}) = \exp\left(i\,\frac{\mathbf{p}\cdot \vec{\sigma}}{\Theta} \right), \label{eq:encoding_gate}
\end{equation}
where $\vec{\sigma}=\{X,Y,Z\}$ is a vector of Pauli operators and $\Theta$ is a fixed hyperparameter introduced to improve numerical stability.
By construction, $E(\mathbf{p})$ depends on the vector $\mathbf{p}$ and is equivariant to 3D rotations: For any $R\in \mathrm{SO}(3)$ with corresponding $U_R\in \mathrm{SU}(2)$,
\begin{equation}
E(R\mathbf{p}) \;=\; U_R\,E(\mathbf{p})\,U_R^\dagger,
\end{equation}
from the standard $\mathrm{SU}(2)$–$\mathrm{SO}(3)$ correspondence. 

Our $n$-qubit singlet initialization is only invariant under pair-preserving permutations (not arbitrary permutations of the $n$ individual qubits).
Accordingly, we adopt a per-pair selective encoding where only the even-indexed qubits in each pair $(2j,2j+1)$ are geometrically encoded.
This encoding preserves the pairwise structure and makes the overall encoding pairwise permutation-equivariant. 

As shown in Observation~\ref{obv:no_joint_fixed_vectors}, any nontrivial pure state is not simultaneously invariant under global $\mathrm{SU}(2)$ rotations and all permutations of the qubits. Therefore, it is impossible to use all qubits for encoding while preserving both symmetries.
Consequently, our selective per-pair scheme is a minimal and effective design.

\subsection{Quantum network with rotation and permutation equivariance}
\label{sec:our_nwk}

Proposition~\ref{thm:dim_opr} indicates that a large number of independent generators are available. However, classical simulation cannot fully capture this expressivity, since both this dimension and the size of the corresponding matrix representation grow rapidly with the number of qubits. In this work, we restrict attention to a particular family determined by the sign structure within each pair. Specifically, since \(\mathfrak{S}_{\mathrm{pair}}\) never mixes the two wires inside a pair, we are free to choose the relative sign on the non-encoding partner wire in each pair.
Accordingly, we distinguish two cases, positive and negative, and perform the twirl separately for each sign to obtain the generators of the dual-equivariant gates.

For each $k\in\{2,\dots,N\}$ and the set of pairs
\begin{equation}
\mathcal{P}=\{(0,1),(2,3),\ldots,(n-2,n-1)\},
\end{equation}
let $\mathrm{Perm}(\mathcal{P},k)$ be the set of ordered $k$–tuples $(p_{1},\dots,p_{k})$ of distinct elements, where each $p_{m}=(2j_m, 2j_m+1) \in \mathcal{P}$. Given a $k$-tuple of pairs $\pi\in\mathrm{Perm}(\mathcal{P},k)$ and a selection vector $\mathbf{s}=(s_1,\dots,s_k)\in\{0,1\}^k$, we define $\tau_{\pi}^{\mathbf{s}}$ to
be the permutation matrix that acts as the $k$–cycle
\begin{equation}
(2j_1 + s_1,\ 2j_2 + s_2,\ \ldots,\ 2j_k + s_k),
\end{equation}
on the selected wires and leaves all other wires fixed. 

We define
\begin{align}
P_k^+
&:= \frac{1}{k!}
    \sum_{\pi \in \mathrm{Perm}(\mathcal{P}, k)}
    \sum_{\mathbf{s} \in \{0,1\}^k} \tau_{\pi}^{\mathbf{s}},
    \label{eq:Pk-plus-def-main} \\[4pt]
P_k^-
&:= \frac{1}{k!}
    \sum_{\pi \in \mathrm{Perm}(\mathcal{P}, k)}
    \sum_{\mathbf{s} \in \{0,1\}^k}
    (-1)^{\|\mathbf{s}\|_1}\,\tau_{\pi}^{\mathbf{s}},
    \label{eq:Pk-minus-def-main}
\end{align}
where $\|\mathbf{s}\|_1 = \sum_{\ell=1}^k s_\ell$ is the Hamming weight of $\mathbf{s}$.

The generators $P^{\pm}_k$ are special form of generators that are equivariant under pair–permuting action of $\Sym_{\mathrm{pair}}$ discussed in the previous section.
Still, we further added symmetric or antisymmetric property under the internal permutation, i.e., swap the qubits within each pair.
This is because we believe adding such constraints makes the overall circuit better trainable while maintaining the performance of the model (see Sec.~\ref{sec:hyqurp_ablation}).

Our dual-equivariant gate combines these matrices over all $k$ values:
\begin{equation}
G^l  = \prod_{k=2}^{N} e^{ic_{l,k}^+ P_{k}^+} e^{ic_{l,k}^- P_{k}^-}
,\end{equation}
and our quantum circuit is just a stack of $B$ blocks $G^\ell$, given as
\begin{equation}
C
=\prod_{l=1}^{B}\ G^l
,\end{equation}
where $c_{l,k}^{\pm}$ are the trainable parameters for block $\ell$ and cycle length $k$.

Each operator $P_k^\pm$ is a sum of $2^k \frac{N!}{(N-k)!}$ terms of the form $\tau_{\pi}^{\mathbf{s}}$, and this number grows rapidly with $k$.
In practice, it is sufficient to retain only a subset of the available orders.
When the largest retained order is $k_{\max} =O(1)$, one can implement each block $G^l$ using, e.g., the linear combination of unitaries (LCU)~\cite{childs2012hamiltonian}.
In this case, the gate count is upper bounded by $\widetilde{O}(N^{k_{\max}+1})$, so the
circuit cost remains polynomial in $N$ on a $2N$-qubit data register.
Further details on the construction of $P_k^\pm$ and the associated
gate complexity are given in Appendix~\ref{app:quantum_network_via_twirling}.

This construction yields a gate that is equivariant with respect to both the global $\mathrm{SU}(2)$ rotation and the pair-preserving permutation actions, allowing pointwise representations under these groups to flow consistently between adjacent components and thereby acting as a symmetry-preserving bridge within the network.

\subsection{Hamiltonian expectation value measurement}

The final quantum processing stage computes expectation values using pairwise
Heisenberg Hamiltonians:
\begin{align}
H_{\langle i,j\rangle}^{+} &= \sum_{\alpha \in \{X,Y,Z\}} (\sigma_{\alpha}^{2i} + \sigma_{\alpha}^{2i+1}) (\sigma_{\alpha}^{2j} + \sigma_{\alpha}^{2j+1})
, \\
H_{\langle i,j\rangle}^{-} &= \sum_{\alpha \in \{X,Y,Z\}}(\sigma_{\alpha}^{2i} - \sigma_{\alpha}^{2i+1}) (\sigma_{\alpha}^{2j} - \sigma_{\alpha}^{2j+1})
,\end{align}
where $\sigma_{\alpha}^{q}$ denotes the Pauli-$\alpha$ operator acting on the $q$-th qubit. 
Because both the initialization and encoding operate on fixed singlet pairs
$(2i,2i{+}1)$, permutation-equivariant operators must also act at the pair level, which consist of factors in the form of $(\sigma_{\alpha}^{2i}\!\pm\!\sigma_{\alpha}^{2i+1})$.

Moreover, we want $\mathrm{SU}(2)$-invariant Hamiltonians at the operator level, in the sense that applying any global $\mathrm{SU}(2)$ rotation acts on $H$ by conjugation and leaves it unchanged.
Thus, we have $H^{\pm}_{\langle i,j \rangle}$, which are of Heisenberg form, as the lowest-order Hamiltonians that are invariant under rigid rotations of the input.
Since $i$ and $j$ run over all distinct sites, and we have the sign choice ($\pm$), the total number of measurement outcomes is $2\binom{N}{2}$.

However, the measurement outcomes exhibit permutational equivariance rather than invariance---when points are permuted, the corresponding expectation values undergo the same permutation. This equivariance property requires subsequent permutation-invariant classical processing steps to convert the quantum measurements into rotation- and permutation-invariant class predictions.

\subsection{Classical network}
After the quantum stages, we process two types of Hamiltonian expectation values $\mathbf{x} \in \mathbb{R}^{\binom{N}{2} \times 2}$ with a classical head, which is invariant under permutations.
Our classical head, which we refer to as Set-MLP, has a structure similar to DeepSets~\cite{zaheer2017deep} and consists of two MLP blocks separated by a symmetric aggregation function. 
The first MLP block is applied to the last dimension, resulting in $\mathbf{y} \in \mathbb{R}^{\binom{N}{2}\times d}$.
We then apply symmetric aggregation functions to $\mathbf{y}$ over the Hamiltonian index:
\begin{gather}
\mathbf{y}_{\mathrm{agg}} = \operatorname{Concat}_{f \in \mathcal{F}} f(\mathbf{y}), \\
\mathcal{F} = \{ \text{mean}, \text{max}, \text{min}, \text{sum}, \text{var}, \text{std} \},
\end{gather} yielding $\mathbf{y}_{\mathrm{agg}} \in \mathbb{R}^{6 d}$. The aggregated representation is then passed through the second MLP block to output class logits, trained with cross-entropy.

\subsection{Group equivariance of overall architecture}

Combining all components, our complete quantum-classical framework achieves both rotational and permutational invariance through distinct mechanisms operating at different stages of the pipeline.

Our framework achieves rotation invariance through the quantum processing pipeline. When a 3D rotation $R \in \mathrm{SO}(3)$ is applied to the input point cloud, let $U_{R} \in \mathrm{SU}(2)$ denote the corresponding rotation on the
$n$-qubit Hilbert space.
Then, we first see that the encoding layer $\mathcal{E} = \bigotimes_{i=0}^{N-1} E(\mathbf{p}_i)\otimes I_2$ transforms as
\begin{align}
\mathcal{E} \mapsto \mathcal{E}' &= \bigotimes_{i=0}^{N-1} E(R\mathbf{p}_i)\otimes I_2 \\
&= \bigotimes_{i=0}^{N-1} U_R E(\mathbf{p}_i) U_R^\dagger\otimes I_2 = U_R^{\otimes n} \mathcal{E}  (U_R^\dagger)^{\otimes n}
\end{align}

\begin{table*}[!t]
\centering
\setlength{\tabcolsep}{4.6pt}
\renewcommand{\arraystretch}{0.45}
\def\resultcol#1{\makebox[2.05cm][c]{#1}}

{\scriptsize
\begin{tabular*}{0.93\textwidth}{@{\extracolsep{\fill}}lcccccc|c@{}}
\toprule
\multirow{3}{*}{Model}
& \multicolumn{6}{c|}{ModelNet}
& \resultcol{All Data} \\
\cmidrule(lr){2-7} \cmidrule(lr){8-8}
& \multicolumn{3}{c}{Light}
& \multicolumn{3}{c|}{Mid}
& \multirow{2}{*}{\resultcol{Avg. Rank}} \\
\cmidrule(lr){2-4} \cmidrule(lr){5-7}
& 4 & 5 & 6 & 4 & 5 & 6 & \\
\midrule
HyQuRP
& 71.19 $\pm$ 1.45 & 74.99 $\pm$ 2.18 & 76.13 $\pm$ 1.96
& 72.39 $\pm$ 1.02 & 74.67 $\pm$ 1.79 & 79.54 $\pm$ 0.75
& \resultcol{1.17} \\
TFN
& 65.29 $\pm$ 2.01 & 68.76 $\pm$ 1.56 & 72.54 $\pm$ 1.69
& 65.56 $\pm$ 1.49 & 70.21 $\pm$ 1.66 & 73.49 $\pm$ 0.74
& \resultcol{3.08} \\
PointNet
& 64.00 $\pm$ 2.86 & 67.66 $\pm$ 4.96 & 71.09 $\pm$ 1.73
& 68.53 $\pm$ 3.30 & 69.00 $\pm$ 3.19 & 73.24 $\pm$ 4.36
& \resultcol{3.58} \\
PointMamba
& 63.74 $\pm$ 2.03 & 68.34 $\pm$ 2.23 & 71.03 $\pm$ 1.36
& 66.33 $\pm$ 3.26 & 72.47 $\pm$ 1.90 & 72.40 $\pm$ 1.91
& \resultcol{3.58} \\
VN-PointNet
& 65.09 $\pm$ 0.99 & 67.60 $\pm$ 2.28 & 71.17 $\pm$ 1.88
& 65.63 $\pm$ 0.67 & 71.41 $\pm$ 1.74 & 72.93 $\pm$ 2.21
& \resultcol{3.67} \\
Point TF
& 60.49 $\pm$ 3.95 & 65.14 $\pm$ 3.46 & 71.03 $\pm$ 1.76
& 62.66 $\pm$ 0.48 & 66.03 $\pm$ 1.99 & 71.40 $\pm$ 0.57
& \resultcol{6.67} \\
Mamba3D
& 61.91 $\pm$ 2.18 & 61.83 $\pm$ 1.83 & 65.14 $\pm$ 2.63
& 60.24 $\pm$ 3.68 & 60.94 $\pm$ 1.34 & 63.97 $\pm$ 4.38
& \resultcol{7.25} \\
RP-EQGNN
& 58.60 $\pm$ 3.48 & 62.60 $\pm$ 2.81 & 65.23 $\pm$ 2.40
& 59.90 $\pm$ 2.55 & 61.84 $\pm$ 1.99 & 66.43 $\pm$ 2.24
& \resultcol{8.50} \\
DGCNN
& 56.07 $\pm$ 1.66 & 58.54 $\pm$ 1.43 & 62.29 $\pm$ 1.67
& 58.49 $\pm$ 1.22 & 62.51 $\pm$ 1.73 & 69.37 $\pm$ 1.58
& \resultcol{8.75} \\
Set-MLP
& 48.76 $\pm$ 3.38 & 49.81 $\pm$ 4.51 & 52.89 $\pm$ 5.31
& 47.93 $\pm$ 2.61 & 50.76 $\pm$ 3.34 & 55.79 $\pm$ 4.36
& \resultcol{10.00} \\
MLP
& 58.56 $\pm$ 1.99 & 55.27 $\pm$ 1.57 & 56.09 $\pm$ 5.53
& 59.03 $\pm$ 1.01 & 59.66 $\pm$ 1.11 & 62.56 $\pm$ 2.00
& \resultcol{10.25} \\
PointMLP
& 53.11 $\pm$ 1.12 & 55.26 $\pm$ 2.32 & 52.71 $\pm$ 2.25
& 57.81 $\pm$ 1.16 & 59.77 $\pm$ 1.61 & 57.19 $\pm$ 1.26
& \resultcol{11.50} \\

\midrule
\multirow{3}{*}{Model}
& \multicolumn{6}{c|}{ShapeNet}
& \resultcol{All Data} \\
\cmidrule(lr){2-7} \cmidrule(lr){8-8}
& \multicolumn{3}{c}{Light}
& \multicolumn{3}{c|}{Mid}
& \multirow{2}{*}{\resultcol{Avg. Acc}} \\
\cmidrule(lr){2-4} \cmidrule(lr){5-7}
& 4 & 5 & 6 & 4 & 5 & 6 & \\
\midrule
HyQuRP
& 71.90 $\pm$ 2.12 & 71.53 $\pm$ 2.42 & 76.89 $\pm$ 1.04
& 74.37 $\pm$ 1.36 & 75.06 $\pm$ 1.20 & 76.80 $\pm$ 2.00
& \resultcol{74.62 $\pm$ 2.90} \\
TFN
& 63.89 $\pm$ 1.20 & 66.94 $\pm$ 1.83 & 74.96 $\pm$ 1.33
& 64.34 $\pm$ 1.14 & 69.53 $\pm$ 1.57 & 77.51 $\pm$ 1.10
& \resultcol{69.42 $\pm$ 4.51} \\
PointNet
& 67.29 $\pm$ 1.42 & 66.01 $\pm$ 3.33 & 72.81 $\pm$ 2.57
& 69.43 $\pm$ 5.60 & 69.57 $\pm$ 3.24 & 73.07 $\pm$ 4.01
& \resultcol{69.31 $\pm$ 4.34} \\
PointMamba
& 67.00 $\pm$ 2.91 & 68.43 $\pm$ 2.31 & 70.03 $\pm$ 2.55
& 70.20 $\pm$ 3.61 & 69.64 $\pm$ 1.74 & 72.30 $\pm$ 2.45
& \resultcol{69.33 $\pm$ 3.45} \\
VN-PointNet
& 62.86 $\pm$ 1.43 & 64.81 $\pm$ 5.15 & 74.83 $\pm$ 1.30
& 64.73 $\pm$ 1.20 & 70.71 $\pm$ 1.50 & 76.89 $\pm$ 0.55
& \resultcol{69.05 $\pm$ 4.79} \\
Point TF
& 59.90 $\pm$ 0.72 & 61.41 $\pm$ 2.63 & 69.83 $\pm$ 2.42
& 61.01 $\pm$ 1.83 & 59.13 $\pm$ 5.60 & 71.41 $\pm$ 1.38
& \resultcol{64.95 $\pm$ 5.28} \\
Mamba3D
& 62.10 $\pm$ 1.43 & 64.97 $\pm$ 2.26 & 67.49 $\pm$ 2.12
& 61.99 $\pm$ 2.36 & 65.14 $\pm$ 1.59 & 67.70 $\pm$ 3.22
& \resultcol{63.62 $\pm$ 3.39} \\
RP-EQGNN
& 59.70 $\pm$ 2.28 & 60.14 $\pm$ 2.95 & 60.47 $\pm$ 2.28
& 57.34 $\pm$ 2.70 & 58.86 $\pm$ 2.66 & 61.76 $\pm$ 1.32
& \resultcol{61.07 $\pm$ 3.49} \\
DGCNN
& 57.54 $\pm$ 1.25 & 56.56 $\pm$ 2.45 & 67.13 $\pm$ 2.76
& 58.41 $\pm$ 1.70 & 60.37 $\pm$ 1.47 & 70.71 $\pm$ 2.07
& \resultcol{61.50 $\pm$ 5.13} \\
Set-MLP
& 59.67 $\pm$ 8.06 & 60.09 $\pm$ 8.62 & 57.09 $\pm$ 1.70
& 61.04 $\pm$ 9.59 & 66.54 $\pm$ 11.82 & 69.86 $\pm$ 11.38
& \resultcol{56.68 $\pm$ 9.44} \\
MLP
& 57.11 $\pm$ 1.68 & 51.67 $\pm$ 1.97 & 51.89 $\pm$ 3.02
& 60.53 $\pm$ 2.45 & 56.30 $\pm$ 2.29 & 57.46 $\pm$ 2.74
& \resultcol{57.18 $\pm$ 3.92} \\
PointMLP
& 52.17 $\pm$ 2.36 & 47.33 $\pm$ 2.46 & 49.09 $\pm$ 2.78
& 52.94 $\pm$ 2.62 & 51.51 $\pm$ 2.80 & 53.43 $\pm$ 2.13
& \resultcol{53.53 $\pm$ 3.97} \\
\bottomrule
\end{tabular*}
}

\caption[Overall Results]{\textbf{Overall Results.} Entries for ModelNet and ShapeNet report mean accuracy $\pm$ standard deviation over 7 seeds under the Light and Mid settings with 4, 5, and 6 points sampled per object. 
The Avg. Rank reports the average rank, where lower is better, and the Avg. Acc reports the mean accuracy $\pm$ standard deviation; both are aggregated across both datasets and all settings. Details are provided in Appendix~\ref{app:summary_metrics}.}
\label{tab:overall_results}
\end{table*}

Accordingly, the transformation of the expectation value is given as:
\begin{align}
&\langle\psi_{\mathrm{output}}|H|\psi_{\mathrm{output}}\rangle_{\mathrm{rotation}} \nonumber \\
&= \langle\psi_0|\mathcal{E}'^\dagger C^\dagger H C \mathcal{E}' |\psi_0\rangle \nonumber \\
&= \langle\psi_0|U_R^{\otimes 2N} \mathcal{E}^\dagger C^\dagger (U_R^\dagger)^{\otimes 2N} H U_R^{\otimes 2N} C \mathcal{E} (U_R^\dagger)^{\otimes 2N} |\psi_0\rangle \nonumber \\
&= \langle\psi_0|\mathcal{E}^\dagger C^\dagger H C \mathcal{E}|\psi_0\rangle \nonumber \\
&= \langle\psi_{\mathrm{output}}|H|\psi_{\text{output}}\rangle_{\mathrm{origin}}
,\end{align}
where $C$ is our $\SU(2)$-equivariant quantum circuit that satisfies $U^{\otimes 2N}C=CU^{\otimes 2N}$ for all single-qubit unitary operators $U \in \SU(2)$ and $|\psi_0\rangle$ denotes a series of $N$ singlet states.
The singlet state is inherently rotation-invariant, remaining unchanged under $\SU(2)$ transformations. 
The equivariant gates naturally commute with rotations, allowing the rotation operators $U_R$ and $U_R^\dagger$ to propagate through to the Hamiltonian. Since our Heisenberg Hamiltonians are $\SU(2)$-invariant, the $\SU(2)$ operators cancel out, demonstrating that the quantum expectation values remain unchanged under arbitrary 3D rotations of the input point cloud.

For permutation invariance, our framework operates through a two-stage process.
The quantum processing stage is permutation-equivariant end-to-end, so it propagates symmetric, pointwise representations up to the classical head.
When the input points undergo an index permutation $\Pi$, the Hamiltonian expectation value, computed in the
same manner as above, is given by: 
\begin{align}
&\Pi\left(\langle\psi_{\text{output}}|H_{\langle i ,j \rangle}^\pm|\psi_{\text{output}}\rangle_{\text{origin}}\right) \nonumber \\
&= \langle\psi_{\text{output}}|H_{\langle \Pi(i) ,\Pi(j) \rangle}^\pm|\psi_{\text{output}}\rangle_{\Pi}.
\end{align}
The classical neural network then passes the measurements through MLP blocks and a statistical pooling operation, which is inherently invariant to input permutations.

This two-stage approach---quantum equivariance followed by classical invariance---ensures that the final classification predictions remain unchanged regardless of the ordering of input points.

By enforcing rotation and permutation invariance at the architectural level, our framework enables more robust and data-efficient 3D point cloud classification.
This symmetry guarantee is grounded in the representation-theoretic interplay between quantum dynamics and the group structures of rotations and permutations, which we leverage to design quantum modules and classical readouts that preserve these invariants end-to-end. Consequently, we resolve a bottleneck in equivariant QML at the architectural level.

\subsection{Comparison to other QML models for 3D point clouds}
Several previous studies have proposed QML models for 3D point clouds. However, these models rely too much on classical preprocessing or do not fully guarantee the dual equivariance.
For example, a model introduced in \citet{li2024enforcing} computes the inner products between all pairs of input point vectors, generating $N(N+1)/2$ input data for the quantum network where $N$ is the number of points in each point set. 
However, this approach has the following limitations: (1) Since the inner product is invariant under the full orthogonal group $O(3)$ and thus cannot distinguish mirror-reflected point clouds, the model cannot be used for applications where such symmetry is essential, e.g., in molecular structures.
(2) The inner product encoding requires a total of $\Theta(N^2)$ qubits. 
Even for $N \approx 200$, a moderate number for a point cloud dataset, the model requires approximately 20K qubits, which is impractical. 
By contrast, HyQuRP operates directly on raw coordinates and uses only $2N$ qubits, enabled by a group-theoretic architecture. This design yields a more expressive and practical model than an inner-product-based model.

Most recently, RP-EQGNN was proposed, claiming higher performance than previous QML models on standard 3D point cloud classification benchmarks~\citep{liu2025rpeqgnn}.
According to the paper, RP-EQGNN exploits rotational and permutational symmetries by parameterizing its equivariant quantum gates using rotation-invariant scalars, such as pairwise distances, angles, and edge types.
However, we found that the design described in the paper and the released implementation are insufficient to guarantee permutation- and rotation-invariant outputs (see Appendix~\ref{app:baseline_arch}).
Moreover, RP-EQGNN primarily relies on scalar rotation-invariant features rather than representation-theoretic constructions for quantum dynamics, which may limit its expressivity and generalization compared to HyQuRP.
In particular, in our experimental setting, HyQuRP consistently outperforms an RP-EQGNN baseline.

\section{Experiments}

In this section, we evaluate point cloud classification under our experimental setting. We compare HyQuRP with various quantum and classical baseline models and conduct ablation studies to examine the contribution of key gate-design choices.
The main results are summarized in Tables~\ref{tab:overall_results} and~\ref{tab:rank}, while the ablation results are reported in Table~\ref{tab:generator_ablation}. Our code, processed data, and experimental results are available in the GitHub repository~\href{https://github.com/YonseiQC/equivariant_QML}{https://github.com/YonseiQC/equivariant\_QML}.

\subsection{Experimental setup}

We use small-class subsets of ModelNet and ShapeNet, where each subset contains 5 selected classes.
ModelNet~\citep{wu20153d} is a widely used collection of clean CAD meshes of man-made objects;
ShapeNetCore~\citep{chang2015shapenet} is a large-scale repository of 3D CAD models with a unified taxonomy and substantial intra-class geometric variation.
We enforce object-level disjointness between the test and the training/validation sets and subsample points via
farthest-point sampling (FPS) from shapes centered at the origin (zero-mean) and
scaled to unit maximum radius (i.e., $\max_i \lVert \mathbf{x}_i \rVert_2 = 1$).
We focus on a sparse-point regime with $N\in\{4,5,6\}$ to assess data efficiency, while keeping the corresponding quantum circuits classically simulable within a reasonable time.

In training, these subsampled points are augmented by applying three types of random transformations: SO(3) rotation, permutation, and jitter. Training uses Adam (batch size 35) for 1,000 epochs; for each model, its own fixed learning rate $\eta$ is used for all datasets and $N \in \{4,5,6\}$, which is optimized for the ModelNet dataset with $N=4$ over $\eta \in \{10^{-2},10^{-3},10^{-4}\}$. 
All results are reported as mean~$\pm$~std of test top-1 accuracy over 7 seeds; for each seed, the checkpoint with the highest validation accuracy is evaluated on the test set. Details of the dataset design and construction are provided in Appendix~\ref{app:c1}; implementation and experimental environments are in Appendix~\ref{app:c2}.

\subsection{Baselines}
Our goal in the sparse-point regime is to compare architectural designs under matched capacity.
Canonical point cloud classification models are typically instantiated with 1024--2048 input points and tens of output classes (e.g., ModelNet40), resulting in multi-million-parameter models.
In contrast, our setting operates on only a small number of points per object and a small number of classes, so directly reusing those standard configurations would introduce far more parameters per point than the data can support, obscuring the effect of the architectural design rather than the raw capacity.
We therefore define two parameter budgets, Light ($\sim$1.5K parameters) and Mid ($\sim$7K), and instantiate both our HyQuRP models and all baselines at these levels so that performance differences primarily reflect architectural choices rather than raw capacity.

Within this capacity-matched setting, we compare HyQuRP against the following model families:
a plain MLP, Set-MLP (our classical ablation: quantum modules removed, classical pipeline unchanged), PointNet~\citep{qi2017pointnet}, Tensor Field Network (TFN)~\cite{thomas2018tensor}, DGCNN~\citep{wang2019dynamic}, 
VN-PointNet~\cite{deng2021vector},
PointTransformer (Point TF)~\citep{zhao2021point},
PointMLP~\citep{ma2022rethinking}, PointMamba~\citep{liang2024pointmamba}, Mamba3D~\cite{han2024mamba3d} and RP-EQGNN~\citep{liu2025rpeqgnn}.
The selected baselines span a broad spectrum, ranging from simple baselines to SOTA models, and from symmetry-agnostic designs to strongly symmetry-aware ones. For each baseline, we preserve its canonical building blocks and heads while systematically reducing depth/width, with only minimal adjustments required by the sparse-point regime, until it fits the Light/Mid budgets.

We refer to Appendix~\ref{app:baseline_arch} for full architectural details of the baselines, including the original designs and our implementation differences.

\subsection{Results and analysis}
\label{subsec:res_ana}

As shown in Table~\ref{tab:overall_results}, across two datasets and both Light ($\sim$1.5K parameters) and Mid ($\sim$7K) regimes, HyQuRP achieves the best overall average rank ($1.17$) and the highest average accuracy ($74.62\%$).
In particular, on ModelNet with 6 points (Light), HyQuRP achieves $76.13\%$ test accuracy, exceeding the top-performing baselines—TFN ($72.54\%$), PointNet ($71.09\%$), and PointMamba ($71.03\%$)—by +$3.59$, +$5.04$, and +$5.10$ percentage points, respectively. A similar pattern holds on ShapeNet and in the Mid setting.
Overall, these results indicate a substantial improvement over a diverse set of strong baselines across datasets and model sizes, under our experimental protocol.

In our experimental setting, on ModelNet with 6 points (Light), Set-MLP achieves $52.89\%$ test accuracy, which is $23.24$ percentage points lower than HyQuRP (Table~\ref{tab:overall_results}). On average over all settings, HyQuRP outperforms Set-MLP by $17.94 \pm 3.03$ percentage points, even though the primary architectural addition is the quantum component with $\sim$100 parameters. Set-MLP is not rotation-invariant, whereas HyQuRP is rotation- and permutation-invariant (Table~\ref{tab:rank}); this invariance property can contribute to the observed gap. Moreover, HyQuRP's quantum blocks operate in an exponentially large Hilbert space, which may offer additional expressivity that is not reflected by the small parameter increase.

To assess whether invariance alone accounts for the gain, we further compare HyQuRP with TFN and VN-PointNet, both of which are invariant to rotations and permutations. The results in Tables~\ref{tab:overall_results} and ~\ref{tab:rank} show that HyQuRP broadly outperforms both TFN and VN-PointNet despite using fewer parameters. 
More specifically, when averaged across all settings, HyQuRP outperforms TFN and VN-PointNet by $5.20 \pm 3.33$ and $5.57 \pm 3.70$ percentage points, respectively (Table~\ref{tab:overall_results}).
This suggests that the observed gain cannot be attributed solely to the invariance property, but rather that the quantum model may offer an additional representational advantage over classical models in this regime.

\begin{table}[htbp]
\centering
\setlength{\tabcolsep}{2pt}
\renewcommand{\arraystretch}{0.98}
\begin{tabular}{lccc}
\toprule
Model & Rot Inv & Perm Inv & Params (L/M) \\
\midrule
HyQuRP            & Y  & Y  & 1453 / 7693 \\
TFN               & Y  & Y  & 1500 / 7900 \\
PointNet          & N & Y  & 1675 / 7818 \\
PointMamba        & N & Y  & 1720 / 7870 \\
VN-PointNet       & Y & Y  & 1612 / 7797 \\
Point TF & N & Y  & 1937 / 8105 \\
Mamba3D           & N & N & 2194 / 8416 \\
RP-EQGNN          & N & Y  & 1485 / 8781 \\
DGCNN             & N & Y  & 1476 / 7751 \\
Set-MLP           & N & Y  & 1365 / 7605 \\
MLP               & N & N & 1479 / 7040 \\
PointMLP          & N & Y  & 1743 / 7853 \\
\bottomrule
\end{tabular}
\caption[Model Characteristics]{\textbf{Model Characteristics.} The table reports the trainable parameter counts of the Light and Mid variants, and whether each model is rotation-invariant (Rot Inv) and permutation-invariant (Perm Inv). For detailed logit comparisons before and after rotations and permutations, see Appendix~\ref{app:logit_comparisons}}
\label{tab:rank}
\end{table}

\subsection{Ablation study} \label{sec:hyqurp_ablation}
To analyze the contribution of the components in HyQuRP’s dual-equivariant gates, we conduct ablation studies on ModelNet in the 4-point setting. Specifically, we examine the effects of cycle length and generator choice. The results are presented in Table~\ref{tab:generator_ablation}.

\begin{table}[htbp]
\centering
\begin{tabular}{@{}l@{\hspace{18pt}}l@{\hspace{18pt}}c@{}}
\toprule
Ablation factor & Setting & Accuracy (\%) \\
\midrule
\multirow{2}{*}{Cycle length}
& $2$       & 70.97 $\pm$ 3.34 \\
& $2,3$     & 71.10 $\pm$ 1.90 \\
& $2,3,4$ (HyQuRP)   & 71.19 $\pm$ 1.45 \\
\midrule
\multirow{2}{*}{Generator}
& $P_k^+$   & 63.20 $\pm$ 5.31 \\
& $P_k^-$   & 70.47 $\pm$ 1.27 \\
& $P_k^\pm$ (HyQuRP) & 71.19 $\pm$ 1.45 \\
& Full generators & 71.79 $\pm$ 3.79 \\
\bottomrule
\end{tabular}
\caption{\textbf{Ablation Study.} Ablation results on ModelNet under the Light setting with 4 input points. Results are reported as mean test accuracy $\pm$ standard deviation over 7 seeds.}
\label{tab:generator_ablation}
\end{table}

\paragraph{Cycle length.}
In the 4-point setting, HyQuRP uses all admissible cycle lengths $k \in \{2,3,4\}$.  To assess the impact of the maximum cycle length, we evaluate two reduced variants that retain only $k=2$ and $k \in \{2,3\}$. As shown in Table~\ref{tab:generator_ablation}, reducing the maximum cycle length leads to only a marginal drop in performance: $71.19\%$ at $k_{\max}=4$, $71.10\%$ at $k_{\max}=3$, and $70.97\%$ at $k_{\max}=2$.
This suggests that much of the performance of HyQuRP is already captured by short cycles. 
Importantly, a small $k_{\max}$ allows HyQuRP to retain polynomial per-block gate complexity, with a per-block two-qubit gate count upper bounded by $\widetilde{O}(N^{k_{\max}+1})$, on quantum hardware while preserving most of the performance.

\paragraph{Generator.}
HyQuRP uses the generator family $P_k^{\pm}$. To assess the effect of the generator choice, we compare HyQuRP with three variants: one using only $P_k^{+}$, one using only $P_k^{-}$, and a richer model using all 84 generators from the full-dimensional operator space in Appendix~\ref{app:extension_to_arbitrary_block}. For the same generator count, $P_k^{-}$ substantially outperforms $P_k^{+}$ by $7.27 \pm 4.66$ percentage points. This suggests that, in this setting, the antisymmetric component under the internal swap within each pair appears more informative than the symmetric one. 

Beyond this pairwise comparison, accuracy further improves as the generator family is enriched, from $70.47\%$ with $P_k^{-}$ to $71.19\%$ with HyQuRP and $71.79\%$ with the full-generator model. 
This empirical trend is consistent with the view that a richer generator family can improve the expressivity of the dual-equivariant gates, and underscores the significance of our representation-theoretic framework, which rigorously characterizes the admissible generator space and enables its principled expansion. 
Notably, the larger standard deviation of the full-generator model relative to HyQuRP suggests a potential expressivity–stability trade-off, which motivates further investigation.

\section{Conclusion and Discussion} \label{sec:conclusion}

In this work, we established the theoretical framework for dual-equivariant quantum circuits under rotational and permutational symmetries, resolving a bottleneck in equivariant QML. We formulated a general class of dual-equivariant gates and charted its expressive landscape through dimension formulas. Based on this, we designed HyQuRP, a quantum-classical hybrid machine learning model with rotation and permutation equivariance to enable group-invariant classification.

Across standard 3D point cloud benchmarks with a small number of classes and points, HyQuRP generally outperforms state-of-the-art baselines in our experimental setting (Tables~\ref{tab:overall_results} and ~\ref{tab:rank}). Notably, it also outperforms VN-PointNet and TFN, despite the fact that they are built to respect the same symmetries.
Taken together, these gains are not fully accounted for by symmetry alone, suggesting the potential benefit of quantum-enhanced representations in this regime.

Our study focuses on a sparse-point regime to assess data efficiency using classical simulation. Since classical simulation does not scale to large qubit counts, our present evaluation does not yet extend to dense point clouds. However, this 3D point cloud testbed should not be conflated with the scope of the framework itself. The dual-equivariant framework and HyQuRP may also extend to a broader range of 3D geometric problems, including molecular structures and crystalline materials, among others.

Our experiments conducted here avoided commonly used techniques such as data re-uploading, resampling, and neighborhood grouping, as they may break the mathematical symmetries.
Moreover, for simplicity, we only utilized a limited set of quantum gates by imposing additional mathematical constraints, leaving room for further performance improvements. We hope our work offers a new perspective on quantum machine learning and encourages further research in this direction.

\section*{Acknowledgements}
This research was supported (in part) by the Yonsei University Research Fund of 2025-12-0256.
This research was also supported (in part) by the BK21FOUR (Fostering Outstanding Universities for Research), funded by the Ministry of Education (MOE) of Korea and the National Research Foundation (NRF) of Korea; and the NRF grants funded by the Korean government (MSIT) (No. RS-2025-16066935, RS-2025-02316431, and  RS-2025-18362970).
This research used resources of the National Energy Research Scientific Computing Center, a DOE Office of Science User Facility
supported by the Office of Science of the U.S. Department of Energy under Contract No. DE-AC02-05CH11231 using NERSC award NERSC DDR-ERCAP0038779.

\bibliography{references}

\appendix
\onecolumngrid

\renewcommand{\thefigure}{A\arabic{figure}}
\setcounter{figure}{0}
\renewcommand{\thetable}{A\arabic{table}}
\setcounter{table}{0}

\section{Group Theory}
\label{app:group_theory}

\subsection{Basics of group and representation theory}
\label{app:basic_of_group}

\subsubsection{Group definition.}
A group $(G, \cdot)$ is a set $G$ equipped with a binary operation $\cdot: G \times G \to G$ satisfying the following axioms: (i) associativity: $(a \cdot b) \cdot c = a \cdot (b \cdot c)$ for all $a, b, c \in G$; (ii) existence of identity element: there exists $e \in G$ such that $e \cdot a = a \cdot e = a$ for all $a \in G$; (iii) existence of inverse: for each $a \in G$, there exists $a^{-1} \in G$ such that $a \cdot a^{-1} = a^{-1} \cdot a = e$.

\subsubsection{Group homomorphism.}
A map $\varphi:G\to H$ is a group homomorphism if $\varphi(g_1 g_2)=\varphi(g_1)\varphi(g_2)$ for all $g_1,g_2\in G$.
Its kernel and image are $\ker\varphi=\{g\in G:\varphi(g)=e_H\}$ and $\mathrm{im}\,\varphi=\varphi(G)$, where $e_G$ and $e_H$ denote the identity elements of $G$ and $H$.
The map $\varphi$ is injective if and only if $\ker\varphi=\{e_G\}$, and it is surjective if and only if $\mathrm{im}\,\varphi=H$.
We denote a set of all homomorphisms to itself, $\iota: G \to G$, by $\mathrm{End}(G)$.
A bijective homomorphism is an isomorphism. If there exists an isomorphism between $G$ and $H$, we write $G \cong H$.
An automorphism is an isomorphism $\alpha:G\to G$; the set of all automorphisms, $\mathrm{Aut}(G)$, forms a group under composition.

\subsubsection{Group representation.}
A (linear) representation of a group $G$ on a vector space $V$ over a field $\mathbb{K}$ is a homomorphism $\rho:G\to \mathrm{GL}(V)$.
Here $\mathrm{GL}(V)$ denotes the group of all invertible linear maps $V\to V$ under composition; if $\dim_{\mathbb{K}} V=d$, then, for a given basis of $V$, $\mathrm{GL}(V)\cong \mathrm{GL}_d(\mathbb{K})$. We write $\rho(g)v$ for the action of $g\in G$ on $v\in V$, which satisfies $\rho(e)=\mathrm{Id}_V$ and $\rho(g_1g_2)=\rho(g_1)\rho(g_2)$.
In Hilbert spaces one often uses a unitary representation $\rho:G\to \mathcal{U}(\mathcal{H})$ where $\mathcal{U}(\mathcal{H})$ is the unitary group on $\mathcal{H}$.

\subsubsection{Irreducible representation.}
Let $\rho:G\to \mathrm{GL}(V)$ be a representation. A subspace $W\subset V$ is $G$-invariant if $\rho(g)W\subseteq W$ for all $g\in G$.
The representation $\rho$ is called irreducible if it has no nontrivial $G$-invariant subspace, i.e., the only $G$-invariant subspaces are $\{0\}$ and $V$; otherwise it is reducible.

\subsubsection{Character.}
Let $\rho:G\to \mathrm{GL}(V)$ be a finite-dimensional representation of a group $G$ over $\mathbb{C}$. 
The character of $V$ is the map $\chi_V:G\to\mathbb{C}$ defined by $\chi_V(g):=\mathrm{Tr}(\rho(g))$ for $g\in G$, where $\mathrm{Tr}$ denotes the trace of the linear operator $\rho(g)$ on $V$. 
If $V$ is irreducible, then $\chi_V$ is called an irreducible character.

The character is a class function, i.e., it is constant on conjugacy classes: $\chi_V(hgh^{-1})=\chi_V(g)$ for all $g,h\in G$.

If $V$ and $W$ are finite-dimensional representations of $G$, then their characters satisfy $\chi_{V\oplus W}=\chi_V+\chi_W$, $\chi_{V\otimes W}=\chi_V\chi_W$, and $\chi_{V^*}(g)=\chi_V(g^{-1})$. 
If, in addition, $\rho$ is unitary, then $\chi_{V^*}(g)=\overline{\chi_V(g)}$.

\subsubsection{Lie algebra.}
A Lie algebra $\mathfrak g$ is (typically) realized as the tangent space at the identity of a Lie group $G$, equipped with a bilinear bracket. Concretely, for a Lie group $G$ the Lie algebra is
\begin{equation}
\mathfrak g := T_e G
= \big\{ \dot\gamma(0) : \gamma:(-\varepsilon,\varepsilon)\to G \text{ smooth},\ \gamma(0)=e \big\},
\end{equation}
together with a bilinear map $[\cdot,\cdot]:\mathfrak g\times\mathfrak g\to\mathfrak g$ (the Lie bracket) that is antisymmetric and satisfies the Jacobi identity
\begin{equation}
[X,Y] = -[Y,X], \qquad
[X,[Y,Z]] + [Y,[Z,X]] + [Z,[X,Y]] = 0.
\end{equation}
For a matrix Lie group $G\subseteq \mathrm{GL}_n(\mathbb C)$, one can identify
\begin{equation}
\mathfrak g = \{ X\in\mathbb C^{n\times n} : \exp(tX)\in G \text{ for all sufficiently small } t\},
\end{equation}
and in this case the Lie bracket is simply the matrix commutator
\begin{equation}
[X,Y] := XY - YX.
\end{equation}

\subsection{SU(2) Group}
\label{app:su2_group}

\subsubsection{$\mathrm{SU}(2)$ group definition.}
$\mathrm{SU}(2)=\{U\in\mathbb{C}^{2\times2}:U^\dagger U=I,\ \det(U)=1\}$, with group operation given by matrix multiplication (identity $I$, inverse $U^{-1}=U^\dagger$).
Every $U\in\mathrm{SU}(2)$ can be written as
\begin{equation}\label{eq:su2-param}
U=\begin{psmallmatrix}\alpha&-\bar\beta\\ \beta&\bar\alpha\end{psmallmatrix},\quad
\alpha,\beta\in\mathbb{C},\ |\alpha|^2+|\beta|^2=1 .
\end{equation}
The Lie algebra of $\mathrm{SU}(2)$ is given by
\begin{equation}\label{eq:su2-lie}
\mathfrak{su}(2)=\{A\in\mathbb{C}^{2\times2}:A^\dagger=-A,\ \mathrm{tr}\,A=0\},
\end{equation}
spanned (over $\mathbb{R}$) by $\{\,i\sigma_X,\ i\sigma_Y,\ i\sigma_Z\,\}$, where the Pauli matrices are
\begin{align}\label{eq:paulis}
\sigma_X = \begin{pmatrix}
0 & 1\\
1 & 0
\end{pmatrix},\quad
\sigma_Y=\begin{pmatrix}
0 & -i\\
i & 0
\end{pmatrix},\quad
\sigma_Z=\begin{pmatrix}
1 & 0\\
0 & -1
\end{pmatrix}.
\end{align}

They satisfy
\begin{equation}\label{eq:pauli-comm}
[\sigma_i,\sigma_j]=2i\,\epsilon_{ijk}\,\sigma_k .
\end{equation}
Equivalently, any Lie-algebra element $X\in\mathfrak{su}(2)$ can be written as
\begin{equation}\label{eq:axis-angle-algebra}
X=-\tfrac{i}{2}\,\theta\,\hat n\cdot\vec\sigma,
\end{equation}
and exponentiating yields a group element $\exp(X)\in\mathrm{SU}(2)$.

\subsubsection{Connection to 3D rotation.}
\begin{equation}
\Phi:\mathrm{SU}(2)\!\to\!\mathrm{SO}(3)\ \text{ is surjective with kernel }\{\pm I\},\ \ \mathrm{SO}(3)\cong \mathrm{SU}(2)/\{\pm I\}.
\end{equation}

\emph{Proof.}
The space of traceless Hermitian $2\times 2$ matrices is a vector space over $\mathbb{R}$ with a basis $\{\sigma_X,\sigma_Y,\sigma_Z\}$.
For $U\in\mathrm{SU}(2)$ and each $j$, unitary conjugation preserves the hermiticity and trace, so
\begin{equation}\label{eq:adjoint-R}
U\sigma_jU^\dagger=\sum_{i=1}^3 R_{ij}(U)\,\sigma_i,
\end{equation}
for $ R_{ij}(U) \in \mathbb{R}$.
Then, we define $\Phi$ as
\begin{align}
\Phi(U) = (R_{ij}(U))_{1\leq i,j \leq 3}.
\end{align}
One can show that $\Phi$ satisfies the desired properties as follows.

First, multiplying by $\sigma_k$ and taking the trace of both sides of Eq.~\eqref{eq:adjoint-R} gives
\begin{equation}
\operatorname{tr}\big(\sigma_k\,U\sigma_j U^\dagger\big) =\sum_i R_{ij}(U)\,\operatorname{tr}(\sigma_k\sigma_i).
\end{equation}
Hence, we obtain
\begin{equation}
R_{kj}(U)=\frac12\,\operatorname{tr}\!\big(\sigma_k\,U\sigma_jU^\dagger\big)
\quad\  (\;\because\; \operatorname{tr}(\sigma_k\sigma_i)=2\delta_{ki}).
\end{equation}

We then prove that $R_{ij}(U)\in \SO(3)$ by evaluating $\tfrac12\,\operatorname{tr}(U\sigma_iU^\dagger\,U\sigma_jU^\dagger)$ in two ways. First, 
\begin{equation}\label{eq:R-orthogonal}
\tfrac12\,\mathrm{tr}\!\big(U\sigma_iU^\dagger\,U\sigma_jU^\dagger\big)
=\tfrac12\,\mathrm{tr}(\sigma_i\sigma_j)=\delta_{ij}
\
\end{equation}

Second, expanding $U\sigma_iU^\dagger$ and $U\sigma_jU^\dagger$ in the Pauli basis via~\eqref{eq:adjoint-R}, we rewrite the left-hand side in terms of $R(U)$:
\begin{equation}
\begin{aligned}
\tfrac12\,\mathrm{tr}\!\big(U\sigma_iU^\dagger\,U\sigma_jU^\dagger\big)
&=\tfrac12\,\mathrm{tr}\!\Big(\Big(\sum_{k=1}^3 R_{ki}(U)\sigma_k\Big)\Big(\sum_{\ell=1}^3 R_{\ell j}(U)\sigma_\ell\Big)\Big) \\
&=\sum_{k,\ell} R_{ki}(U)R_{\ell j}(U)\,\tfrac12\,\mathrm{tr}(\sigma_k\sigma_\ell) \\
&=\sum_{k} R_{ki}(U)R_{k j}(U)
=(R(U)^\top R(U))_{ij}.
\end{aligned}
\end{equation}
Comparing the two expressions yields $(R(U)^\top R(U))_{ij}=\delta_{ij}$, i.e., $R(U)^\top R(U)=I$.
Moreover, by continuity and $R(I)=I$, we get $\det( R(U))=+1$, hence $R(U)\in\mathrm{SO}(3)$.

Next, we prove that $\Phi$ is a homomorphism.
By the linearity, for any $\vec v\in\mathbb{R}^3$,
\begin{equation}\label{eq:adjoint-vector}
U(\vec v\!\cdot\!\vec\sigma)U^\dagger=\sum_{i,j}R_{ij}(U)v_j\,\sigma_i=(R(U)\vec v)\!\cdot\!\vec\sigma.
\end{equation}
Thus the assignment $\Phi:\mathrm{SU}(2)\to\mathrm{SO}(3)$, $U\mapsto R(U)$, is well-defined and satisfies
\begin{equation}\label{eq:phi-hom}
(UV)(\vec v\!\cdot\!\vec\sigma)(UV)^\dagger
=U\!\left(V(\vec v\!\cdot\!\vec\sigma)V^\dagger\right)\!U^\dagger
\ \Rightarrow\ \Phi(UV)=\Phi(U)\Phi(V),
\end{equation}
(i.e., a group homomorphism).

If $\Phi(U)=I$, then $U\sigma_iU^\dagger=\sigma_i$ for all $i$, hence $U$ commutes with the matrix algebra generated by $\{\sigma_i\}$ and must be $\pm I$; therefore $\ker\Phi=\{\pm I\}$. Moreover, the differential $d\Phi_e:\mathfrak{su}(2)\to\mathfrak{so}(3)$ is an isomorphism, so the image is the connected Lie subgroup with Lie algebra $\mathfrak{so}(3)$, namely $\mathrm{SO}(3)$, 
hence $\operatorname{im}(\Phi)=\mathrm{SO}(3)$.
Since $\operatorname{im}(\Phi)=\mathrm{SO}(3)$ and $\ker\Phi=\{\pm I\}$, the First Isomorphism Theorem yields $\mathrm{SO}(3)\cong \mathrm{SU}(2)/\{\pm I\}$.

\subsubsection{Single-qubit gates.}
Single-qubit operations can be represented by $\mathrm{SU}(2)$. Standard rotations
\begin{equation}
R_X(\theta)=e^{-i\theta\sigma_x/2},\quad
R_Y(\theta)=e^{-i\theta\sigma_y/2},\quad
R_Z(\theta)=e^{-i\theta\sigma_z/2}
\end{equation}
generate arbitrary $U\in\mathrm{SU}(2)$ (e.g., ZYZ/Euler decompositions).

\subsubsection{$\mathrm{SU}(2)$ irreducible representations.}
All finite-dimensional unitary irreps of $\mathrm{SU}(2)$ are indexed by total angular momentum $j\in\{0,\tfrac12,1,\tfrac32,\dots\}$. The spin-$j$ irrep $V^{(j)}$ has dimension $2j+1$ with weight basis $\{|j,m\rangle: m=-j,-j+1,\dots,j\}$ on which the generators $J_x,J_y,J_z$ satisfy $[J_x,J_y]=iJ_z$ (cyclic), $J_z|j,m\rangle=m|j,m\rangle$, $J^2|j,m\rangle=j(j+1)|j,m\rangle$, and with $J_\pm:=J_x\pm iJ_y$,
\begin{equation}\label{eq:J-raising}
J_\pm|j,m\rangle=\sqrt{(j\mp m)(j\pm m+1)}\,|j,m\!\pm\!1\rangle.
\end{equation}
Group elements $U=\exp(-i\theta\,\hat n\!\cdot\!\vec\sigma/2)$ act in spin-$j$ as $D^{(j)}(U)=\exp(-i\theta\,\hat n\!\cdot\!\mathbf J^{(j)})$, with $\mathbf J^{(1/2)}=\tfrac12\vec\sigma$. Tensor products decompose by the Clebsch–Gordan rule $V^{(j_1)}\otimes V^{(j_2)}=\bigoplus_{j=|j_1-j_2|}^{\,j_1+j_2} V^{(j)}$; in particular 
\begin{equation}
V^{(1/2)} \otimes V^{(1/2)} = V^{(0)}\oplus V^{(1)},
\end{equation}
where the trivial one-dimensional $j=0$ block is the singlet. For simplicity, one often uses $j$ to denote $V^{(j)}$, e.g., $1/2 \otimes 1/2 = 0 \oplus 1$.

For \(n\)-qubit systems, within $(\mathbb C^2)^{\otimes n}$, the $\mathrm{SU}(2)$ irreducible sectors are indexed by partitions with at most two parts. Specifically, a partition
\(
\lambda=(\lambda_1,\lambda_2)
\)
corresponds to the $\mathrm{SU}(2)$ irrep of spin
\(
j=\frac{\lambda_1-\lambda_2}{2}.
\)
Thus, within $(\mathbb C^2)^{\otimes n}$, the partition $(n-k,k)$ labels the spin-\(\frac{n-2k}{2}\) sector in the Schur--Weyl decomposition (Appendix~\ref{app:schur_weyl_duality}).

\subsection{Symmetric group} 
\label{app:symmetric_group}

\subsubsection{Symmetric group definition.}
Let $[n]\!:=\!\{1,\dots,n\}$ and define the symmetric group as the set of all set–set bijections
\begin{equation}\label{eq:Sn-def}
\Sym_n \equiv \mathrm{Sym}([n]) := \{\sigma:[n]\to[n]\;|\;\sigma\ \text{bijective}\}.
\end{equation}
With composition $(\sigma, \sigma')\mapsto \sigma\circ\sigma'$, identity $\mathrm{id}_{[n]}$, and inverses $\sigma^{-1}$, we have $|\Sym_n|=n!$. 

On the $n$-qubit space $(\mathbb{C}^2)^{\otimes n}$ we use the action that permutes tensor factors. Define
\begin{equation}\label{eq:perm-rep}
\Pi: \Sym_n \longrightarrow \mathrm{U}\!\big((\mathbb{C}^2)^{\otimes n}\big),
\quad
\Pi(\sigma)\big(v_1\!\otimes\!\cdots\!\otimes\! v_n\big)
:= v_{\sigma^{-1}(1)}\!\otimes\!\cdots\!\otimes\! v_{\sigma^{-1}(n)} .
\end{equation}
On computational basis states $\ket{x_1\cdots x_n}$ with $x_i\in\{0,1\}$,
\begin{equation}\label{eq:perm-basis}
\Pi(\sigma)\ket{x_1\cdots x_n}
= \ket{x_{\sigma^{-1}(1)}\cdots x_{\sigma^{-1}(n)}} .
\end{equation}
The inverse of indices ensures homomorphism:
\begin{equation}\label{eq:perm-hom}
\Pi(\sigma)\Pi(\sigma')=\Pi(\sigma\sigma'), \qquad \Pi(\mathrm{id})=I .
\end{equation}

\subsubsection{Symmetric group irreducible representations.}
Conjugacy classes of $\Sym_n$ are classified by cycle type. 
Indeed, every permutation decomposes into disjoint cycles, uniquely up to reordering of the cycles, and conjugation preserves this cycle structure. 
Hence conjugacy classes of $\Sym_n$ are indexed by partitions of $n$.

A partition of $n$ is a weakly decreasing sequence
\begin{equation}\label{eq:partition-def-app}
\lambda=(\lambda_1,\dots,\lambda_\ell)\vdash n,
\qquad
\lambda_1\ge \cdots \ge \lambda_\ell >0,
\qquad
\sum_{i=1}^{\ell}\lambda_i=n.
\end{equation}
Over $\mathbb{C}$, the irreducible representations of $\Sym_n$ are likewise indexed by partitions:
\begin{equation}\label{eq:Sn-irrep-indexed-by-partitions}
\{\text{irreducible representations of }\Sym_n\}
=
\{\,\Sym^\lambda \mid \lambda\vdash n\,\}.
\end{equation}
A concrete construction of $\Sym^\lambda$ is obtained from Young symmetrizers, as described below.

To a partition $\lambda\vdash n$ we associate its Young diagram, namely the left-justified array with $\lambda_i$ boxes in row $i$. 
For example, the partition $\lambda=(4,2,1)\vdash 7$ is represented by
\[
\begin{array}{|c|c|c|c|}
\hline
\phantom{x}&\phantom{x}&\phantom{x}&\phantom{x}\\
\hline
\phantom{x}&\phantom{x}\\
\cline{1-2}
\phantom{x}\\
\cline{1-1}
\end{array}
\]
A Young tableau of shape $\lambda$ is a filling of these boxes with $1,\dots,n$. 
For example, a tableau of shape $(4,2,1)$ is
\[
\begin{array}{|c|c|c|c|}
\hline
1&2&3&4\\
\hline
5&6\\
\cline{1-2}
7\\
\cline{1-1}
\end{array}
\]

Fix a partition $\lambda\vdash n$ and a Young tableau $T$ of shape $\lambda$.
Let $R_T\le \Sym_n$ be the row stabilizer and $C_T\le \Sym_n$ the column stabilizer:
\begin{equation}\label{eq:row-col-stabilizer-app}
R_T:=\{\sigma\in\Sym_n \mid \sigma \text{ preserves each row of }T\},
\qquad
C_T:=\{\tau\in\Sym_n \mid \tau \text{ preserves each column of }T\}.
\end{equation}
Define
\begin{equation}\label{eq:young-symmetrizer-app}
a_T=\sum_{\sigma\in R_T}\sigma,
\qquad
b_T=\sum_{\tau\in C_T}\mathrm{sgn}(\tau)\,\tau,
\qquad
e_T=a_T b_T\in\mathbb{C}[\Sym_n].
\end{equation}
In general one has
\begin{equation}\label{eq:quasi-idempotent-app}
e_T^2=c_T\,e_T
\end{equation}
for some nonzero scalar $c_T$, and we set the primitive idempotent
\begin{equation}\label{eq:primitive-idempotent-app}
e_T^\circ:=c_T^{-1}e_T,
\qquad
(e_T^\circ)^2=e_T^\circ.
\end{equation}

The left ideal
\begin{equation}\label{eq:specht-module-app}
\mathbb{C}[\Sym_n]\,e_T
\;\cong\;
\mathbb{C}[\Sym_n]\,e_T^\circ
\end{equation}
is called the Specht module associated with $T$. 
Although the element $e_T$ depends on the labeling of the tableau $T$, the resulting representation depends only on the shape $\lambda$. 
More precisely, if $T$ and $T'$ have the same shape $\lambda$, then
\begin{equation}\label{eq:labeling-not-important}
\mathbb{C}[\Sym_n]\,e_T \cong \mathbb{C}[\Sym_n]\,e_{T'} \cong \Sym^\lambda.
\end{equation}
Thus the labeling is not essential; up to isomorphism, only the partition $\lambda$ matters. 
Accordingly, one also refers to this representation as the Specht module of shape $\lambda$.

The dimension of $\Sym^\lambda$ equals the number of standard Young tableaux of shape $\lambda$, and is given by the hook-length formula:
\begin{equation}\label{eq:hook-length-formula-app}
\dim \Sym^\lambda
=
f^\lambda
=
\frac{n!}{\prod_{b\in\lambda} h(b)}.
\end{equation}
Here the hook length $h(b)$ of a box $b$ is the number of boxes directly to the right of $b$, directly below $b$, plus the box $b$ itself.

Two basic examples are
\begin{equation}\label{eq:trivial-sign-app}
\Sym^{(n)} \cong \mathbf{1},
\qquad
\Sym^{(1^n)} \cong \mathrm{sgn}.
\end{equation}
More generally, the partition $(n-1,1)$ yields the standard representation of $\Sym_n$.

\subsubsection{Example: The irreducible representations of \texorpdfstring{$\Sym_3$}{S3}.}
The partitions of $3$ are
\begin{equation}
(3),\qquad (2,1),\qquad (1,1,1).
\end{equation}
Hence $\Sym_3$ has exactly three inequivalent irreducible representations:
\begin{equation}
\Sym^{(3)},\qquad \Sym^{(2,1)},\qquad \Sym^{(1,1,1)}.
\end{equation}

For the shape $(3)$, take the tableau
\[
T=
\begin{array}{|c|c|c|}
\hline
1&2&3\\
\hline
\end{array}.
\]
Then
\begin{equation}
R_T=\Sym_3,
\qquad
C_T=\{e\},
\end{equation}
so
\begin{equation}\label{eq:S3-trivial-idempotent}
e_T=\sum_{\sigma\in\Sym_3}\sigma.
\end{equation}
The corresponding module is the trivial representation:
\begin{equation}\label{eq:S3-trivial-rep}
\mathbb{C}[\Sym_3]\,e_T \cong \Sym^{(3)} \cong \mathbf{1}.
\end{equation}
Its dimension is obtained from the hook lengths
\[
\begin{array}{|c|c|c|}
\hline
\phantom{x}&\phantom{x}&\phantom{x}\\
\hline
\end{array}
\qquad \leadsto \qquad
\begin{array}{|c|c|c|}
\hline
3&2&1\\
\hline
\end{array}
\]
as
\begin{equation}\label{eq:S3-dim-3}
\dim \Sym^{(3)}
=
\frac{3!}{3\cdot 2\cdot 1}
=
1.
\end{equation}

For the shape $(2,1)$, take the tableau
\[
T=
\begin{array}{|c|c|}
\hline
1&2\\
\hline
3\\
\cline{1-1}
\end{array}.
\]
Then
\begin{equation}
R_T=\{e,(12)\},
\qquad
C_T=\{e,(13)\},
\end{equation}
and hence
\begin{equation}\label{eq:S3-standard-idempotent}
e_T=(e+(12))(e-(13)).
\end{equation}
This gives the two-dimensional irreducible representation
\begin{equation}\label{eq:S3-standard-rep}
\mathbb{C}[\Sym_3]\,e_T \cong \Sym^{(2,1)}.
\end{equation}
Its dimension follows from the hook lengths
\[
\begin{array}{|c|c|}
\hline
\phantom{x}&\phantom{x}\\
\hline
\phantom{x}\\
\cline{1-1}
\end{array}
\qquad \leadsto \qquad
\begin{array}{|c|c|}
\hline
3&1\\
\hline
1\\
\cline{1-1}
\end{array}
\]
as
\begin{equation}\label{eq:S3-dim-21}
\dim \Sym^{(2,1)}
=
\frac{3!}{3\cdot 1\cdot 1}
=
2.
\end{equation}

For the shape $(1,1,1)$, take the tableau
\[
T=
\begin{array}{|c|}
\hline
1\\
\hline
2\\
\hline
3\\
\hline
\end{array}.
\]
Then
\begin{equation}
R_T=\{e\},
\qquad
C_T=\Sym_3,
\end{equation}
so
\begin{equation}\label{eq:S3-sign-idempotent}
e_T=\sum_{\tau\in\Sym_3}\mathrm{sgn}(\tau)\,\tau.
\end{equation}
The corresponding module is the sign representation:
\begin{equation}\label{eq:S3-sign-rep}
\mathbb{C}[\Sym_3]\,e_T \cong \Sym^{(1,1,1)} \cong \mathrm{sgn}.
\end{equation}
Its dimension is obtained from the hook lengths
\[
\begin{array}{|c|}
\hline
\phantom{x}\\
\hline
\phantom{x}\\
\hline
\phantom{x}\\
\hline
\end{array}
\qquad \leadsto \qquad
\begin{array}{|c|}
\hline
3\\
\hline
2\\
\hline
1\\
\hline
\end{array}
\]
as
\begin{equation}\label{eq:S3-dim-111}
\dim \Sym^{(1,1,1)}
=
\frac{3!}{3\cdot 2\cdot 1}
=
1.
\end{equation}

Therefore, the irreducible representations of $\Sym_3$ are
\begin{equation}\label{eq:S3-summary}
\Sym^{(3)},\qquad \Sym^{(2,1)},\qquad \Sym^{(1,1,1)},
\end{equation}
with dimensions
\begin{equation}\label{eq:S3-dimensions}
\dim \Sym^{(3)}=1,\qquad
\dim \Sym^{(2,1)}=2,\qquad
\dim \Sym^{(1,1,1)}=1.
\end{equation}

\subsection{Schur's Lemma}
\label{app:schurs_lemma}

\subsubsection{Statement \& proof.}
Let $\rho_V: G \to \mathrm{GL}(V)$ and $\rho_W: G \to \mathrm{GL}(W)$ be irreducible $G$-representations over a field $K$,
and let $T:V\to W$ be a linear map such that $\rho_W(g)\,T = T\,\rho_V(g)$ for all $g\in G$ (i.e., $T$ is a $G$-intertwiner).
Then:
\begin{enumerate}
\item Either $T=0$ or $T$ is an isomorphism $V \xrightarrow{\ \cong\ } W$.
\item In particular, if $V=W$ and $\rho_V=\rho_W$ over an algebraically closed field, then $T=\lambda I_V$ for some $\lambda\in K$.
\end{enumerate}

\emph{Proof.}
Because $T$ intertwines the actions, both $\ker T\subseteq V$ and $\mathrm{Im}\,T\subseteq W$ are $G$-invariant:
For any $v\in\ker T$ and $g\in G$, $T(\rho_V(g)v)=\rho_W(g)T(v)=0$, hence $\rho_V(g)v\in\ker T$.
For any $u\in\mathrm{Im}\,T$, choose $v'\in V$ with $u=T(v')$; then
$\rho_W(g)u=\rho_W(g)T(v')=T(\rho_V(g)v')\in\mathrm{Im}\,T$.
Irreducibility gives $\ker T\in\{0,V\}$ and $\mathrm{Im}\,T\in\{0,W\}$.
If $T\neq 0$, then $\mathrm{Im}\,T=W$ and $\ker T=0$, hence $T$ is bijective.
Over an algebraically closed field, any $T\in\mathrm{End}_G(V)$ has an eigenvalue $\lambda$; then $T-\lambda I$ is an intertwiner with nontrivial kernel, so it must be zero, yielding $T=\lambda I$.

\subsubsection{Schur's Lemma application in $\mathrm{SU}(2)$.}
View $n$ qubits as $n$ copies of the spin-$\tfrac{1}{2}$ irrep of $\mathrm{SU}(2)$.
The tensor space decomposes as
\begin{equation}
(\mathbb{C}^{2})^{\otimes n}
\;\cong\;
\bigoplus_{j}\,\mathbb{C}^{m_{j}}\otimes J^{(j)} .
\end{equation}
where $J^{(j)}$ is the spin-$j$ irrep and $m_{j}$ is its multiplicity.
The Schur (spin) map $S_{n}$ implements this change of basis, sending the computational basis to the spin basis in which the $\mathrm{SU}(2)$ action is block diagonal by total spin $j$.

Let $U$ be any $\mathrm{SU}(2)$-equivariant linear operator on $(\mathbb{C}^{2})^{\otimes n}$ (equivalently, $U$ commutes with the group action).
Conjugating into the spin basis, let $U_{j,j'}$ be a block of $U$ between spin-$j$ and spin-$j'$ space, i.e., $U_{j,j'}:\mathbb{C}^{m_{j'}}\otimes J^{(j')} \to \mathbb{C}^{m_j}\otimes J^{(j)}$.
The commutation relation reads
\begin{equation}
\bigl(I_{m_j}\otimes \rho_j(g)\bigr)\,U_{j,j'}
\;=\;
U_{j,j'}\,\bigl(I_{m_{j'}}\otimes \rho_{j'}(g)\bigr),
\qquad \forall g\in SU(2),
\end{equation}
where $\rho_j:\mathrm{SU}(2)\to \mathrm{U}(J^{(j)})$ denotes the spin-$j$ irrep.
Each $U_{j,j'}$ is therefore an intertwiner between irreps.
By Schur's Lemma, $U_{j,j'}=0$ unless $j=j'$, so $U$ is block diagonal in the spin basis; moreover, on each irrep block,
\begin{equation}
U\big|_{\mathbb{C}^{m_j}\otimes J^{(j)}}
\;=\;
A_j \otimes I_{J^{(j)}},
\qquad A_j \in \mathrm{End}\!\bigl(\mathbb{C}^{m_j}\bigr),
\end{equation}
Thus, after transforming by $\Sym_{n}$, every $\mathrm{SU}(2)$-equivariant operator decomposes as 
$\bigoplus_j \bigl(A_j\otimes I_{J^{(j)}}\bigr)$: It acts as the identity on each spin irrep $J^{(j)}$, with freedom only on the multiplicity spaces.

\subsection{Schur--Weyl duality}
\label{app:schur_weyl_duality}

\subsubsection{Statement \& proof}
Let \(V\cong\mathbb{C}^d\) and \(W:=V^{\otimes n}\). There are commuting actions
\begin{equation}
\rho:\mathrm{GL}(V)\to \mathrm{GL}(W),\quad \rho(T)=T^{\otimes n},
\qquad
\Pi:\Sym_n\to \mathrm{GL}(W),\quad \Pi(\sigma)\ \text{permutes tensor factors},
\end{equation}
so that \(\rho(T)\Pi(\sigma)=\Pi(\sigma)\rho(T)\) for all \(T\in \mathrm{GL}(V),\sigma\in \Sym_n\).
Write
\begin{equation}
\mathcal{A} := {\Pi}(\mathbb{C}[\Sym_n]) \subseteq \mathrm{End}(W),\qquad
\mathcal{B} := \mathrm{alg}\langle T^{\otimes n}: T\in \mathrm{GL}(V)\rangle \subseteq \mathrm{End}(W).
\end{equation}
Then
\begin{equation}
\mathcal{A}=\mathrm{End}_{\mathrm{GL}(V)}(W),
\qquad
\mathcal{B}=\mathrm{End}_{\mathbb{C}[\Sym_n]}(W).
\end{equation}
Moreover, there is a canonical simultaneous decomposition

\begin{equation}
W
\;\cong\;
\bigoplus_{\lambda\vdash n,\;\ell(\lambda)\le d}
\mathcal{V}^{\lambda}\,\otimes\,\Sym^{\lambda},
\end{equation}
where \(\Sym^{\lambda}\) is the Specht module (irreducible \(\Sym_n\)-module) and \(\mathcal{V}^{\lambda}\) is the Schur functor applied to \(V\) (irreducible \(\mathrm{GL}(V)\)-module).
The sum is over all $\lambda$, the shape of a Young tableau with $n$ boxes, and $\ell(\lambda)$ is the number of rows in the Young tableau with shape $\lambda$.
We require $\ell(\lambda) \leq d$, which is the necessary and sufficient condition for \(\mathcal{V}^{\lambda}\neq 0\).

\emph{Proof.}
Since \(\rho\) and \(\Pi\) commute, \(\mathcal{A}\subseteq \mathrm{End}_{\mathrm{GL}(V)}(W)\) and \(\mathcal{B}\subseteq \mathrm{End}_{\mathbb{C}[\Sym_n]}(W)\).
Over \(\mathbb{C}\), \(W\) is semisimple for the commuting algebras \(\mathcal{A}\) and \(\mathcal{B}\).
By the double centralizer theorem, the inclusions are equalities and \(W\) decomposes as a direct sum of tensor products of irreducibles:
\begin{equation}
W\cong \bigoplus_\lambda \mathcal{V}^{\lambda}\otimes \Sym^{\lambda}.
\end{equation}
In the following paragraph, we characterize $\Sym_{\lambda}(V)$ inside $W$ using
Young symmetrizers.

\subsubsection{How to construct the block $\mathcal{V}^{\lambda}\otimes \Sym^\lambda$}
Fix a partition \(\lambda\vdash n\) and a Young tableau \(T\) of shape \(\lambda\).
Let \(R_T\le \Sym_n\) be the row stabilizer and \(C_T\le \Sym_n\) the column stabilizer.
Define
\begin{equation}
a_T=\sum_{\sigma\in R_T}\sigma,\qquad
b_T=\sum_{\tau\in C_T}\mathrm{sgn}(\tau)\,\tau,\qquad
e_T=a_T\,b_T\in\mathbb{C}[\Sym_n].
\end{equation}
In general \(e_T^2=c_T\,e_T\) for some nonzero scalar \(c_T\); set the primitive idempotent \(e_T^\circ:=c_T^{-1}e_T\) so that \((e_T^\circ)^2=e_T^\circ\).
Then
\begin{equation}
\Sym^{\lambda} \;\cong\; \mathbb{C}[\Sym_n]\cdot e_T^\circ,\qquad
\mathcal{V}^{\lambda} \;\cong\; \Pi(e_T^\circ)\,W.
\end{equation}
Choosing a different \(T\) of the same shape gives an isomorphic module.

A linear basis of \(\Sym^{\lambda}\) is given by polytabloids attached to the standard Young tableaux (SYT) of shape \(\lambda\):
if \(\{T\}\) denotes the tabloid of \(T\), the basis vectors are
\begin{equation}
E_T\ :=\ e_T^\circ\,\{T\}.
\end{equation}
A linear basis of \(\mathcal{V}^{\lambda}\) is indexed by semi-standard Young tableaux (SSYT) of shape \(\lambda\) filled with \(\{1,\dots,d\}\): rows weakly increasing and columns strictly increasing.
Writing \(V=\mathbb{C}^d\) with basis \(\{e_1,\dots,e_d\}\) and listing the SSYT entries as \((t_1,\dots,t_n)\) in box order, the corresponding basis vectors are
\begin{equation}
v_t\ :=\ \Pi(e_T^\circ)\big(e_{t_1}\otimes e_{t_2}\otimes\cdots\otimes e_{t_n}\big).
\end{equation}
(Here \(T\) is any fixed tableau of shape \(\lambda\); different choices produce canonically isomorphic bases.)

\subsubsection{Dimensions}
For \(\lambda\vdash n\) with boxes \(u\), hook length \(h(u)\) and content \(c(u)=\mathrm{col}(u)-\mathrm{row}(u)\),
\begin{equation}
\dim \Sym^{\lambda}=\frac{n!}{\prod_{u\in\lambda}h(u)},\qquad
\dim \mathcal{V}^{\lambda}=\prod_{u\in\lambda}\frac{d+c(u)}{h(u)}.
\end{equation}

\subsection{SU(2) equivariant Gate and generalized permutation}
\label{app:generalized_permutation}

\subsubsection{$\mathrm{SU}(2)$ Equivariant gate to generalized permutation.}

Let $W=(\mathbb{C}^{2})^{\otimes n}$ be a vector space with $\mathrm{SU}(2)$-action $\rho(U)=U^{\otimes n}$
and the permutation action $\Pi(\sigma)$ of $\sigma \in \Sym_n$. Then
\begin{equation}
\mathrm{End}_{\mathrm{SU}(2)}(W)=\Pi\!\big(\mathbb{C}[\Sym_n]\big),
\end{equation}
i.e., every operator commuting with $\rho(U)$ lies in the group algebra generated by
$\{\Pi(\sigma)\}_{\sigma\in \Sym_n}$.
In particular, any operator of the form
\begin{equation}
Q=\exp\!\Big(\sum_{j=1}^{m} c_j\,\Pi(\sigma_j)\Big),
\qquad c_j\in\mathbb{C},\ \sigma_j\in \Sym_n,
\end{equation}
is $\mathrm{SU}(2)$-equivariant.

\emph{Proof.}
Recall that $\Sym_n$ acts on $W$ by permuting tensor factors via
\begin{equation}
\Pi(\sigma)\big(\ket{v_1}\otimes\cdots\otimes\ket{v_n}\big)
=\ket{v_{\sigma^{-1}(1)}}\otimes\cdots\otimes\ket{v_{\sigma^{-1}(n)}}.
\end{equation}
Moreover, the diagonal $\mathrm{SU}(2)$-action $\rho(U)=U^{\otimes n}$ commutes with this permutation action:
\begin{equation}
\rho(U)\,\Pi(\sigma)=\Pi(\sigma)\,\rho(U),
\qquad \forall\,U\in\mathrm{SU}(2),\ \sigma\in \Sym_n.
\end{equation}
By Schur--Weyl duality specialized to $d=2$,
\begin{equation}
W
\cong
\bigoplus_{k=0}^{\lfloor n/2\rfloor}
J^{(n/2-k)} \otimes \mathfrak{S}^{(n-k,k)},
\end{equation}
where $J^{(n/2-k)}$ are the $\mathrm{SU}(2)$ irreps (total spin blocks) and $\Sym^{\lambda}$ are the Specht modules of $\Sym_n$.
With respect to this decomposition the $\mathrm{SU}(2)$–commutant is
\begin{equation}
\mathrm{End}_{\mathrm{SU}(2)}(W)
\;=\;\Pi\big(\mathbb{C}[\Sym_n]\big)
\;\cong\;\bigoplus_{k=0}^{\lfloor n/2\rfloor} \mathbf{1}_{J^{(n/2-k)}}\otimes \mathrm{End}(\Sym^{(n-k,k)}),
\end{equation}
i.e., operators commuting with all $\rho(U)$ act trivially on $J^{(n/2-k)}$ and arbitrarily on the multiplicity spaces $\Sym^{\lambda}$.

We call
\begin{equation}
Q:=\exp\!\Big(\sum_{j=1}^{m} c_j\,\Pi(\sigma_j)\Big),
\qquad c_j\in\mathbb{C},\ \sigma_j\in \Sym_n,
\end{equation}
a generalized-permutation operator. Since $\Pi(\mathbb{C}[\Sym_n])$ is an algebra contained in $\mathrm{End}_{\mathrm{SU}(2)}(W)$, any polynomial in
$\{\Pi(\sigma)\}$ (hence the exponential) is $\mathrm{SU}(2)$-equivariant. Additionally, if $\sum_{j} c_j\,\Pi(\sigma_j)$ is skew-Hermitian, then $Q$ is unitary and usable as a quantum gate.

\subsubsection{Generalized permutation to $\mathrm{SU}(2)$ equivariant gate.}
Let $W:=V^{\otimes n}$ and $V=\mathbb{C}^{2}$. We show that every invertible $\mathrm{SU}(2)$–equivariant operator
(in particular, every unitary gate) $T$ on $W$ is a form of generalized-permutation operator:
$T=\exp[\sum_{j=1}^m c_j\,\Pi(\sigma_j)]$.

\emph{Proof.}
By the Schur–Weyl duality (double-commutant form),
\begin{equation}
\mathrm{End}_{\mathrm{SU}(2)}(W)=\Pi(\mathbb{C}[\Sym_n]) =: \mathcal A,
\end{equation}
a finite-dimensional $\ast$-algebra (i.e., a finite-dimensional algebra equipped with an adjoint operation (conjugate-linear, reverses multiplication, and applying it twice returns the original element)). If $T$ is $\mathrm{SU}(2)$–equivariant, then $T\in\mathcal A$.
Since $T$ is invertible (unitary suffices), the holomorphic functional calculus in the
finite-dimensional algebra $\mathcal A$ provides a logarithm $\log T\in\mathcal A$ such that
$\exp(\log T)=T$.
Because $\mathcal A$ is the linear span of $\{\Pi(\sigma):\sigma\in \Sym_n\}$, there exist
$c_1,\dots,c_m\in\mathbb C$ and $\sigma_1,\dots,\sigma_m\in \Sym_n$ with
\begin{equation}
\log T=\sum_{j=1}^m c_j\,\Pi(\sigma_j).
\end{equation}

Exponentiating gives
\begin{equation}
T=\exp\!\Big(\sum_{j=1}^m c_j\,\Pi(\sigma_j)\Big).
\end{equation}

\subsection{Group twirling.}
\label{app:group_twirling}
We show that $\mathcal{T}_G[A]$ commutes with the group action, i.e., it is $G$–equivariant as an operator.
For any fixed $\tau\in G$,
\begin{equation}
\begin{aligned}
V[\tau]\,\mathcal{T}_G[A]
&= \frac{1}{|G|}\sum_{g\in G}V[\tau]\,V[g]\,A\,V[g]^{-1}
= \frac{1}{|G|}\sum_{g\in G}V[\tau g]\,A\,V[g]^{-1} \\
&\stackrel{g'=\tau g}{=}\frac{1}{|G|}\sum_{g'\in G}V[g']\,A\,V[\tau^{-1}g']^{-1}
= \frac{1}{|G|}\sum_{g'\in G}V[g']\,A\,V[g']^{-1}\,V[\tau] \\
&= \mathcal{T}_G[A]\,V[\tau],
\end{aligned}
\end{equation}
where we used $(\tau^{-1}g')^{-1} = g'^{-1}\tau$ and the homomorphism property of $V$.
Hence $\mathcal{T}_G[A]$ lies in the commutant
\begin{equation}
\{M\in~\mathrm{End}(W):\,M\,V[\tau]=V[\tau]\,M,\ \forall\tau\in G\},
\end{equation}
and is therefore invariant under the adjoint (conjugation) action of $G$. For a unitary representation, we can write $V[\tau]^{-1}$ as $V[\tau]^{\dagger}$.

\section{Proofs}
\label{app:proof}

\subsection{Proof of observation~\ref{obv:joint_su2_sn_equiv_gate}}
\label{app:proof_joint_su2_sn_equiv_gate}

If $M \in \mathrm{End}(W)$ satisfies
\begin{equation}
[M,\rho(U)] = 0
\; \forall U \in \mathrm{SU}(2),\;
\text{and}\;
[M,\Pi(\sigma)] = 0
\; \forall \sigma \in \mathfrak{S}_n,
\label{eq:joint_equivariance_prop}
\end{equation}
then
\begin{equation}
M
=
\bigoplus_{k=0}^{\lfloor n/2\rfloor}
c_k\,
I_{J^{(n/2-k)}}
\otimes
I_{\Sym^{(n-k,k)}}
\qquad
\text{for some } c_k\in\mathbb{C}.
\label{eq:joint_equivariance_scalar_blocks}
\end{equation}
In particular, if $M$ is unitary, then $c_k=e^{i\theta_k}$ for some $\theta_k\in\mathbb{R}$.

\begin{proof}
For each \(k\), consider the summand
\begin{align}
J^{(n/2-k)}\otimes \Sym^{(n-k,k)}
\end{align}
in the decomposition of \(W\). 
By Schur's lemma (Appendix~\ref{app:schurs_lemma}), the condition
\(
[M,\rho(U)] = 0
\)
for all \(U\in \mathrm{SU}(2)\) implies that
\begin{equation}
M|_{(n-k,k)}
=
I_{J^{(n/2-k)}}\otimes B_k
\end{equation}
for some \(B_k\in \mathrm{End}(\Sym^{(n-k,k)})\).
Likewise, the condition
\(
[M,\Pi(\sigma)] = 0
\)
for all \(\sigma\in \mathfrak{S}_n\) implies that
\begin{equation}
M|_{(n-k,k)}
=
C_k\otimes I_{\Sym^{(n-k,k)}}
\end{equation}
for some \(C_k\in \mathrm{End}(J^{(n/2-k)})\).
Therefore,
\begin{equation}
M|_{(n-k,k)}
=
c_k\,
I_{J^{(n/2-k)}}
\otimes
I_{\Sym^{(n-k,k)}}
\end{equation}
for some $c_k\in\mathbb{C}$. Moreover, by Schur's lemma
(Appendix~\ref{app:schurs_lemma}), \(M\) has no off-diagonal
components between distinct \(k\)-sectors. Hence
\begin{equation}
M
=
\bigoplus_{k=0}^{\lfloor n/2\rfloor}
c_k\,
I_{J^{(n/2-k)}}
\otimes
I_{\Sym^{(n-k,k)}}.
\end{equation}
If $M$ is unitary, then each block $M|_{(n-k,k)}$ is unitary. Since
\begin{equation}
M|_{(n-k,k)} = c_k I,
\end{equation}
this implies $|c_k|=1$, so $c_k=e^{i\theta_k}$ for some $\theta_k\in\mathbb{R}$.
\end{proof}

\subsection{Proof of observation~\ref{obv:no_joint_fixed_vectors}}
\label{app:proof_no_joint_fixed_vectors}
Define
\begin{equation}
W^{\mathrm{SU}(2)}
:=
\{v\in W : \rho(U)v=v \text{ for all } U\in \mathrm{SU}(2)\}
\end{equation}
and
\begin{equation}
W^{\mathfrak{S}_n}
:=
\{v\in W : \Pi(\sigma)v=v \text{ for all } \sigma\in \mathfrak{S}_n\}.
\end{equation}
Then
\begin{equation}
W^{\mathrm{SU}(2)} \cap W^{\mathfrak{S}_n} = \{0\}.
\end{equation}

\begin{proof}
Using the above decomposition of \(W\), a vector fixed by \(\mathrm{SU}(2)\) must lie in a summand for which \(J^{(n/2-k)}\) is the trivial \(\mathrm{SU}(2)\)-module. By Appendix~\ref{app:su2_group}, the trivial $\mathrm{SU}(2)$-module is the spin-0 sector
$J^{(0)}$, which corresponds under Schur--Weyl duality to the
partition $(n/2,n/2)$. This occurs only for $k=n/2$, and only
when $n$ is even.

On the other hand, a vector fixed by \(\mathfrak{S}_n\) must lie in a summand for which \(\Sym^{(n-k,k)}\) is the trivial \(\mathfrak{S}_n\)-module. By Appendix~\ref{app:symmetric_group}, the trivial \(\mathfrak{S}_n\)-module is \(\Sym^{(n)}\), this occurs only when \((n-k,k)=(n)\), namely when \(k=0\). Since \(\frac n2 \neq 0\), there is no summand carrying both the trivial \(\mathrm{SU}(2)\)-module and the trivial \(\mathfrak{S}_n\)-module. Hence \begin{equation} W^{\mathrm{SU}(2)} \cap W^{\mathfrak{S}_n}=\{0\}. \end{equation} 
\end{proof}

\subsection{Proof of proposition~\ref{pro:dim_group}} \label{app:proof_dim_group}
The dimension of $\mathbb{C}[\mathfrak{S}_{2N}]^{\mathfrak{S}_{\mathrm{pair}}}$ is
\begin{equation}
\label{eq:pair-dim-formula}
\sum_{\lambda \vdash N}
\prod_{r \ge 1}
r^{\,m_r(\lambda)}
\frac{(2m_r(\lambda))!}{m_r(\lambda)!}.
\end{equation}

\begin{proof}
Since $\mathbb{C}[\mathfrak{S}_{2N}]$ has basis $\mathfrak{S}_{2N}$, conjugation by
$\mathfrak{S}_{\mathrm{pair}}$ permutes this basis. Hence
\begin{equation}
\dim\!\bigl(\mathbb{C}[\mathfrak{S}_{2N}]^{\mathfrak{S}_{\mathrm{pair}}}\bigr)
\end{equation}
is exactly the number of $\mathfrak{S}_{\mathrm{pair}}$-orbits in $\mathfrak{S}_{2N}$
under conjugation. By Burnside's lemma,
\begin{equation}
\dim\!\bigl(\mathbb{C}[\mathfrak{S}_{2N}]^{\mathfrak{S}_{\mathrm{pair}}}\bigr)
=
\frac{1}{|\mathfrak{S}_{\mathrm{pair}}|}
\sum_{\sigma \in \mathfrak{S}_{\mathrm{pair}}}
\bigl|C_{\mathfrak{S}_{2N}}(\sigma)\bigr|,
\end{equation}
where
\begin{align}
C_{\mathfrak{S}_{2N}}(\sigma) = \{\pi \in \mathfrak{S}_{2N} : \sigma\pi\sigma^{-1}=\pi\},
\end{align}
is the centralizer of $\sigma$.
Using the identification $\mathfrak{S}_{\mathrm{pair}} \cong \mathfrak{S}_N$, this becomes
\begin{equation}
\dim\!\bigl(\mathbb{C}[\mathfrak{S}_{2N}]^{\mathfrak{S}_{\mathrm{pair}}}\bigr)
=
\frac{1}{N!}
\sum_{h \in \mathfrak{S}_N}
\bigl|C_{\mathfrak{S}_{2N}}(\widetilde{h})\bigr|.
\end{equation}
For $h \in \mathfrak{S}_N$ of cycle type $\lambda \vdash N$, the corresponding permutation $\widetilde{h}$ has centralizer size
\begin{equation}
\bigl|C_{\mathfrak{S}_{2N}}(\widetilde{h})\bigr|
=
\prod_{r \ge 1}
r^{\,2m_r(\lambda)} (2m_r(\lambda))!,
\end{equation}
since a permutation of cycle type $1^{a_1}2^{a_2}3^{a_3}\cdots$ has centralizer size
\begin{equation}
\prod_{r \ge 1} r^{a_r} a_r!.
\end{equation}
The number of elements of $\mathfrak{S}_N$ of cycle type $\lambda$ is
\begin{equation}
\frac{N!}{z_\lambda},
\qquad
z_\lambda := \prod_{r \ge 1} r^{\,m_r(\lambda)} m_r(\lambda)!.
\end{equation}
Grouping the Burnside sum by cycle type, we obtain
\begin{equation}
\dim\!\bigl(\mathbb{C}[\mathfrak{S}_{2N}]^{\mathfrak{S}_{\mathrm{pair}}}\bigr)
=
\sum_{\lambda \vdash N}
\frac{1}{z_\lambda}
\prod_{r \ge 1}
r^{\,2m_r(\lambda)} (2m_r(\lambda))!.
\end{equation}
Substituting the expression for $z_\lambda$ and simplifying yields
\begin{equation}
\dim\!\bigl(\mathbb{C}[\mathfrak{S}_{2N}]^{\mathfrak{S}_{\mathrm{pair}}}\bigr)
=
\sum_{\lambda \vdash N}
\prod_{r \ge 1}
r^{\,m_r(\lambda)}
\frac{(2m_r(\lambda))!}{m_r(\lambda)!},
\end{equation}
as claimed.
\end{proof}

\subsection{Proof of proposition~\ref{thm:dim_opr}}
\label{app:proof_dim_opr}
The dimension of $\mathcal{A}_{2N}^{\mathfrak{S}_{\mathrm{pair}}}$ is
\begin{equation}
\sum_{\lambda \vdash N}
\frac{1}{z_\lambda}
\sum_{k=0}^{N}
\chi^{(2N-k,k)}(\widetilde{\lambda})^2,
\end{equation}
where \(z_\lambda := \prod_{r\ge 1} r^{m_r(\lambda)} m_r(\lambda)!\), and $\chi^{(2N-k,k)}$ is the irreducible character of $\mathfrak S_{2N}$.

\begin{proof}
By Proposition~\ref{pro:twir_equal},
\begin{equation}
\mathcal{T}_{\mathfrak{S}_{\mathrm{pair}}}[\mathcal{A}_{2N}]
=
\mathcal{A}_{2N}^{\mathfrak{S}_{\mathrm{pair}}}.
\end{equation}
Since $\mathcal{T}_{\mathfrak{S}_{\mathrm{pair}}}$ is the averaging operator over a finite group, it is an idempotent projection onto $\mathcal{A}_{2N}^{\mathfrak{S}_{\mathrm{pair}}}$. Therefore,
\begin{equation}
\dim\!\bigl(\mathcal{A}_{2N}^{\mathfrak{S}_{\mathrm{pair}}}\bigr)
=
\mathrm{Tr}\!\left(
\mathcal{T}_{\mathfrak{S}_{\mathrm{pair}}}\big|_{\mathcal{A}_{2N}}
\right),
\end{equation}
where the trace is taken in the operator space.
In other words, by identifying an operator \begin{align}
    A=\sum_{ij} A_{ij} \ket{i}\bra{j} \longrightarrow  \sum_{i,j} A_{ij}|i,j \rrangle
\end{align}
the twirling map can be written as
\begin{align}
    \mathcal{T}_{\Sym_{\rm pair}} = \frac{1}{|\Sym_{\rm pair}|} \sum_{\sigma \in \Sym_{\rm pair}} \sum_{ij} | \sigma(i), \sigma(j) \rrangle  \llangle i,j|
\end{align}
and
\begin{align}
    \Tr\Bigl[\mathcal{T}_{\Sym_{\rm pair}} \big|_{\mathcal{A}_{2N}} \Bigr] = \sum_{|v\rangle \!\rangle} \llangle v | \mathcal{T}_{\Sym_{\rm pair}} |v \rrangle
\end{align}
where the sum is over all basis elements $|v\rrangle$ of $\mathcal{A}_{2N}$.

Using the identification $\mathfrak{S}_{\mathrm{pair}} \cong \mathfrak{S}_N$, we obtain
\begin{equation}
\dim\!\bigl(\mathcal{A}_{2N}^{\mathfrak{S}_{\mathrm{pair}}}\bigr)
=
\frac{1}{N!}
\sum_{h \in \mathfrak{S}_N}
\mathrm{Tr}\!\left(
\left(
M \mapsto \Pi(\widetilde{h}) M \Pi(\widetilde{h})^\dagger
\right)\big|_{\mathcal{A}_{2N}}
\right).
\end{equation}
By Schur--Weyl duality,
\begin{align}
W
\cong
\bigoplus_{k=0}^{N}
J^{(N-k)}\otimes \Sym^{(2N-k,k)}.
\end{align}
and
\begin{equation}
\mathcal{A}_{2N}
=
\Pi\bigl(\mathbb{C}[\mathfrak{S}_{2N}]\bigr)
\cong
\bigoplus_{k=0}^{N}
\mathrm{End}\!\bigl(\Sym^{(2N-k,k)}\bigr).
\end{equation}
On the $k$-th summand, the map
\begin{equation}
M \mapsto \Pi(\widetilde{h}) M \Pi(\widetilde{h})^\dagger
\end{equation}
acts as
\begin{equation}
T \longmapsto \rho_k(\widetilde{h})\, T\, \rho_k(\widetilde{h})^{-1},
\qquad
T \in \mathrm{End}\!\bigl(\Sym^{(2N-k,k)}\bigr),
\end{equation}
where $\rho_k$ is the representation of $\mathfrak{S}_{\mathrm{pair}}$ on $\Sym^{(2N-k,k)}$ obtained by restriction.

Using the natural identification $\mathrm{End}(V)\cong V\otimes V^*$,
it suffices to consider a decomposable tensor
\begin{align}
T = v\otimes \alpha,\qquad v\in V,\ \alpha\in V^* .
\end{align}
For $x\in V$, we have
\begin{align}
\bigl(\rho_k(\widetilde h)T\rho_k(\widetilde h)^{-1}\bigr)(x)
&=
\rho_k(\widetilde h)
\Bigl(T\bigl(\rho_k(\widetilde h)^{-1}x\bigr)\Bigr) \\
&=
\rho_k(\widetilde h)
\Bigl(\alpha\bigl(\rho_k(\widetilde h)^{-1}x\bigr)v\Bigr) \\
&=
\bigl(\rho_k(\widetilde h^{-1})^*\alpha\bigr)(x)\,
\rho_k(\widetilde h)v .
\end{align}
Thus
\begin{align}
\rho_k(\widetilde h)T\rho_k(\widetilde h)^{-1}
=
\rho_k(\widetilde h)v
\otimes
\rho_k(\widetilde h^{-1})^*\alpha .
\end{align}
Since decomposable tensors span $V\otimes V^*$, the conjugation action corresponds to
\begin{equation}
\rho_k(\widetilde h)\otimes \rho_k(\widetilde h^{-1})^* .
\end{equation}

Here, the superscript $*$ denotes taking the dual map; that is, $\rho_k(\widetilde{h}^{-1})^*$ is the dual map of $\rho_k(\widetilde{h}^{-1})$ acting on $V^*$.
Therefore,
\begin{equation}
\Tr\!\left(
T \mapsto \rho_k(\widetilde{h})\, T\, \rho_k(\widetilde{h})^{-1}
\right)
=
\chi^{(2N-k,k)}(\widetilde{h})\,
\chi^{(2N-k,k)}(\widetilde{h}^{-1})
=
\chi^{(2N-k,k)}(\widetilde{h})\,
\chi^{(2N-k,k)}(\widetilde{h}),
\end{equation}
since the trace of a dual map is equal to the trace of the original map, and $\widetilde{h}$ and $\widetilde{h}^{-1}$ have the same cycle type, so their character values are equal.
Summing over $k$ gives
\begin{equation}
\mathrm{Tr}\!\left(
\left(
M \mapsto \Pi(\widetilde{h}) M \Pi(\widetilde{h})^\dagger
\right)\big|_{\mathcal{A}_{2N}}
\right)
=
\sum_{k=0}^{N}
\chi^{(2N-k,k)}(\widetilde{h})^2.
\end{equation}
Hence
\begin{equation}
\dim\!\bigl(\mathcal{A}_{2N}^{\mathfrak{S}_{\mathrm{pair}}}\bigr)
=
\frac{1}{N!}
\sum_{h \in \mathfrak{S}_N}
\sum_{k=0}^{N}
\chi^{(2N-k,k)}(\widetilde{h})^2.
\end{equation}
Finally, group the sum by cycle type exactly as in the proof of Proposition~\ref{pro:dim_group}.
If $h \in \mathfrak{S}_N$ has cycle type $\lambda \vdash N$, then $\widetilde{h} \in \mathfrak{S}_{\mathrm{pair}} \subseteq \mathfrak{S}_{2N}$
has cycle type $\widetilde{\lambda}$, and the number of such $h$ is $N!/z_\lambda$. Therefore,
\begin{equation}
\dim\!\bigl(\mathcal{A}_{2N}^{\mathfrak{S}_{\mathrm{pair}}}\bigr)
=
\sum_{\lambda \vdash N}
\frac{1}{z_\lambda}
\sum_{k=0}^{N}
\chi^{(2N-k,k)}(\widetilde{\lambda})^2,
\end{equation}
as claimed.
\end{proof}

\subsection{Extension to arbitrary block size} \label{app:extension_to_arbitrary_block}

The construction for pair blocks can be extended to blocks of arbitrary size $b$. Let
\(
W^{(b)} \coloneqq (\mathbb{C}^{2})^{\otimes bN},
\)
and for each $\ell=0,\dots,N-1$, define the $b$-element block
\begin{align}
B_\ell^{(b)} \coloneqq \{b\ell,\, b\ell+1,\, \dots,\, b\ell+b-1\}.
\end{align}
For each $h \in \mathfrak{S}_N$, define its block lift $\widetilde h^{(b)} \in \mathfrak{S}_{bN}$ by
\begin{align}
\widetilde h^{(b)}(b\ell+j)=bh(\ell)+j,
\qquad
(0\le \ell\le N-1,\; 0\le j\le b-1).
\label{eq:block_lift}
\end{align}
Let
\begin{align}
\mathfrak{S}_{\mathrm{block}}^{(b)}
:=
\bigl\{
\widetilde h^{(b)} : h \in \mathfrak{S}_N
\bigr\}
\le \mathfrak{S}_{bN}.
\label{eq:block_subgroup}
\end{align}

Now let the cycle type of $h \in \mathfrak{S}_N$ be
\begin{align}
\lambda
=
1^{m_1(\lambda)}2^{m_2(\lambda)}3^{m_3(\lambda)}\cdots
\vdash N.
\label{eq:block_lambda}
\end{align}
Each $r$-cycle of $h$ gives rise to $b$ disjoint $r$-cycles on the $bN$ indices, so the cycle type of $\widetilde h^{(b)}$ in $\mathfrak{S}_{bN}$ is
\begin{align}
\widetilde{\lambda}^{(b)}
:=
1^{b m_1(\lambda)}2^{b m_2(\lambda)}3^{b m_3(\lambda)}\cdots
\vdash bN.
\label{eq:block_lifted_cycle}
\end{align}

\begin{theorem}
\label{thm:block_dim}
For the block-permuting subgroup $\mathfrak{S}_{\mathrm{block}}^{(b)}$, the following hold:
\begin{align}
\dim\!\Bigl(
\mathbb{C}[\mathfrak{S}_{bN}]^{\mathfrak{S}_{\mathrm{block}}^{(b)}}
\Bigr)
&=
\sum_{\lambda \vdash N}
\prod_{r\ge 1}
r^{(b-1)m_r(\lambda)}
\frac{(b\,m_r(\lambda))!}{m_r(\lambda)!},
\label{eq:block_group_dim}
\end{align}
and
\begin{align}
\dim\!\Bigl(
\mathcal{A}_{bN}^{\mathfrak{S}_{\mathrm{block}}^{(b)}}
\Bigr)
&=
\sum_{\lambda \vdash N}
\frac{1}{z_\lambda}
\sum_{k=0}^{\lfloor bN/2 \rfloor}
\chi^{(bN-k,k)}\!\bigl(\widetilde{\lambda}^{(b)}\bigr)^2,
\label{eq:block_operator_dim}
\end{align}
where $\chi^{(bN-k,k)}$ is the irreducible character of $\mathfrak{S}_{bN}$.
In particular, when $b=2$, these reduce to
Proposition~\ref{pro:dim_group} and Proposition~\ref{thm:dim_opr}.
\end{theorem}

\begin{proof}

The proof is identical to those of Proposition~\ref{pro:dim_group} and
Proposition~\ref{thm:dim_opr}, with $\mathfrak{S}_{\mathrm{pair}}$ replaced by
$\mathfrak{S}_{\mathrm{block}}^{(b)}$.
If $h \in \mathfrak{S}_N$ has cycle type $\lambda$, then $\widetilde h^{(b)}$ has cycle type
$\widetilde{\lambda}^{(b)}$, and hence
\begin{align}
\bigl|C_{\mathfrak{S}_{bN}}(\widetilde h^{(b)})\bigr|
=
\prod_{r\ge 1} r^{b m_r(\lambda)} (b\,m_r(\lambda))!.
\end{align}
Applying Burnside's lemma and grouping by cycle type gives \eqref{eq:block_group_dim}.
Likewise, twirling over $\mathfrak{S}_{\mathrm{block}}^{(b)}$ is the projection onto
$A_{bN}^{\mathfrak{S}_{\mathrm{block}}^{(b)}}$, and the same Schur--Weyl/character argument as before gives \eqref{eq:block_operator_dim}. For readers who have followed the previous two proofs, the argument above already establishes the claim.
For completeness, we now give a full proof.

We prove the two identities in turn. First, since $\mathbb{C}[\mathfrak{S}_{bN}]$ has basis $\mathfrak{S}_{bN}$, conjugation by
$\mathfrak{S}_{\mathrm{block}}^{(b)}$ permutes this basis. Hence
\begin{equation}
\dim\!\bigl(\mathbb{C}[\mathfrak{S}_{bN}]^{\mathfrak{S}_{\mathrm{block}}^{(b)}}\bigr)
\end{equation}
is exactly the number of $\mathfrak{S}_{\mathrm{block}}^{(b)}$-orbits in $\mathfrak{S}_{bN}$
under conjugation. By Burnside's lemma,
\begin{equation}
\dim\!\bigl(\mathbb{C}[\mathfrak{S}_{bN}]^{\mathfrak{S}_{\mathrm{block}}^{(b)}}\bigr)
=
\frac{1}{|\mathfrak{S}_{\mathrm{block}}^{(b)}|}
\sum_{\sigma \in \mathfrak{S}_{\mathrm{block}}^{(b)}}
\bigl|C_{\mathfrak{S}_{bN}}(\sigma)\bigr|,
\end{equation}
where
\begin{equation}
C_{\mathfrak{S}_{bN}}(\sigma)
=
\{\pi \in \mathfrak{S}_{bN} : \sigma \pi \sigma^{-1} = \pi\}
\end{equation}
is the centralizer of $\sigma$. Using the identification $\mathfrak{S}_{\mathrm{block}}^{(b)} \cong \mathfrak{S}_N$, this becomes
\begin{equation}
\dim\!\bigl(\mathbb{C}[\mathfrak{S}_{bN}]^{\mathfrak{S}_{\mathrm{block}}^{(b)}}\bigr)
=
\frac{1}{N!}
\sum_{h \in \mathfrak{S}_N}
\bigl|C_{\mathfrak{S}_{bN}}(\widetilde{h}^{(b)})\bigr|.
\end{equation}
For $h \in \mathfrak{S}_N$ of cycle type $\lambda \vdash N$, the corresponding block lift
$\widetilde{h}^{(b)}$ has cycle type
\begin{equation}
\widetilde{\lambda}^{(b)}
=
1^{b m_1(\lambda)}
2^{b m_2(\lambda)}
3^{b m_3(\lambda)}
\cdots \vdash bN.
\end{equation}
Therefore, since a permutation of cycle type $1^{a_1}2^{a_2}3^{a_3}\cdots$ has centralizer size
\begin{equation}
\prod_{r \ge 1} r^{a_r} a_r!,
\end{equation}
we obtain
\begin{equation}
\bigl|C_{\mathfrak{S}_{bN}}(\widetilde{h}^{(b)})\bigr|
=
\prod_{r \ge 1}
r^{\,b m_r(\lambda)} (b\,m_r(\lambda))!.
\end{equation} 
The number of elements of $\mathfrak{S}_N$ of cycle type $\lambda$ is
\begin{equation}
\frac{N!}{z_\lambda},
\qquad
z_\lambda := \prod_{r \ge 1} r^{\,m_r(\lambda)} m_r(\lambda)!.
\end{equation}
Grouping the Burnside sum by cycle type, we obtain
\begin{equation}
\dim\!\bigl(\mathbb{C}[\mathfrak{S}_{bN}]^{\mathfrak{S}_{\mathrm{block}}^{(b)}}\bigr)
=
\sum_{\lambda \vdash N}
\frac{1}{z_\lambda}
\prod_{r \ge 1}
r^{\,b m_r(\lambda)} (b\,m_r(\lambda))!.
\end{equation}
Substituting the expression for $z_\lambda$ and simplifying yields
\begin{equation}
\dim\!\bigl(\mathbb{C}[\mathfrak{S}_{bN}]^{\mathfrak{S}_{\mathrm{block}}^{(b)}}\bigr)
=
\sum_{\lambda \vdash N}
\prod_{r \ge 1}
r^{(b-1)m_r(\lambda)}
\frac{(b\,m_r(\lambda))!}{m_r(\lambda)!},
\end{equation}
which proves the first identity.

For the second identity, let
\begin{equation}
W^{(b)} := (\mathbb{C}^2)^{\otimes bN}
\qquad \text{and} \qquad
\mathcal{A}_{bN} := \Pi\bigl(\mathbb{C}[\mathfrak{S}_{bN}]\bigr) \subseteq \mathrm{End}(W^{(b)}).
\end{equation}
By Proposition~\ref{pro:twir_equal},
\begin{equation}
\mathcal{T}_{\mathfrak{S}_{\mathrm{block}}^{(b)}}[\mathcal{A}_{bN}]
=
\mathcal{A}_{bN}^{\mathfrak{S}_{\mathrm{block}}^{(b)}}.
\end{equation}
Since $\mathcal{T}_{\mathfrak{S}_{\mathrm{block}}^{(b)}}$ is the averaging operator over a finite group, it is an idempotent projection onto $\mathcal{A}_{bN}^{\mathfrak{S}_{\mathrm{block}}^{(b)}}$. Therefore,
\begin{equation}
\dim\!\bigl(\mathcal{A}_{bN}^{\mathfrak{S}_{\mathrm{block}}^{(b)}}\bigr)
=
\Tr\!\left(
\mathcal{T}_{\mathfrak{S}_{\mathrm{block}}^{(b)}}\big|_{\mathcal{A}_{bN}}
\right).
\end{equation}
Using the identification $\mathfrak{S}_{\mathrm{block}}^{(b)} \cong \mathfrak{S}_N$, we obtain
\begin{equation}
\dim\!\bigl(\mathcal{A}_{bN}^{\mathfrak{S}_{\mathrm{block}}^{(b)}}\bigr)
=
\frac{1}{N!}
\sum_{h \in \mathfrak{S}_N}
\Tr\!\left(
\left(
M \mapsto \Pi(\widetilde{h}^{(b)}) M \Pi(\widetilde{h}^{(b)})^\dagger
\right)\big|_{\mathcal{A}_{bN}}
\right).
\end{equation}
By Schur--Weyl duality,
\begin{equation}
W^{(b)}
\cong
\bigoplus_{k=0}^{\lfloor bN/2 \rfloor}
J^{(bN/2-k)}(\mathbb{C}^2)\otimes \Sym^{(bN-k,k)},
\end{equation}
and hence
\begin{equation}
\mathcal{A}_{bN}
=
\Pi\bigl(\mathbb{C}[\mathfrak{S}_{bN}]\bigr)
\cong
\bigoplus_{k=0}^{\lfloor bN/2 \rfloor}
\mathrm{End}\!\bigl(\Sym^{(bN-k,k)}\bigr).
\end{equation}
On the $k$-th summand, the map
\begin{equation}
M \mapsto \Pi(\widetilde{h}^{(b)}) M \Pi(\widetilde{h}^{(b)})^\dagger
\end{equation}
acts as
\begin{equation}
T \longmapsto \rho_k(\widetilde{h}^{(b)})\, T\, \rho_k(\widetilde{h}^{(b)})^{-1},
\qquad
T \in \mathrm{End}\!\bigl(\Sym^{(bN-k,k)}\bigr),
\end{equation}
where $\rho_k$ is the representation of $\mathfrak{S}_{\mathrm{block}}^{(b)}$ on $\Sym^{(bN-k,k)}$ obtained by restriction.

Using the natural identification $\mathrm{End}(V)\cong V\otimes V^*$,
it suffices to consider a decomposable tensor
\begin{align}
T = v\otimes \alpha,\qquad v\in V,\ \alpha\in V^* .
\end{align}
For $x\in V$, we have
\begin{align}
\bigl(\rho_k(\widetilde h^{(b)})T\rho_k(\widetilde h^{(b)})^{-1}\bigr)(x)
&=
\rho_k(\widetilde h^{(b)})
\Bigl(T\bigl(\rho_k(\widetilde h^{(b)})^{-1}x\bigr)\Bigr) \\
&=
\rho_k(\widetilde h^{(b)})
\Bigl(\alpha\bigl(\rho_k(\widetilde h^{(b)})^{-1}x\bigr)v\Bigr) \\
&=
\bigl(\rho_k((\widetilde h^{(b)})^{-1})^*\alpha\bigr)(x)\,
\rho_k(\widetilde h^{(b)})v .
\end{align}
Thus
\begin{align}
\rho_k(\widetilde h^{(b)})T\rho_k(\widetilde h^{(b)})^{-1}
=
\rho_k(\widetilde h^{(b)})v
\otimes
\rho_k((\widetilde h^{(b)})^{-1})^*\alpha .
\end{align}
Since decomposable tensors span $V\otimes V^*$, the conjugation action corresponds to
\begin{equation}
\rho_k(\widetilde h^{(b)})\otimes \rho_k((\widetilde h^{(b)})^{-1})^* .
\end{equation}

Here, the superscript $*$ denotes taking the dual map; that is, $\rho_k\!\bigl((\widetilde{h}^{(b)})^{-1}\bigr)^*$ is the dual map of $\rho_k\!\bigl((\widetilde{h}^{(b)})^{-1}\bigr)$ acting on $V^*$. Therefore,
\begin{equation}
\Tr\!\left(
T \mapsto \rho_k(\widetilde{h}^{(b)})\, T\, \rho_k(\widetilde{h}^{(b)})^{-1}
\right)
=
\chi^{(bN-k,k)}\!\bigl(\widetilde{h}^{(b)}\bigr)\,
\chi^{(bN-k,k)}\!\bigl(\widetilde{h}^{(b)}\bigr),
\end{equation} 
since the trace of a dual map is equal to the trace of the original map, and $\widetilde{h}^{(b)}$ and $(\widetilde{h}^{(b)})^{-1}$ have the same cycle type, so their character values are equal.

Summing over $k$ gives
\begin{equation}
\Tr\!\left(
\left(
M \mapsto \Pi(\widetilde{h}^{(b)}) M \Pi(\widetilde{h}^{(b)})^\dagger
\right)\big|_{\mathcal{A}_{bN}}
\right)
=
\sum_{k=0}^{\lfloor bN/2 \rfloor}
\chi^{(bN-k,k)}\!\bigl(\widetilde{h}^{(b)}\bigr)^2.
\end{equation}
Hence
\begin{equation}
\dim\!\bigl(\mathcal{A}_{bN}^{\mathfrak{S}_{\mathrm{block}}^{(b)}}\bigr)
=
\frac{1}{N!}
\sum_{h \in \mathfrak{S}_N}
\sum_{k=0}^{\lfloor bN/2 \rfloor}
\chi^{(bN-k,k)}\!\bigl(\widetilde{h}^{(b)}\bigr)^2.
\end{equation}
Finally, group the sum by cycle type exactly as above. If $h \in \mathfrak{S}_N$ has cycle type $\lambda \vdash N$, then $\widetilde{h}^{(b)} \in \mathfrak{S}_{\mathrm{block}}^{(b)} \subseteq \mathfrak{S}_{bN}$ has cycle type $\widetilde{\lambda}^{(b)}$, and the number of such $h$ is $N!/z_\lambda$. Therefore,
\begin{equation}
\dim\!\bigl(\mathcal{A}_{bN}^{\mathfrak{S}_{\mathrm{block}}^{(b)}}\bigr)
=
\sum_{\lambda \vdash N}
\frac{1}{z_\lambda}
\sum_{k=0}^{\lfloor bN/2 \rfloor}
\chi^{(bN-k,k)}\!\bigl(\widetilde{\lambda}^{(b)}\bigr)^2,
\end{equation}
as claimed.
\end{proof}

Thus, Theorem~\ref{thm:block_dim} extends the pair-block dimension formulas to arbitrary block size $b$. 

The values in Table~\ref{tab:block-operator-dim} show that, whenever \(N>1\), the dimension of the admissible invariant operator space grows rapidly with $b$. Notably, even the minimal nontrivial choice $b=2$ yields a substantially larger invariant operator space than the standard case $b=1$, and the gap widens as $N$ increases. Consequently, this construction circumvents the representation-theoretic obstruction at the operator level and thereby opens a concrete route to dual-equivariant QML models.

\begin{table}[htbp]
\centering
\small
\setlength{\tabcolsep}{6pt}
\renewcommand{\arraystretch}{1}
\begin{tabular}{c|ccccc}
\hline
& \multicolumn{5}{c}{Block size $b$} \\
\cline{2-6}
$N$ & $1$ & $2$ & $3$ & $4$ & $5$ \\
\hline
$2$ & $2$ & $10$ & $76$ & $750$ & $8524$ \\
$3$ & $2$ & $26$ & $834$ & $3.49\times 10^{4}$ & $1.62\times 10^{6}$ \\
$4$ & $3$ & $84$ & $9226$ & $1.49\times 10^{6}$ & $2.74\times 10^{8}$ \\
$5$ & $3$ & $206$ & $8.85\times 10^{4}$ & $5.60\times 10^{7}$ & $4.08\times 10^{10}$ \\
\hline
\end{tabular}
\caption{\textbf{Dual-Equivariant Operator Space Dimension.} For small $N$ and $b$, we list the dimension of dual-equivariant operators, \(\dim \mathcal{A}_{bN}^{\mathfrak{S}_{\mathrm{block}}^{(b)}}\).}
\label{tab:block-operator-dim}
\end{table}

\section{Quantum Network}
\label{app:quantum-network}

\subsection{Singlet state}

\subsubsection{Singlet state definition}
The two-qubit singlet
\begin{equation}
\ket{\psi^-}\;:=\;\frac{\ket{01}-\ket{10}}{\sqrt{2}}
\end{equation}
is the unique spin-$0$ ($j=0$) irreducible component inside
\begin{equation}
\tfrac{1}{2}\!\otimes\!\tfrac{1}{2}=0\oplus 1.
\end{equation}
Equivalently, it spans the one-dimensional $j=0$ irrep.

\subsubsection{$\mathrm{SU}(2)$-invariance.}
Let $\rho(U)=U^{\otimes 2}$ denote the diagonal action of $U\in\mathrm{SU}(2)$ on two qubits.
Then
\begin{equation}
\rho(U)\ket{\psi^-}=\ket{\psi^-}\qquad(\forall\,U\in\mathrm{SU}(2)).
\end{equation}

\emph{Proof.}
Let $T=\mathrm{SWAP}$ on $\mathbb{C}^2\otimes\mathbb{C}^2$. Since $T$ is an involution ($T^2=I$) with eigenvalues $\pm1$,
we have the eigenspace decomposition
\begin{equation}
\mathbb{C}^2\otimes\mathbb{C}^2=\mathrm{Sym}^2(\mathbb{C}^2)\oplus \wedge^2(\mathbb{C}^2),
\qquad \dim\mathrm{Sym}^2(\mathbb{C}^2)=3,\ \dim\wedge^2(\mathbb{C}^2)=1,
\end{equation}
where $E_{+1}(T)=\mathrm{Sym}^2(\mathbb{C}^2)$ and $E_{-1}(T)=\wedge^2(\mathbb{C}^2)$

The vector $\ket{\psi^-}$ is antisymmetric under the qubit swap (i.e., $\mathrm{SWAP}\,\ket{\psi^-}=-\ket{\psi^-}$),
so $\ket{\psi^-}\in\wedge^2(\mathbb{C}^2)$; since $\dim\wedge^2(\mathbb{C}^2)=1$, it follows that
\begin{equation}
\wedge^2(\mathbb{C}^2)=\mathrm{span}\{\ket{\psi^-}\}.
\end{equation}

Now consider $U\in \mathrm{SU}(2)$ and the diagonal action $\rho(U)=U^{\otimes 2}$.
Since $\mathrm{SWAP}$ permutes tensor factors, it commutes with $U\otimes U$, and hence $\wedge^2(\mathbb{C}^2)$ is $\rho(U)$-invariant, up to a phase factor.
We can further show that the phase factor must be $1$ from the following argument.
Since $\dim\wedge^2(\mathbb{C}^2)=1$, there exists a phase $\chi(U)\in \mathrm{U}(1)$ such that
$(U\otimes U)\ket{\psi^-}=\chi(U)\ket{\psi^-}$.
But the induced action on $\wedge^2(\mathbb{C}^2)$ is $\wedge^2 U=\det(U)\,I$, so $\chi(U)=\det(U)=1$ for $U\in \mathrm{SU}(2)$; hence $(U\otimes U)\ket{\psi^-}=\ket{\psi^-}$.

\subsection{Geometric encoding} \label{app:geometric_encoding}

\subsubsection{From $\mathrm{SU}(2)$ to 3D rotations.}
As shown in Appendix~\ref{app:su2_group}, unitary conjugation by $U_{R}\!\in\!\mathrm{SU}(2)$ induces a 3D rotation
on Pauli vectors:
\begin{equation}\label{eq:U-Pauli-rotation}
U_{R}(\vec v\!\cdot\!\vec\sigma)U_{R}^\dagger
=\sum_{i,j}R_{ij}(U_{R})v_j\,\sigma_i
=(R(U_{R})\vec v)\!\cdot\!\vec\sigma,
\end{equation}
where $\Phi:U\mapsto R(U)$ is a group homomorphism $\mathrm{SU}(2)\!\to\!\mathrm{SO}(3)$.
This lets us convert rotations of 3D points into unitary transformations of qubits.
\subsubsection{Geometric encoding's rotation equivariance.}
Let $\Phi:\mathrm{SU}(2)\!\to\!\mathrm{SO}(3)$ be the covering homomorphism. 
For any $R\!\in\!\mathrm{SO}(3)$ pick $U_{R}\!\in\!\mathrm{SU}(2)$ with $\Phi(U_{R})=R$. 
Since $\|{\mathbf p}\|$ is rotation-invariant and $R\mathbf p=\mathbf p'$, we have
\begin{align}
E(\mathbf p')
&=\exp\!\big(i\,\frac{\|\mathbf{p}\|}{\Theta}\,(R\hat{\mathbf p})\!\cdot\!\vec\sigma\big) \\
&=\exp\!\big(i\,\frac{\|\mathbf{p}\|}{\Theta}\,U_{R}(\hat{\mathbf p}\!\cdot\!\vec\sigma)U_{R}^\dagger\big) 
\quad \text{(by \eqref{eq:U-Pauli-rotation})}
\\
&=U_{R}\,\exp\!\big(i\,\frac{\|\mathbf{p}\|}{\Theta}\,\hat{\mathbf p}\!\cdot\!\vec\sigma\big)\,U_{R}^\dagger \\
&=U_{R}\,E(\mathbf p)\,U_{R}^\dagger.
\end{align} 

Hence the encoder is exactly rotation equivariant (via $\mathrm{SU}(2)$ adjoint action):
\begin{equation}\label{eq:equivariance-Ep}
E(R\mathbf{p}) \;=\; U_{R}\,E(\mathbf{p})\,U_{R}^\dagger,\qquad \Phi(U_{R})=R.
\end{equation}

\subsubsection{Quantum gate expression.}
For practical quantum circuit implementation, encodings must be decomposed into elementary single qubit gates supported by current hardware. We encode each 3D point $\mathbf p_i\in\mathbb R^3$ by the unitary
\begin{equation}\label{eq:exp-encode}
E(\mathbf p_i)
=\exp\!\left(i\,\frac{\|\mathbf p_i\|}{\Theta}\,\hat{\mathbf p}_i\!\cdot\!\vec{\sigma}\right)
=R_Z(\alpha_i)\,R_Y(\beta_i)\,R_Z(\gamma_i),
\end{equation}
i.e., we exactly realize the exponential as a Z–Y–Z sequence on each encoding qubit $i$.
Here $\Theta>0$ is a fixed geometric scale, $\hat{\mathbf p}_i=\mathbf p_i/\|\mathbf p_i\|$ and $\vec{\sigma}=(\sigma_X,\sigma_Y,\sigma_Z)$.

Given the preprocessed point cloud, we first normalize by $\Theta$:
\begin{equation}
\tilde{\mathbf p}_i=\mathbf p_i/\Theta,\qquad
\phi_i=\|\tilde{\mathbf p}_i\|,\qquad
\mathbf n_i=\tilde{\mathbf p}_i/\|\tilde{\mathbf p}_i\|=(n_x,n_y,n_z).
\end{equation}
The Z–Y–Z angles are
\begin{align}
\alpha_i \;&=\; \arctan\!\big(-\,n_z \tan\phi_i\big)\;+\;\arctan\!\!\left(\frac{-\,n_x}{\,n_y}\right),\\[2pt]
\gamma_i \;&=\; \arctan\!\big(-\,n_z \tan\phi_i\big)\;-\;\arctan\!\!\left(\frac{-\,n_x}{\,n_y}\right),\\[2pt]
\beta_i \;&=\; 2\,\arcsin\!\left(\frac{\sin\phi_i\; n_x}{\sin\!\big((\alpha_i-\gamma_i)/2\big)}\right).
\end{align}

In practice, we use $\operatorname{arctan2}$ for robust angle-branch (quadrant) selection, which remains well-defined even when $n_y=0$.
Axis-aligned degenerate cases (e.g., $n_x=n_y=0$) are handled separately.

\subsubsection{On the scale $\Theta$.}
The angle formulas involve $\tan(\phi_i)$ and an $\arcsin(\cdot)$ in $\beta_i$, both of which are numerically sensitive: $\tan(\cdot)$ becomes ill-conditioned near $(k+\tfrac{1}{2})\pi$, and $\arcsin(\cdot)$ requires inputs in $[-1,1]$ with margins to avoid floating-point clipping. To improve stability, we set the geometric scale $\Theta$ larger than the maximum point cloud radius, so that $\phi_i=\|\mathbf p_i\|/\Theta$ stays well within numerically benign regions and the $\arcsin$ argument remains bounded away from $\pm1$. We then performed hyperparameter tuning over several $\Theta$ values under this constraint and found that a single fixed $\Theta = 1.7$ works well for most instances, which we also used for all reported experiments.

\subsubsection{Pair selective encoding and permutation.}
To obtain a pairwise permutation-equivariant quantum part while preserving the singlet structure, we encode
one qubit per pair $(2j,2j{+}1)$:
\begin{equation}\label{eq:pair-select-enc}
\ket{\Psi_{\mathrm{enc}}}
=\bigotimes_{j=0}^{N-1}\!\big(E(\mathbf p_{j})\otimes I\big)\,\ket{\psi^-}_{(2j,\,2j+1)},
\qquad
\ket{\psi^-}=\tfrac{1}{\sqrt{2}}(\ket{01}-\ket{10}).
\end{equation}
Here $E(\mathbf p)$ is the single–qubit encoding unitary.

We used per-pair encoding to make the effect of permutation consistent for any $\sigma \in \Sym_N$.
If two distinct points are fed to one pair,
\begin{equation}
\ket{\Psi(\mathbf p_a,\mathbf p_b)}
=\big(E(\mathbf p_a)\!\otimes\!E(\mathbf p_b)\big)\ket{\psi^-},
\end{equation}
then swapping the two inputs gives
\begin{equation}
S\,\ket{\Psi(\mathbf p_a,\mathbf p_b)}
=\big(E(\mathbf p_b)\!\otimes\!E(\mathbf p_a)\big)\,S\ket{\psi^-}
= -\,\big(E(\mathbf p_b)\!\otimes\!E(\mathbf p_a)\big)\ket{\psi^-}
= -\,\ket{\Psi(\mathbf p_b,\mathbf p_a)},
\end{equation}
which implies that the within-pair wire swap changes only the global phase, and since the global phase does not affect any physical outcomes, this effect cannot be observed.
In contrast, consider four qubits initialized as two singlet pairs,
$\ket{\psi^-}_{01}\otimes\ket{\psi^-}_{23}$, and let $\Pi_{(0,3)}$ denote the SWAP between wires $0$ and $3$.
Define $E_i:=E(\mathbf p_i)$ and
\begin{equation}\label{eq:psi-enc-4q}
\ket{\Psi(\mathbf p_0,\mathbf p_1,\mathbf p_2,\mathbf p_3)}
:= (E_0\!\otimes\!E_1)\ket{\psi^-}_{01}\;\otimes\;(E_2\!\otimes\!E_3)\ket{\psi^-}_{23}.
\end{equation}

For the point-permutation $\sigma=(0\,3)$ (acting on inputs only),
\begin{equation}\label{eq:psi-point-perm}
\ket{\Psi(\mathbf p_3,\mathbf p_1,\mathbf p_2,\mathbf p_0)}
= (E_3\!\otimes\!E_1)\ket{\psi^-}_{01}\;\otimes\;(E_2\!\otimes\!E_0)\ket{\psi^-}_{23}.
\end{equation}
This state is product across $(01)\,|\,(23)$, hence the reduced state on wires $(0,1)$ is pure:
\begin{equation}\label{eq:rho-pi-rank}
\rho^{(\sigma)}_{01}
:=\Tr_{23}\!\Big[\ket{\Psi(\mathbf p_3,\mathbf p_1,\mathbf p_2,\mathbf p_0)}\bra{\Psi(\mathbf p_3,\mathbf p_1,\mathbf p_2,\mathbf p_0)}\Big],
\qquad
\rank\!\big(\rho^{(\sigma)}_{01}\big)=1.
\end{equation}

Now let $\Pi_{(0,3)}$ swap wires $0$ and $3$. Then
\begin{equation}\label{eq:wire-swap-state}
\Pi_{(0,3)}\ket{\Psi(\mathbf p_0,\mathbf p_1,\mathbf p_2,\mathbf p_3)}
=(E_3\!\otimes\!E_1\!\otimes\!E_2\!\otimes\!E_0)\,
\big(\ket{\psi^-}_{31}\otimes\ket{\psi^-}_{20}\big).
\end{equation}
Taking the reduced state on wires $(0,1)$ and using $\Tr_{q}\ket{\psi^-}\bra{\psi^-}=\tfrac{I}{2}$ gives
\begin{equation}\label{eq:rho-swap-rank}
\rho^{(\Pi_{(0,3)})}_{01}
:=\Tr_{23}\!\Big[\Pi_{(0,3)}\ket{\Psi(\mathbf p_0,\mathbf p_1,\mathbf p_2,\mathbf p_3)}\bra{\Psi(\mathbf p_0,\mathbf p_1,\mathbf p_2,\mathbf p_3)}\Pi_{(0,3)}^\dagger\Big]
=\frac{I_{01}}{4},
\qquad
\rank\!\big(\rho^{(\Pi_{(0,3)})}_{01}\big)=4,
\end{equation}
Therefore, the two states cannot be proportional (even up to a global phase):
\begin{equation}\label{eq:not-proportional}
\Pi_{(0,3)}\ket{\Psi(\mathbf p_0,\mathbf p_1,\mathbf p_2,\mathbf p_3)}
\not\propto
\ket{\Psi(\mathbf p_3,\mathbf p_1,\mathbf p_2,\mathbf p_0)},
\end{equation}
Thus, a two-point-per-pair encoding is not necessarily permutation equivariant.
The pair selective scheme \eqref{eq:pair-select-enc} is the  effective compromise used in our experiments:
with $n$ qubits arranged into $n/2$ singlet pairs, one can encode $n/2$ points while maintaining permutation equivariance.

\subsection{Our quantum network via twirling}
\label{app:quantum_network_via_twirling}
\subsubsection{Derivation of the pair-symmetrized blocks $P_k^{\pm}$.}

We recall the notations from Sec.~\ref{sec:our_nwk}. We denote by $\Pi$ the permutation representation of $\Sym_n$ on the $n$ wires.
For any finite subgroup $H \subseteq \Sym_n$, we define the (operator) twirl on
$M \in \mathbb{C}^{2^n\times 2^n}$ by 
\begin{equation}
\mathcal{T}_H[M]
:= \frac{1}{|H|}\sum_{h\in H} \Pi(h)\,M\,\Pi(h)^\dagger.
\label{eq:twirl-def-appendix}
\end{equation}
In particular, in our quantum circuit architecture, we use $H = \Sym_{\mathrm{pair}}$.

Fix $2\le k\le N$ and consider all ordered selections of $k$ distinct pairs
from the $N$ available pairs. Let
\begin{equation}
\mathcal{P}=\{(0,1),(2,3),\ldots,(n-2,n-1)\}
\end{equation}
be the set of pairs, and write $\mathrm{Perm}(\mathcal{P},k)$ for the set of
ordered $k$–tuples.
For a given $\pi=(p_1,\cdots,p_k)\in\mathrm{Perm}(\mathcal{P},k)$ where
$p_{m}=(2j_m,2j_{m}+1)$ and a selection vector
$\mathbf{s}=(s_1,\dots,s_k)\in\{0,1\}^k$, we define $\tau_{\pi}^{\mathbf{s}}\in \mathrm{GL}(\mathcal H)\subset \mathbb{C}^{2^n\times 2^n}$ to be the permutation matrix representing the $k$--cycle
\begin{equation}
(2j_1 + s_1,\ 2j_2 + s_2,\ \ldots,\ 2j_k + s_k).
\end{equation}
on the selected wires and leaves all other wires fixed. Intuitively, for each
pair $B_{i_\ell}$ we choose either its first wire ($s_\ell=0$) or second wire
($s_\ell=1$) to participate in the cycle, while the partner wire in the same
pair stays put.

For a fixed interaction order $k$ and an ordered $k$–tuple of pairs
$\pi\in\mathrm{Perm}(\mathcal{P},k)$, we define a class of generalized $k$–cycle generators supported on
these pairs by 
\begin{equation}
\label{eq:general-weighted-combo}
M_{k,\pi}^{(w)}
:= \sum_{\mathbf{s}\in\{0,1\}^k} w(\mathbf{s})\,\tau_{\pi}^{\mathbf{s}},
\end{equation}
where $w:\{0,1\}^k\to\mathbb{C}$ determines the coefficient for $2^k$ choices of within–pair selections. 
Our goal is to endow these already rotation-equivariant operators with $\Sym_{\mathrm{pair}}$-equivariance, which can be achieved in principle by twirling a generator in $\Pi(\mathbb{C}[\Sym_n])$ under $\Sym_{\mathrm{pair}}$.

While twirling can be applied to arbitrary operators, we restrict the operators to those with the coefficients defined as
\begin{equation}
\label{eq:weight-two-patterns}
w^{(+)}(\mathbf{s}) := 1,
\qquad
w^{(-)}(\mathbf{s}) := (-1)^{\|\mathbf{s}\|_1},
\end{equation}
which are natural choices from the structure of $\Sym_{\mathrm{pair}}$, where $\|\mathbf{s}\|_1$ is the Hamming weight of $\mathbf{s}$.
The corresponding operators are given as
\begin{equation}
\label{eq:Mpiplusminus-def}
M_{k,\pi}^{(+)}
:= \sum_{\mathbf{s}\in\{0,1\}^k}\tau_{\pi}^{\mathbf{s}},
\qquad
M_{k,\pi}^{(-)}
:= \sum_{\mathbf{s}\in\{0,1\}^k}(-1)^{\|\mathbf{s}\|_1}\,\tau_{\pi}^{\mathbf{s}}.
\end{equation}

We now compute the pair–permuting twirl of these two special combinations and
show that they are proportional to $P_k^+$ and $P_k^-$, respectively.

\emph{Uniform pattern.}
Using the definition of the pair–permuting twirl and the conjugation rule
$\Pi(\sigma)\,\tau_{\pi}^{\mathbf{s}}\,\Pi(\sigma)^\dagger
= \tau_{\sigma\!\cdot\!\pi}^{\mathbf{s}}$ for $\sigma\in \Sym_{\mathrm{pair}}$, we get
\begin{equation}
\label{eq:twirl-uniform-expansion}
\begin{aligned}
\mathcal{T}_{\Sym_{\mathrm{pair}}}\big[M_{k,\pi}^{(+)}\big]
&= \frac{1}{|\Sym_{\mathrm{pair}}|}\sum_{\sigma\in \Sym_{\mathrm{pair}}}
   \Pi(\sigma)\Big(\sum_{\mathbf{s}}\tau_{\pi}^{\mathbf{s}}\Big)\Pi(\sigma)^\dagger \\[2pt]
&= \frac{1}{|\Sym_{\mathrm{pair}}|}\sum_{\sigma\in \Sym_{\mathrm{pair}}}
   \sum_{\mathbf{s}}\tau_{\sigma\!\cdot\!\pi}^{\mathbf{s}}.
\end{aligned}
\end{equation}
The induced action of $S_{\mathrm{pair}}\cong S_N$ on ordered $k$--tuples of
distinct pairs is transitive: for any
$\pi,\pi' \in \mathrm{Perm}(\mathcal{P},k)$ there exists
$\sigma \in \Sym_{\mathrm{pair}}$ such that $\sigma\!\cdot\!\pi = \pi'$.
For a fixed $\pi$, let 
\begin{equation}
\mathrm{Stab}(\pi)
:= \{\sigma \in \Sym_{\mathrm{pair}} : \sigma\!\cdot\!\pi = \pi\}
\end{equation}
be its stabilizer. In our setting, keeping the ordered list of $k$ active
pairs $(B_{i_1},\dots,B_{i_k})$ fixed forces $\sigma$ to permute only the
remaining $N-k$ inactive pairs, so $|\mathrm{Stab}(\pi)| = (N-k)!$ depends
only on $N$ and $k$, but not on the particular choice of $\pi$. By the
orbit--stabilizer theorem,
\begin{equation}
|\Sym_{\mathrm{pair}}|
= |\mathrm{Orbit}(\pi)| \cdot |\mathrm{Stab}(\pi)|
= |\mathrm{Perm}(\mathcal{P},k)| \cdot |\mathrm{Stab}(\pi)|,
\end{equation}
and for every $\pi' \in \mathrm{Perm}(\mathcal{P},k)$ the number of
$\sigma \in \Sym_{\mathrm{pair}}$ with $\sigma\!\cdot\!\pi = \pi'$ is exactly
$|\mathrm{Stab}(\pi)|$. Thus each $\pi'$ appears equally often in the multiset
$\{\sigma\!\cdot\!\pi : \sigma \in \Sym_{\mathrm{pair}}\}$, and the average over
$\sigma$ is proportional to the uniform average over all ordered $k$--tuples: 
\begin{equation}
\frac{1}{|\Sym_{\mathrm{pair}}|}\sum_{\sigma\in \Sym_{\mathrm{pair}}}
\sum_{\mathbf{s}\in\{0,1\}^k}
\,\tau_{\sigma\!\cdot\!\pi}^{\mathbf{s}}
\;\propto\;
\frac{1}{k!}\sum_{\pi'\in\mathrm{Perm}(\mathcal{P},k)}
\sum_{\mathbf{s}\in\{0,1\}^k}
\tau_{\pi'}^{\mathbf{s}}.
\end{equation}

In particular, for the uniform pattern
$M_{k,\pi}^{(+)} := \sum_{\mathbf{s}\in\{0,1\}^k}\tau_{\pi}^{\mathbf{s}}$, we obtain
\begin{equation}
\label{eq:twirl-uniform-to-Pkplus-final}
\mathcal{T}_{\Sym_{\mathrm{pair}}}\big[M_{k,\pi}^{(+)}\big]
\;\propto\;
\frac{1}{k!}\sum_{\pi'\in\mathrm{Perm}(\mathcal{P},k)}
\sum_{\mathbf{s}\in\{0,1\}^k}\tau_{\pi'}^{\mathbf{s}}
\;=\; P_k^+.
\end{equation}
\emph{Parity pattern.}
The computation for the parity–weighted pattern is identical, except that the
factor $(-1)^{\|\mathbf{s}\|_1}$ does not depend on the pair indices and is
therefore unchanged by the $\Sym_{\mathrm{pair}}$ action:
\begin{equation}
\label{eq:twirl-parity-expansion}
\begin{aligned}
\mathcal{T}_{\Sym_{\mathrm{pair}}}\big[M_{k,\pi}^{(-)}\big]
&= \frac{1}{|\Sym_{\mathrm{pair}}|}\sum_{\sigma\in \Sym_{\mathrm{pair}}}
   \Pi(\sigma)\Big(\sum_{\mathbf{s}}(-1)^{\|\mathbf{s}\|_1}
   \tau_{\pi}^{\mathbf{s}}\Big)\Pi(\sigma)^\dagger \\[2pt]
&= \frac{1}{|\Sym_{\mathrm{pair}}|}\sum_{\sigma\in \Sym_{\mathrm{pair}}}
   \sum_{\mathbf{s}}(-1)^{\|\mathbf{s}\|_1}\tau_{\sigma\!\cdot\!\pi}^{\mathbf{s}} \\[2pt]
&\;\propto\;
\frac{1}{k!}\sum_{\pi'\in\mathrm{Perm}(\mathcal{P},k)}
\sum_{\mathbf{s}\in\{0,1\}^k}(-1)^{\|\mathbf{s}\|_1}\tau_{\pi'}^{\mathbf{s}}
\;=\; P_k^-.
\end{aligned}
\end{equation}

\subsubsection{Unitarity of the quantum network.}
On a quantum device, every implemented gate must be unitary.
In particular, each phase factor of the form $e^{i c_k^\pm P_k^\pm}$ used in our network must be unitary.
We verify this by showing that the operators $P_k^\pm$ are Hermitian, and then invoking a standard fact
about exponentials of Hermitian operators.

By definition in the main text,
$P_k^+$ and $P_k^-$ are (up to normalization) finite linear combinations of the permutation matrices
$\tau_{\pi'}^{\mathbf t}$:
\begin{equation}
P_k^+ \;\propto\; \sum_{\pi'\in\mathrm{Perm}(\mathcal{P},k)} \ \sum_{\mathbf t\in\{0,1\}^k} \tau_{\pi'}^{\mathbf t},
\qquad
P_k^- \;\propto\; \sum_{\pi'\in\mathrm{Perm}(\mathcal{P},k)} \ \sum_{\mathbf t\in\{0,1\}^k} (-1)^{\|\mathbf t\|_1}\,\tau_{\pi'}^{\mathbf t}.
\end{equation}
Each $\tau_{\pi'}^{\mathbf t}$ is a permutation matrix and hence unitary, so
\begin{equation}
\big(\tau_{\pi'}^{\mathbf t}\big)^\dagger
= \big(\tau_{\pi'}^{\mathbf t}\big)^{-1}.
\end{equation}
The inverse of the $k$–cycle implemented by $\tau_{\pi'}^{\mathbf t}$ is the same cycle with the
order of its $k$ support wires reversed, which is again of the form $\tau_{\pi''}^{\mathbf t'}$ for
some ordered $k$–tuple $\pi''$. Since we sum over all ordered $k$–tuples of distinct pairs,
the index set $\{(\pi',\mathbf t)\}$ is closed under inversion: for every $(\pi',\mathbf t)$ there
is $(\pi'',\mathbf t')$ with
\begin{equation}
\big(\tau_{\pi'}^{\mathbf t}\big)^{-1} = \tau_{\pi''}^{\mathbf t'}.
\end{equation}

Using this closure under inversion, we obtain
\begin{equation}
\begin{aligned}
(P_k^+)^\dagger
&\;\propto\; \sum_{\pi',\mathbf t} \big(\tau_{\pi'}^{\mathbf t}\big)^\dagger
= \sum_{\pi',\mathbf t} \big(\tau_{\pi'}^{\mathbf t}\big)^{-1}
= \sum_{\pi'',\mathbf t'} \tau_{\pi''}^{\mathbf t'}
\;\propto\; P_k^+,
\end{aligned}
\end{equation}
where the last step is just a relabeling of the summation index $(\pi',\mathbf t)\mapsto(\pi'',\mathbf t')$.
Similarly,
\begin{equation}
\begin{aligned}
(P_k^-)^\dagger
&\;\propto\; \sum_{\pi',\mathbf t} (-1)^{\|\mathbf t\|_1}\big(\tau_{\pi'}^{\mathbf t}\big)^\dagger
= \sum_{\pi',\mathbf t} (-1)^{\|\mathbf t\|_1}\big(\tau_{\pi'}^{\mathbf t}\big)^{-1}
= \sum_{\pi'',\mathbf t'} (-1)^{\|\mathbf t'\|_1}(\tau_{\pi''}^{\mathbf t'})
\;\propto\; P_k^-,
\end{aligned}
\end{equation}
since inversion may permute the coordinates of the selection vector $\mathbf t$ (i.e., $\mathbf t'=\rho\cdot \mathbf t$ for some $\rho\in \Sym_k$), but it preserves the Hamming weight $\|\mathbf t\|_1$, hence $(-1)^{\|\mathbf t'\|_1}=(-1)^{\|\mathbf t\|_1}$.
After fixing normalizations, this shows that
\begin{equation}
(P_k^\pm)^\dagger = P_k^\pm,
\end{equation}
i.e., both $P_k^+$ and $P_k^-$ are Hermitian.

To conclude that
\begin{equation}
U_k^\pm \;:=\; e^{i c_k^\pm P_k^\pm}
\end{equation}
is unitary for every real coefficient $c_k^\pm \in \mathbb{R}$, it suffices to recall a standard fact: whenever $H$ is Hermitian and $c\in\mathbb{R}$, the exponential
\begin{equation}
U := e^{i c H}
\end{equation}
is unitary. 

\subsubsection{Gate complexity.}

To estimate the circuit cost of
\begin{equation}
G^\ell = \prod_{k=2}^{k_{\max}} e^{ic^+_{\ell,k}P_k^+}\,e^{ic^-_{\ell,k}P_k^-},
\end{equation}
we implement each factor $e^{ic_{\ell,k}^\pm P_k^\pm}$ by constructing a block-encoding of $P_k^\pm$ from its linear-combination of unitaries (LCU) decomposition~\cite{childs2012hamiltonian} and then applying standard qubitization or QSVT-based Hamiltonian simulation~\cite{low2019hamiltonian,gilyen2019quantum}. For each factor $e^{ic_{\ell,k}^\pm P_k^\pm}$, let $B_k^\pm$ denote the two-qubit gate cost per block-encoding query, $Q_{\ell,k}^\pm$ the query count, and $C_{\ell,k}^\pm$ the resulting total two-qubit gate count.

Recall from Sec.~\ref{sec:our_nwk} that
\begin{equation}
P_k^\pm
=
\frac{1}{k!}
\sum_{\pi \in \mathrm{Perm}(P,k)}
\sum_{\mathbf{s}\in\{0,1\}^k}
(\pm 1)^{\|\mathbf{s}\|_1}\,
\tau_\pi^{\mathbf{s}}.
\label{eq:lcu_pk_pm}
\end{equation}
Thus $P_k^\pm$ admits an LCU decomposition with
\(
L_k = 2^k\,\frac{N!}{(N-k)!}
\label{eq:num_terms_pk}
\)
labeled summands, each having coefficient magnitude $1/k!$. Accordingly, we write
\begin{equation}
P_k^\pm = \sum_{j=1}^{L_k} a_{k,j}^\pm U_{k,j}^\pm,
\qquad
|a_{k,j}^\pm|=\frac{1}{k!},
\label{eq:pk_lcu_labeled}
\end{equation}
where the labeling need not identify pairwise distinct permutation unitaries. The corresponding LCU weight is
\begin{equation}
\lambda_k
:=
\sum_{j=1}^{L_k}|a_{k,j}^\pm|
=
\frac{L_k}{k!}
=
2^k \binom{N}{k}.
\label{eq:lambda_k}
\end{equation} Using the standard LCU block-encoding construction, we define the PREPARE and SELECT oracles by
\begin{equation}
\mathrm{PREP}_k^\pm \ket{0}^{\otimes m_k}
=
\frac{1}{\sqrt{\lambda_k}}
\sum_{j=1}^{L_k}
\sqrt{|a_{k,j}^\pm|}\,\ket{j},
\qquad
\mathrm{SELECT}_k^\pm
=
\sum_{j=1}^{L_k}\ket{j}\!\bra{j}\otimes \widetilde U_{k,j}^\pm +
\sum_{j=L_k+1}^{2^{m_k}}\ket{j}\!\bra{j}\otimes I,
\label{eq:prep_select_pk}
\end{equation}
where $m_k=\left\lceil \log_2 L_k \right\rceil$ and $\widetilde U_{k,j}^\pm := \operatorname{sgn}(a_{k,j}^\pm)\,U_{k,j}^\pm$ is still unitary. Then
\begin{equation}
W_k^\pm
:=
(\mathrm{PREP}_k^\pm{}^\dagger \otimes I)\,
\mathrm{SELECT}_k^\pm\,
(\mathrm{PREP}_k^\pm \otimes I)
\label{eq:block_encoding_pk}
\end{equation}
is a block-encoding of $P_k^\pm/\lambda_k$, namely
\begin{equation}
(\bra{0}^{\otimes m_k}\otimes I)\,W_k^\pm\,(\ket{0}^{\otimes m_k}\otimes I)
=
\frac{P_k^\pm}{\lambda_k}.
\label{eq:block_encoding_identity}
\end{equation}
Equation~\eqref{eq:block_encoding_identity} shows that $W_k^\pm$ is a block-encoding of the normalized operator $P_k^\pm/\lambda_k$. More precisely, if the ancilla register, consisting of $m_k$ qubits, is initialized in $\ket{0}^{\otimes m_k}$ and is projected back onto $\ket{0}^{\otimes m_k}$ after applying $W_k^\pm$, then the induced operation on the $2N$-qubit data register is exactly $P_k^\pm/\lambda_k$. We then implement $e^{ic_{\ell,k}^\pm P_k^\pm}$ with error $\varepsilon_{\ell,k}^\pm$ using qubitization or QSVT. 

To estimate the cost of this implementation, we define
\begin{equation}
B_k^\pm
:=
2\mathrm{Cost}(\mathrm{PREP}_k^\pm)
+
\mathrm{Cost}(\mathrm{SELECT}_k^\pm),
\label{eq:bk_def}
\end{equation}
namely, the two-qubit gate cost of one block-encoding query. Applying qubitization to the block-encoding in Eq.~\eqref{eq:block_encoding_pk} then yields query complexity 
\begin{equation}
Q_{\ell,k}^\pm
=
O\!\left(
\lambda_k\,|c_{\ell,k}^\pm|
+
\log\!\frac{1}{\varepsilon_{\ell,k}^\pm}
\right)
=
O\!\left(
2^k\binom{N}{k}|c_{\ell,k}^\pm|
+
\log\!\frac{1}{\varepsilon_{\ell,k}^\pm}
\right).
\label{eq:query_complexity_pk}
\end{equation}
The same asymptotic dependence can also be obtained in the QSVT framework. Hence the total two-qubit gate cost for implementing $e^{ic_{\ell,k}^\pm P_k^\pm}$ is
\begin{equation}
C_{\ell,k}^\pm
=
O\!\left(
B_k^\pm
\left(
2^k\binom{N}{k}|c_{\ell,k}^\pm|
+
\log\!\frac{1}{\varepsilon_{\ell,k}^\pm}
\right)
\right).
\label{eq:ck_from_bkqk}
\end{equation} 
To make the oracle cost explicit, we fix a concrete compilation model in which the $2N$ data qubits are arranged on a one-dimensional nearest-neighbor line, and the branch label $j$ is stored in a form from which the selected wires in $\widetilde U_{k,j}^\pm$ can be decoded with polylogarithmic overhead. In this model, $\widetilde U_{k,j}^\pm$ acts as a $k$-cycle on $k$ selected wires. Decomposing this $k$-cycle into $k-1$ transpositions and routing each transposition along the line yields the implementation-dependent upper bound
\begin{align}
 \mathrm{Cost}(\mathrm{SELECT}_k^\pm) = \widetilde O(kN)
\label{eq:bk_line_routing}
\end{align}
for fixed $k$.

By definition of $\mathrm{PREP}_k^\pm$ and using
$|a_{k,j}^\pm|=1/k!$ and $\lambda_k=L_k/k!$, we obtain
\begin{align}
\mathrm{PREP}_k^\pm \ket{0}^{\otimes m_k}
&=
\frac{1}{\sqrt{\lambda_k}}
\sum_{j=1}^{L_k}\sqrt{|a_{k,j}^\pm|}\ket{j}
=
\frac{1}{\sqrt{L_k}}\sum_{j=1}^{L_k}\ket{j}.
\end{align}
Therefore, $\mathrm{PREP}_k^\pm$ reduces to preparing the uniform superposition over the valid labels $j=1,\dots,L_k$ in the $m_k$-qubit label register. This can be prepared by a
standard recursive interval-state preparation routine. Define
\begin{align}
|\mathrm{Unif}_{m_k}(L_k)\rangle
=
\frac{1}{\sqrt{L_k}}\sum_{j=1}^{L_k}|j\rangle .
\end{align}
Writing \(q=2^{m_k-1}\), \(L_{k,0}=\min(L_k,q)\), and
\(L_{k,1}=\max(0,L_k-q)\), one has
\begin{align}
|\mathrm{Unif}_{m_k}(L_k)\rangle
=
\sqrt{\frac{L_{k,0}}{L_k}}
|0\rangle|\mathrm{Unif}_{m_k-1}(L_{k,0})\rangle
+
\sqrt{\frac{L_{k,1}}{L_k}}
|1\rangle|\mathrm{Unif}_{m_k-1}(L_{k,1})\rangle .
\end{align}
Thus, by recursively applying classically precomputed controlled rotations, one
prepares the uniform state on the \(L_k\) valid labels with zero amplitude on the
remaining \(2^{m_k}-L_k\) computational basis states, using
\(O(\operatorname{poly}(m_k))\) gates up to standard rotation-synthesis precision
factors. This implies that
\begin{align}
\label{eq:cost_prep}
\mathrm{Cost}(\mathrm{PREP}_k^\pm)=O(\operatorname{poly}(m_k))
=O(\operatorname{polylog}(L_k))
=O(\operatorname{polylog}(N))
\qquad (k\ \text{fixed}).
\end{align}

Thus, for fixed $k$ and bounded coefficients $c_{\ell,k}^\pm = O(1)$, Eqs.~\eqref{eq:ck_from_bkqk},~\eqref{eq:bk_line_routing}, and~\eqref{eq:cost_prep} imply
\begin{equation}
C_{\ell,k}^\pm
=
\widetilde O\!\left(
kN \cdot 2^k\binom{N}{k}
\right)
=
\widetilde O(N^{k+1}).
\label{eq:ck_line_routing}
\end{equation}
Indeed, for fixed $k$ one has $2^k\binom{N}{k}=\Theta(N^k)$, so the additional factor of $N$ comes from the line-routing cost in Eq.~\eqref{eq:bk_line_routing}.
If only orders up to $k_{\max}$ are retained, then the per-block gate count satisfies
\begin{equation}
C_{\mathrm{block}}
=
\widetilde O\!\left(
\sum_{k=2}^{k_{\max}} N^{k+1}
\right)
=
\widetilde O(N^{k_{\max}+1})
\qquad
(k_{\max}\ \text{fixed}).
\label{eq:cblock_line_routing}
\end{equation}
Hence, under this compilation model, the per-block circuit cost remains polynomial in $N$ on a $2N$-qubit data register, together with a logarithmic-size label ancilla register.

We emphasize that this is an oracle-dependent upper bound for a particular LCU block-encoding strategy together with a specific nearest-neighbor compilation model, not an optimized synthesis result for the full exponentials $e^{ic_{\ell,k}^\pm P_k^\pm}$. A more structured, symmetry-adapted implementation of the PREPARE and SELECT oracles could reduce the gate count further.

\subsection{Heisenberg Hamiltonian}
\subsubsection{Hamiltonian and expectation value.}
A Hamiltonian $H$ is a Hermitian operator acting on the system Hilbert space $\mathcal{H}$. Its physical prediction is the energy expectation
$\langle H\rangle_\rho := \mathrm{Tr}(\rho H)$ for a density operator $\rho$, which reduces to $\langle \psi|H|\psi\rangle$ for a pure state $|\psi\rangle$. We consider $n$ spin-$\tfrac12$ degrees of freedom arranged on a graph $G=(V,E)$ where edges indicate pairs that interact. At site $i\in V$, let $\vec S_i=(S_i^X,S_i^Y,S_i^Z)$ denote the spin operators with the standard normalization $S_i^\alpha=\tfrac\hbar2\,\sigma_i^\alpha$, where $\sigma_i^\alpha$ are Pauli matrices acting nontrivially only on site $i$.

\subsubsection{Heisenberg Hamiltonian.}
Let $G=(V,E)$ be an undirected interaction graph. The (isotropic) Heisenberg Hamiltonian is
\begin{equation}
H \;=\; J_{ex} \sum_{\langle i,j\rangle \in E} \vec S_i \!\cdot\! \vec S_j
\;=\; J_{ex} \sum_{\langle i,j\rangle \in E}\sum_{\alpha\in\{X,Y,Z\}} S_i^\alpha S_j^\alpha
\;=\; \frac{J_{ex}\,\hbar^2}{4}\sum_{\langle i,j\rangle \in E}\sum_{\alpha\in\{X,Y,Z\}}
\sigma_i^\alpha \sigma_j^\alpha,
\end{equation}
where $J_{ex}\in\mathbb{R}$ is the exchange coupling (ferromagnetic $J_{ex}<0$, antiferromagnetic $J_{ex}>0$), and $S_i^\alpha=\tfrac{\hbar}{2}\sigma_i^\alpha$.

\subsubsection{$\mathrm{SU}(2)$ invariance.}
Let $U\in\mathrm{SU}(2)$ act globally as $U^{\otimes n}$. The adjoint action rotates spin components: there exists $R(U)\in\mathrm{SO}(3)$ such that
\begin{equation}
U\,S_i^\alpha\,U^\dagger \;=\; \sum_{\beta\in\{X,Y,Z\}} R_{\alpha\beta}(U)\,S_i^\beta,
\qquad \text{or equivalently,}\qquad
U\,(\vec S_i\!\cdot\!\vec v)\,U^\dagger \;=\; \vec S_i\!\cdot\!\big(R(U)\vec v\big).
\end{equation}

Hence, each pairwise scalar product is invariant:
\begin{align}
U^{\otimes 2}\big(\vec S_i\!\cdot\!\vec S_j\big)(U^\dagger)^{\otimes 2}
&= \sum_{\alpha}\!\big(U^{\otimes 2} S_i^\alpha (U^\dagger)^{\otimes 2}\big)\big(U^{\otimes 2} S_j^\alpha (U^\dagger)^{\otimes 2}\big)
= \sum_{\alpha,\beta,\gamma}\! R_{\alpha\beta}(U)R_{\alpha\gamma}(U)\,S_i^\beta S_j^\gamma \\
&= \sum_{\beta,\gamma}\!\big(R(U)^\top R(U)\big)_{\beta\gamma}\,S_i^\beta S_j^\gamma
= \sum_{\beta} S_i^\beta S_j^\beta
= \vec S_i\!\cdot\!\vec S_j,
\end{align}
using orthogonality $R(U)^\top R(U)=I$. Summing over $(i,j)\in E$ yields
\begin{equation}
U^{\otimes n} H \big(U^{\otimes n}\big)^\dagger=H,
\end{equation}
i.e., the Heisenberg Hamiltonian is globally $\mathrm{SU}(2)$–invariant.

\section{Additional Experimental Details}

\subsection{Data construction}\label{app:c1}

We evaluate on two object-level datasets using small-class subsets: ModelNet (bottle, bowl, cup, lamp, stool) and ShapeNet (birdhouse, bottle, bowl, bus, cap). The original taxonomies contain ModelNet40 (40 classes) and ShapeNetCore (55 classes). Because each object is represented by only a few points in our sparse-point regime ($N \in \{4,5,6\}$), the available geometric information is severely limited. Discrimination over the full 40/55-class taxonomies, however, requires much finer geometric cues. We therefore use fixed five-class subsets for the clean CAD datasets to keep comparisons across methods more interpretable and stable.

For each normalized object, we first form a dense candidate set of surface points (typically $\sim$2{,}048 candidates).
Candidates are drawn from the underlying mesh surface representation to cover the shape broadly.
Objects are partitioned into training, validation, and test sets. For ModelNet, we follow the official train/test split and reserve a portion of the training data for validation; for ShapeNet, splits are constructed using disjoint sets of objects. Given a point budget of $N$, we maintain class-balanced splits with a 7:1:2 ratio, resulting in $700N$, $100N$, and $200N$ points per class for training, validation, and testing, respectively. For $N\in\{4,5,6\}$, points are selected via farthest point sampling (FPS) with a random initial seed. This iterative process selects the candidate furthest from the current set to maximize spatial coverage and minimize redundancy. To ensure sample diversity, points are sampled without replacement.

We summarize data preprocessing for each dataset as follows:
\begin{itemize}
    \item ModelNet: We use the widely adopted HDF5 release (\texttt{modelnet40\_ply\_hdf5\_2048}), which provides 2,048 pre-sampled surface points per object derived from the official meshes. We run the FPS on the provided candidate set to select $N$ points. We do not apply any additional centering or unit max-radius scaling, since the released data are already centered and unit max-radius scaled. Additionally, compared to ShapeNet, intra-class variation is smaller, and the models are cleaner.
    
    \item ShapeNet: We draw a dense candidate set by area-weighted uniform sampling over the mesh surface. We then center and unit max-radius scale the candidates and apply FPS to obtain $N$ points.
\end{itemize}

\subsection{Implementation settings}\label{app:c2}

Hardware mappings are listed in Table~\ref{tab:env-by-group}. Classical baselines can be run on local hardware, whereas quantum baselines typically require substantially more memory due to the exponential growth of the Hilbert space with the number of qubits. For example, HyQuRP caches $2(N-1)$ unitary matrices of size $2^{2N}\times 2^{2N}$ (from the eigendecompositions of the twirling operators), which makes large-memory GPUs (e.g., A100 or H100) necessary in our experiments.

\begin{table}[htbp]
\centering
\caption{Compute environments by run group.}
\label{tab:env-by-group}
\begin{tabular}{lll}
\toprule
Group & Datasets & Hardware \\
\midrule
Classical Baselines   & All datasets & Apple M4 chip   \\
RP-EQGNN             & All datasets                    & NVIDIA A100 
\\
HyQuRP ($N{=}4$)      & All datasets                   & Apple M4 chip  \\
HyQuRP ($N{=}5$)    & All datasets                   & NVIDIA A100  \\
HyQuRP ($N{=}6$)    & All datasets                   & NVIDIA H100  \\
\bottomrule
\end{tabular}
\end{table}

We use Adam, batch size of $35$, and train for $1000$ epochs. For each model, the same learning rate $\eta$ is used for all datasets and $N \in \{4,5,6\}$, which is optimized for the ModelNet dataset with $N=4$ over $\eta \in \{10^{-2},10^{-3},10^{-4}\}$. We do not use schedulers, weight decay, dropout, or early stopping. Model selection uses the validation accuracy: we report the best top-1 test accuracy at the epoch achieving the best validation accuracy. Results are averaged over $7$ seeds. 
Table~\ref{tab:train} summarizes our training implementation. During training, some runs failed; therefore, we introduce an additional rule to handle these cases (see Appendix~\ref{app:summary_metrics}).

\begin{table}[htbp]
\centering
\caption{Training protocol.}
\label{tab:train}
\begin{tabular}{p{0.28\textwidth} p{0.62\textwidth}}
\toprule
Setting & Value \\
\midrule
Augmentation & random point permutation, rotation and jitter \\
Optimizer & Adam \\
Loss function & cross-entropy\\
Batch size & 35 \\
Epochs & 1000 \\
Learning rate & optimized on ModelNet with $N{=}4$ from $\{10^{-2},10^{-3},10^{-4}\}$ is used for all other datasets\\
Selection rule & best validation accuracy epoch \\
Seeds & 7 (121, 831, 1557, 2023–2026) \\
Regularization/Dropout/Gradient clipping/Scheduler & none \\
\bottomrule
\end{tabular}
\end{table}

Training-time augmentations are identical across all methods: SO(3) rotation, permutation, and jitter with standard deviation $\sigma_{jitter}$. In this paper, we set $\sigma_{jitter}$ to $0.02$, $\Theta$ to $1.7$, the block repetition $B$ to $12$, and use the aggregation $\{\mathrm{mean},\mathrm{max},\mathrm{min},\mathrm{sum},\mathrm{var},\mathrm{std}\}$. For neighborhood-based baselines (DGCNN, VN-PointNet, Point TF, PointMLP), we use $k\in\{1, 2,3\}$ when $N\in\{4,5,6\}$. More details of the implementation are available in the GitHub repository~\href{https://github.com/YonseiQC/equivariant_QML}{https://github.com/YonseiQC/equivariant\_QML}.

\FloatBarrier
\subsection{Baseline architectures}\label{app:baseline_arch}

This appendix details how we instantiate the Light and Mid variants for HyQuRP and each baseline family.
We follow a simple rule: scale widths and, where explicitly stated, the number of repeated blocks
until the target parameter count is reached, while avoiding structural omissions. 
In particular, we keep all architectural components of the original backbones 
(e.g., pooling operators, neighborhood definitions, positional-encoding schemes), 
except for operations that become ill-defined or redundant at our tiny point budgets 
(such as internal resampling layers); only layer widths (and the number of block repetitions, when noted) are reduced. For each family below, we briefly describe the original model, and the concrete width/depth settings
for the Light/Mid variants are summarized next to the comparison tables in this appendix
(see Tables~\ref{tab:pointnet},
~\ref{tab:tfn_family},
\ref{tab:dgcnn}, 
~\ref{tab:vnpointnet},
\ref{tab:pointtransformer}, \ref{tab:pointmlp}, \ref{tab:pointmamba}, \ref{tab:mamba3d_family}, \ref{tab:rpeqgnn_family},
\ref{tab:hyq_setmlp_light} and
\ref{tab:hyq_setmlp_mid}).

\subsubsection{MLP.}
A plain fully connected network that consumes flattened coordinates in $\mathbb{R}^{3N}$.
The vector is processed by a stack of dense layers with $\tanh$ activations, followed by a linear classifier.
No permutation-equivariant structure, attention, or graph operations are used; geometry is inferred from the raw coordinate vector.
Depth/width are scaled only to satisfy the Light/Mid capacity targets.

\subsubsection{Set-MLP.}
A shared per-point MLP maps $x_i \in \mathbb{R}^3$ to $\mathbb{R}^d$, 
followed by a symmetric aggregation (e.g., mean, max, or standard-deviation pooling) and a small fully connected head, 
so that the resulting set representation is permutation-invariant with respect to the input point order.
Light/Mid variants reduce the per-point tower and head widths while maintaining the aggregation and normalization of the original.

\subsubsection{PointNet~\citep{qi2017pointnet}.}
Our PointNet backbone follows the original design with input and feature transform networks (T-Nets) and a global set pooling stage. An input T-Net first predicts a $3\times 3$ affine transform that is applied to the raw coordinates $x_i \in \mathbb{R}^3$, roughly canonicalizing the point cloud and improving robustness to rigid transforms such as rotations; a shared per-point MLP then maps the transformed points to higher-dimensional features. A second feature T-Net predicts a feature-space transform that is applied before another shared MLP tower. The resulting per-point features are aggregated into a global shape descriptor via symmetric max pooling over the point dimension, and a shallow MLP classifier head maps this global feature to logits, yielding a permutation-invariant set representation.

For the Light/Mid variants, we uniformly narrow the widths of the per-point MLP towers and the classifier head, and proportionally scale down the T-Net MLPs to meet the $\sim$1.5K and $\sim$7K parameter budgets, while keeping the input-transform $\rightarrow$ per-point MLP $\rightarrow$ feature-transform  $\rightarrow$ global max-pooling $\rightarrow$ classifier-head computation pattern identical to the original. Concrete width/depth configurations for all Light/Mid variants are summarized in the PointNet family tables in this appendix (see Table~\ref{tab:pointnet}).

\subsubsection{Tensor Field Network ~\citep{thomas2018tensor}.}
Our Tensor Field Network (TFN) backbone follows the original equivariant point-convolution design of Tensor Field Network: for each pair of points, the network forms relative geometric quantities from their coordinate differences and applies rotation-equivariant filters constructed from learned radial functions and spherical harmonics. Point features are organized by rotation order, and each TFN module applies the available \(l=0\) and \(l=1\) convolution paths, concatenates the resulting features, mixes channels with a self-interaction layer, and then applies a rotation-equivariant nonlinearity. After repeating these modules, the network uses only the final \(l=0\) (scalar) features for classification and aggregates them over the point set to obtain a shape-level descriptor, yielding a rotation and permutation-invariant prediction.

For the Light/Mid variants, we preserve the original TFN computation pattern based on equivariant relative-geometry filtering, order-specific feature propagation, concatenation, self-interaction, and rotation-equivariant nonlinearity. To meet the \(\sim\)1.5K and \(\sim\)7K parameter budgets, we scale the channel widths through the layer dims configurations while keeping the number of modules fixed. Concrete Light/Mid settings are summarized in the TFN family tables in this appendix (see Table~\ref{tab:tfn_family}).

\subsubsection{DGCNN~\citep{wang2019dynamic}.}
Our DGCNN backbone follows the original dynamic-graph design: at each EdgeConv stage, a $k$-NN graph is constructed in the current feature space, and edge features are formed by concatenating central-node and offset features (e.g., $[h_i \,\|\, h_j - h_i]$) over neighbors $j \in \mathcal{N}(i)$. A shared MLP is applied to these edge features and aggregated over neighbors (typically by max pooling) to produce updated point features, and the $k$-NN graph is recomputed on the new features before the next EdgeConv stage.
After repeating these steps multiple times, we concatenate or stack multiscale point features, apply a global symmetric pooling over the point dimension to obtain a shape-level descriptor, and pass it through a shallow MLP classifier head, yielding a permutation-invariant set representation.

For the Light/Mid variants, we preserve dynamic $k$-NN graph updates, the EdgeConv formulation, and we only scale down channel widths in the EdgeConv MLPs and classifier head; when necessary to meet the $\sim$1.5K and $\sim$7K parameter budgets, we also reduce the number of EdgeConv stages while keeping the overall dynamic-graph $\rightarrow$ EdgeConv $\rightarrow$ global pooling $\rightarrow$ classifier-head computation pattern unchanged. 
Concrete width/depth configurations for all Light/Mid variants are summarized in the DGCNN family tables in this appendix (see Table~\ref{tab:dgcnn}).

\subsubsection{VN-PointNet \cite{deng2021vector}.}
Our VN-PointNet backbone follows the official Vector Neurons implementation. It retains the overall PointNet style shape encoding paradigm, while extending it to be SO(3) equivariant. In particular, the model adopts vector-neuron representations instead of standard scalar features. Accordingly, it employs vector neuron specific linear layers, nonlinearities, normalization, and canonicalization, while the input transformation network used in the original PointNet is omitted. As a result, VN-PointNet preserves PointNet's permutation-invariant shape modeling after symmetric aggregation, while additionally achieving rotation invariance through its equivariant vector neuron design. 

For the Light/Mid variants, we keep all components of the official VN-PointNet architecture unchanged except for dropout, which is omitted, and we reduce the channel widths to match the target parameter budgets. We do not alter the underlying VN operators, the use of the $k$-NN-based local graph construction, the feature-transform branch, the VNStdFeature canonicalization step, or the final global pooling pattern. Concrete Light/Mid settings are summarized in the VN-PointNet family tables in this appendix (see
Table~\ref{tab:vnpointnet}).

\subsubsection{PointTransformer~\citep{zhao2021point}.}
The original PointTransformer (Point TF) backbone first projects input features to a 32-dimensional latent space and then applies a stack of point-transformer blocks that perform self-attention over local $k$-NN.
Trainable relative positional encodings $\delta(p_i - p_j)$ are injected into both the attention generation branch and the feature transformation branch. Feature widths increase through multiple stages with hierarchical downsampling (e.g., FPS), followed by global average pooling and a fully connected classifier head. 

For the Light/Mid variants, we keep the explicit input projection, the relative positional-encoding scheme, vector attention method, FFN block structure and residual connections, but remove the multiscale hierarchy and any FPS-style resampling: a fixed small $k$-NN graph (defined in coordinate space) is used on the full point set (recomputed per block for simplicity but effectively identical for a given input), and two compact point-transformer blocks operate at reduced widths (tens of channels rather than hundreds).
A lightweight Conv1d classifier head maps the final per-point features to logits and then averages over points, yielding a permutation-invariant prediction. All modifications are implemented by shrinking the projection and block widths (and, where necessary, the number of blocks) to meet the $\sim$1.5K and $\sim$7K parameter budgets, with concrete configurations summarized in the Point Transformer family table (Table~\ref{tab:pointtransformer}).

\subsubsection{PointMLP~\citep{ma2022rethinking}.}
Our PointMLP backbone follows the residual MLP framework of Ma et al.: given an input point cloud, it progressively extracts local features over a sequence of stages built on $k$-NN.
At each stage, a lightweight geometric affine module first transforms local point coordinates/features within each neighborhood, after which residual per-point MLP (ResP) blocks process per-point features before and after a simple $k$-NN-based local aggregation. Between stages, the network performs point downsampling to gradually increase the receptive field and reduce the number of points. After the final stage, per-point features are aggregated by a global symmetric pooling over the remaining points to produce a shape-level descriptor, which is fed into a shallow MLP classifier head to output logits.

For the Light/Mid variants, we simplify the original hierarchical architecture to suit our experimental setting: We remove FPS-based resampling and have the network operate at a single resolution, using only $k$-NN to build local neighborhoods in the initial Geometric Affine encoder. 
We preserve the residual point-MLP block structure and the geometric affine/context-fusion design, but uniformly scale down channel widths in the stem, stages, and classifier head to meet the $\sim$1.5K and $\sim$7K parameter budgets; when necessary, we also shorten the number of stages. Global features are obtained by max pooling over all $N$ points. Concrete width/depth configurations for our Light/Mid PointMLP baselines are summarized in the PointMLP family tables in this appendix (see Table~\ref{tab:pointmlp}).

\subsubsection{PointMamba~\citep{liang2024pointmamba}.}
The original PointMamba backbone converts the unordered 3D point set into a 1D token sequence using a space-filling-curve (SFC) based tokenizer. Starting from the raw point cloud, it selects a subset of points as centers via farthest-point sampling (FPS) and serializes these centers using fixed SFC scans (e.g., Hilbert and Trans-Hilbert). For each serialized center, it forms a local patch via $k$-NN and applies a lightweight PointNet-style MLP to aggregate patch points into a patch token. To distinguish tokens from different scans, a learnable order indicator (e.g., per-scan scale/shift) is added, and the two token sequences are concatenated and fed into stacked Mamba blocks. This sequence is processed by a stack of Mamba state-space blocks that apply selective 1D sequence modeling: each block uses pre-LayerNorm, a linear projection into two branches, a depthwise 1D convolution for local mixing, and an input-conditioned selective SSM recurrence parameterized by $A,B,C,\Delta t,D$, followed by a gated output and residual connection. Finally, a global pooling plus linear classifier maps the final sequence features to class logits.

For the Light/Mid variants, we retain the selective Mamba block structure and a pooling head, but substantially reduce the model depth and channel widths. We keep the SFC-style point-to-sequence serialization based on Hilbert/Trans-Hilbert scans, but in the sparse-point regime, we directly tokenize each point and omit any FPS- or $k$-NN-based grouping (patching) altogether (compared to architectures that rely on $k$-NN as a core inductive bias, PointMamba does not inherently require them).
 To ensure a uniform implementation environment across all baselines, we do not rely on the official \texttt{mamba-ssm} library. Instead, we implement our own compact Mamba block in plain PyTorch that follows the same selective state-space formulation, but forgoes the fused CUDA kernels and other low-level optimizations of the original implementation.
All modifications are realized by shrinking the embedding, state, and block widths (and, where needed, the number of blocks) to fit the $\sim$1.5K and $\sim$7K parameter budgets, with concrete Light/Mid configurations summarized in the PointMamba family table (Table~\ref{tab:pointmamba}).

\subsubsection{Mamba3D~\citep{han2024mamba3d}.}
The original Mamba3D backbone is built around two key components---a local norm pooling (LNP)  and a bidirectional state-space module (bi-SSM) tailored for 3D point clouds. Given an unordered point set in $\mathbb{R}^3$, Mamba3D first applies FPS\ \&\ $k$-NN and obtains initial patch tokens via a lightweight PointNet-style embedding. The LNP block then aggregates local context over token neighborhoods via $k$-norm propagation and $k$-pooling (with a shared MLP for channel alignment). This sequence is subsequently processed by a bi-SSM block, which applies a Mamba-style selective state-space model along the token dimension (L$+$SSM) and a second selective SSM along the feature-channel dimension (C$-$SSM) using a channel flip, and fuses both directions through residual connections followed by a classification head. 

For the Light/Mid variants, we retain the LNP token-mixing and bi-SSM channel-mixing blocks, including the feature-channel flip for C-SSM, along with the original pre-norm residual block structure and positional encoding injected at each encoder layer. However, since our evaluation focuses on extremely small point regimes, we omit the FPS- and $k$-NN-based patch construction (compared to architectures that rely on $k$-NN as a core inductive bias, Mamba3D does not inherently require them) and instead operate directly on per-point tokens (with LNP applied over a global token neighborhood).
To ensure a uniform implementation environment across all baselines, we do not rely on the original Mamba3D codebase or the official \texttt{mamba\mbox{-}ssm} library. While the reference Mamba3D implementation builds its bidirectional C\mbox{-}SSM core using the Mamba selective state-space module from \texttt{mamba\mbox{-}ssm}, our Light/Mid variants replace this component with a custom PyTorch implementation of the same selective state-space update with the feature-channel flip used for C\mbox{-}SSM. The resulting compact backbone is mathematically consistent with the original formulation but omits fused CUDA kernels and other low-level engineering optimizations. All modifications are realized by shrinking the embedding, state, and block widths (and, where needed, the number of blocks) to fit the $\sim$1.5K and $\sim$7K parameter budgets, with concrete Light/Mid configurations summarized in the Mamba3D family table (Table~\ref{tab:mamba3d_family}).

\subsubsection{RP-EQGNN~\cite{liu2025rpeqgnn}.}
RP-EQGNN is a quantum graph neural network that exploits rotational and permutational symmetries for 3D graph data. In the paper, the model is organized into three quantum modules: (i) The rotation- and permutation-equivariant module $U_R$ comprises a geometric information encoding circuit $U_\varphi$ and a geometric information entanglement circuit $U_\phi$, where $U_\varphi$ encodes node coordinates using universal rotation gates, mapping them into a Hilbert space while $U_\phi$ entangles qubits to learn geometric relations. (ii) The convolution and entanglement module $U_N$ consists of a non-geometric information convolution circuit $U_C$ and a non-geometric information entanglement circuit $U_E$, and is responsible for extracting node features. (iii) A geometric de-entanglement module $U_R^\dagger$ is applied to preserve equivariance, after which the model output is produced via measurement and followed by classical post-processing for prediction. 

However, the description of RP-EQGNN within the original paper cannot produce simultaneously rotation- and permutation-invariant outputs for mathematical reasons.
To obtain an invariant output within their framework, the model would require an $\mathrm{SU}(2)$-invariant initial state and Hamiltonian, and an $\mathrm{SU}(2)$-equivariant encoding gate. However, they initialize the state as $\ket{0}$, which is not $\mathrm{SU}(2)$-invariant, and the paper provides no details on the specific Hamiltonian or the construction of the encoding gate. 
On the other hand, in the official implementation from GitHub%
\footnote{\href{https://github.com/wqs1999/RP-EQGNN_YiFanZhu}{GitHub repository}; commit \texttt{da7cba0}, accessed on 2025-12-31.}, they use a Pauli-$Z$ Hamiltonian and a direct coordinate encoding, which does not ensure $\mathrm{SU}(2)$ invariance.

For the Light/Mid variants, we inspected the authors' official public implementation.
We found that the released code does not provide a one-to-one realization of the above modular decomposition.
Specifically, the implementation consists of (i) coordinate encoding, (ii) node-feature encoding, and (iii) edge-conditioned entanglement, without the full structure described in the paper.
In our baseline, we implement a compact RP-EQGNN circuit that follows the overall design of the official implementation, while adapting the feature definitions to the point cloud setting. For node-feature encoding, we apply single-qubit $\mathrm{RY}$ and $\mathrm{RX}$ rotations parameterized by a scalar node feature defined as the point norm.
For edge-conditioned entanglement, we use the pairwise distance between two points and the angle between their position vectors (with respect to the origin) as edge features. The resulting architectural differences from the official implementation and our concrete Light/Mid configurations are summarized in Table~\ref{tab:rpeqgnn_family}.

\FloatBarrier

\begin{table}[t]
\centering
\footnotesize
\setlength{\tabcolsep}{6pt}
\renewcommand{\arraystretch}{1.15}
\caption{PointNet family: Original vs.\ Light vs.\ Mid.}
\label{tab:pointnet}
\begin{tabular}{p{0.25\textwidth} p{0.24\textwidth} p{0.24\textwidth} p{0.24\textwidth}}
\toprule
\textbf{Component} & \textbf{PointNet\_Original} & \textbf{PointNet\_Light} & \textbf{PointNet\_Mid} \\
\midrule
Input T-Net ($3\times3$) & $3\to64\to128\to1024$ & $3\to8\to10$ & $3\to8\to10$ \\
\midrule
Per-point stem & $3\to64\to64$ & $3\to9$ & $3\to4\to8$ \\
\midrule
Feature T-Net ($k\times k$) & $64\to64\to128\to1024\to k^2$ ($k{=}64$) & $9\to9\to10\to k^2$ ($k{=}9$) & $8\to16\to 20\to k^2$ ($k{=}8$) \\
\midrule
Feature blocks (per-point) & $64\to64\to128\to1024$ & (none) & $8\to32\to120\to8$ \\
\midrule
Pooling & $\to1024$ & $\to9$ & $\to8$ \\
\midrule
Classifier / Head & $1024\to512\to256\to K$ + Dropout(0.3) & $9\to16\to K$ (no dropout) & $8\to K$ (no dropout) \\
\midrule
Activation / Normalization & ReLU + BatchNorm & ReLU + BatchNorm & ReLU + BatchNorm \\
\midrule
Structural differences & Baseline architecture & T-Nets/feature blocks present (reduced); head reduced & T-Nets/feature blocks present (reduced); head reduced \\
\bottomrule
\end{tabular}
\end{table}

\begin{table*}[t]
\centering
\footnotesize
\setlength{\tabcolsep}{5pt}
\renewcommand{\arraystretch}{1.18}
\caption{TFN family: Original vs.\ Light vs.\ Mid.}
\label{tab:tfn_family}
\begin{tabular}{p{0.18\textwidth} p{0.26\textwidth} p{0.26\textwidth} p{0.26\textwidth}}
\toprule
\textbf{Component} & \textbf{TFN\_Original} & \textbf{TFN\_Light} & \textbf{TFN\_Mid} \\
\midrule

Input / Geometric Encoding
&
Point cloud coordinates $\mathbf{r}_i \in \mathbb{R}^3$ are used directly. A constant scalar field is embedded from ones into one $l=0$ channel. Relative geometry is encoded through pairwise offsets $\mathbf{r}_{ij}$, distances $d_{ij}$, and Gaussian radial basis functions.
&
Same as original.
&
Same as light. \\

\midrule
Channel Dimensions
&
Official shape-classification demo uses layer dims $= [1,4,4,4]$. Thus the retained orders evolve as
$l=0: 1 \rightarrow 4 \rightarrow 4 \rightarrow 4$,
$l=1: 0 \rightarrow 4 \rightarrow 4 \rightarrow 4$.
&
layer dims $= [1,4,8,16]$. 
&
layer dims $= [1,8,12,14]$. \\

\midrule
Radial / Filter Encoding &
Gaussian RBF encoding with $rbf\_count=4$; the radial function is a two-layer MLP whose hidden width defaults to the RBF input dimension, i.e., 4.
&
Uses the same RBF range and count as the original implementation, with radial hidden width 4.
&
Uses the same RBF range and count as the original implementation, but increases the radial hidden width to 39.
\\

\midrule
Equivariant Convolution / Self-Interaction
&
Each module applies the retained $l=0,1$ TFN convolution paths
($0\!\to\!0$, $0\!\to\!1$, $1\!\to\!1$, $1\!\to\!0$, $1\!\to\!1$),
followed by concatenation and order-wise self-interaction.
Only $l=0$ and $l=1$ outputs are retained.
&
Same as original.
&
Same as light. \\

\midrule
Feature Blocks / Depth
&
Three TFN modules. Module 1 starts from scalar-only input; Modules 2--3 process both scalar and vector fields.
Before self-interaction, the concatenated channel sizes are
$(l=0,l=1)=(1,1)$ in Module 1 and $(8,12)$ in Modules 2--3, then projected to width $4$.
&
Three TFN modules.
Before self-interaction, the concatenated channel sizes are
$(1,1)$ in Module 1,
$(8,12)$ in Module 2,
and $(16,24)$ in Module 3,
then projected to widths $4,8,16$.
&
Three TFN modules.
Before self-interaction, the concatenated channel sizes are
$(1,1)$ in Module 1,
$(16,24)$ in Module 2,
and $(24,36)$ in Module 3,
then projected to widths $8,12,14$. \\

\midrule
Pooling
&
Global pooling is applied only on the final $l=0$ scalar features to obtain a rotation-invariant descriptor.
&
Same as original.
&
Same as light. \\

\midrule
Classifier / Head
&
Single fully connected classifier $4 \rightarrow K$, where $K$ is the number of classes of the selected dataset.
&
Single fully connected classifier $16 \rightarrow K$.
&
Single linear head $14 \rightarrow K$. \\

\midrule
Activation / Normalization
&
Rotation-equivariant nonlinearity is applied after each self-interaction. The implementation uses ELU-style equivariant nonlinearity. No BatchNorm or LayerNorm is used.
&
Same as original.
&
Same as light. \\

\midrule
Structural Differences
&
Baseline architecture
&
Same TFN framework as the original, but with a slightly larger capacity.
&
Same framework as light, but with a larger capacity. \\

\bottomrule
\end{tabular}
\end{table*}

\begin{table}[t]
\centering
\footnotesize
\setlength{\tabcolsep}{6pt}
\renewcommand{\arraystretch}{1.15}
\caption{DGCNN family: Original vs.\ Light vs.\ Mid.}
\label{tab:dgcnn}
\begin{tabular}{p{0.25\textwidth} p{0.24\textwidth} p{0.24\textwidth} p{0.24\textwidth}}
\toprule
\textbf{Component} & \textbf{DGCNN\_Original} & \textbf{DGCNN\_Light} & \textbf{DGCNN\_Mid} \\
\midrule
Input Encoding & Edge feature MLPs in EdgeConv, from $(p_i,\,p_j,\,p_i{-}p_j)\to64$ & Edge feature MLPs with narrow channel width $(6\to7)$; we use the same formulation, but with edge features $(p_j - p_i,\, p_i)$ instead of $(p_i,\, p_j,\, p_i - p_j)$.
& Same as light, but with a narrow edge feature MLP mapping $(6\to8)$.

 \\
\midrule
Positional Encoding & No explicit PE; uses relative offsets $p_i{-}p_j$ within EdgeConv & No explicit PE; relative offsets used implicitly (same) & No explicit PE; relative offsets used implicitly (same) \\
\midrule
Stem / Input Embedding & EdgeConv-1: dynamic $k$-NN graph, $6\to64$ & EdgeConv-1: dynamic $k$-NN graph, $6\to7$ & EdgeConv-1: dynamic $k$-NN graph, $6\to8$ \\
\midrule
Feature Blocks and Depth & EdgeConv-2: $2{\times}64\to64$; EdgeConv-3: $2{\times}64\to128$; EdgeConv-4: $2{\times}128\to256$ & EdgeConv-2: $2{\times}7\to15$; EdgeConv-3: $2{\times}15\to8$ & EdgeConv-2: $2{\times}8\to32$; EdgeConv-3: $2{\times}32\to72$ \\
\midrule
Sampling / Grouping & $k$-NN with dynamic graph (per layer); no FPS & $k$-NN with dynamic graph; no FPS & $k$-NN with dynamic graph; no FPS \\
\midrule
Pooling & Global max pooling over concatenated multiscale edge features $\to512$ & Global max pooling $\to30$ & Global max pooling $\to112$ \\
\midrule
Classifier / Head & Fully connected $512\to256\to K$ & Fully connected $30\to13\to9\to K$ & Fully connected $112\to16\to9\to K$ \\
\midrule
Activation / Normalization & Leaky ReLU + BatchNorm & Leaky ReLU + BatchNorm & Leaky ReLU + BatchNorm \\
\midrule
Structural Differences  & Baseline architecture & Fewer stages; narrower channels; no extra $1{\times}1$ conv before pooling & Fewer stages; narrower channels; no extra $1{\times}1$ conv before pooling \\
\bottomrule
\end{tabular}
\end{table}

\begin{table}[t]
\centering
\footnotesize
\setlength{\tabcolsep}{6pt}
\renewcommand{\arraystretch}{1.15}
\caption{VN-PointNet family: Original vs.\ Light vs.\ Mid.}
\label{tab:vnpointnet}
\begin{tabular}{p{0.25\textwidth} p{0.24\textwidth} p{0.24\textwidth} p{0.24\textwidth}}
\toprule
\textbf{Component} & \textbf{VN-PointNet\_Original} & \textbf{VN-PointNet\_Light} & \textbf{VN-PointNet\_Mid} \\
\midrule

Local graph 
& Fixed $k$-NN ($k{=}20$)
& Same as original with small $k$.
& Same as light. \\

\midrule
VN stem + local pooling
& conv\_pos: $3 \rightarrow 21$, then symmetric neighbor pooling
(Mean or VN-Max)
& conv\_pos: $3 \rightarrow 2$, then Mean pool
& conv\_pos: $3 \rightarrow 4$, then Mean pool \\

\midrule
VN encoder trunk
& conv1: $21 \rightarrow 21$ $\to$
feature-transform branch $\to$
conv2: $42 \rightarrow 42$ $\to$
conv3: $42 \rightarrow 341$ + VNBatchNorm
& conv1: $2 \rightarrow 2$ $\to$
feature-transform branch $\to$
conv2: $4 \rightarrow 4$ $\to$
conv3: $4 \rightarrow 10$ + VNBatchNorm
& conv1: $4 \rightarrow 4$ $\to$
feature-transform branch $\to$
conv2: $8 \rightarrow 8$ $\to$
conv3: $8 \rightarrow 24$ + VNBatchNorm \\

\midrule
Feature-transform branch 
& stn\_conv1 $21 \rightarrow 21$,
stn\_conv2 $21 \rightarrow 42$,
stn\_conv3 $42 \rightarrow 341$,
stn\_fc1 $341 \rightarrow 170$,
stn\_fc2 $170 \rightarrow 85$,
stn\_fc3 $85 \rightarrow 21$
& stn\_conv1 $2 \rightarrow 2$,
stn\_conv2 $2 \rightarrow 4$,
stn\_conv3 $4 \rightarrow 10$,
stn\_fc1 $10 \rightarrow 5$,
stn\_fc2 $5 \rightarrow 2$,
stn\_fc3 $2 \rightarrow 2$
& stn\_conv1 $4 \rightarrow 4$,
stn\_conv2 $4 \rightarrow 8$,
stn\_conv3 $8 \rightarrow 24$,
stn\_fc1 $24 \rightarrow 12$,
stn\_fc2 $12 \rightarrow 6$,
stn\_fc3 $6 \rightarrow 4$ \\

\midrule
Global context + canonicalization
& Concatenate global mean to each point; VNStdFeature on $682$ channels; flatten to $2046$ scalars; global max over points
& Same pipeline; VNStdFeature on $20$ channels; flatten to $60$ scalars; global max over points
& Same pipeline; VNStdFeature on $48$ channels; flatten to $144$ scalars; global max over points \\

\midrule
Classifier / head
& Linear $2046 \rightarrow 512 \rightarrow 256 \rightarrow K$ + BN + ReLU;
dropout ($p{=}0.4$) before final FC
& Linear $60\to 6\to 5\to K$ + BN + ReLU 
& Linear $144\to 15\to 14\to K$ + BN + ReLU  \\

\midrule
Structural differences
& Baseline architecture
& Widths reduced and dropout removed
& Same as Light, with larger widths \\

\bottomrule
\end{tabular}
\end{table}

\begin{table}[t]
\centering
\footnotesize
\setlength{\tabcolsep}{6pt}
\renewcommand{\arraystretch}{1.15}
\caption{Point TF family: Original vs.\ Light vs.\ Mid.}
\label{tab:pointtransformer}
\begin{tabular}{p{0.25\textwidth} p{0.24\textwidth} p{0.24\textwidth} p{0.24\textwidth}}
\toprule
\textbf{Component} & \textbf{Point TF\_Original} & \textbf{Point TF\_Light} & \textbf{Point TF\_Mid} \\
\midrule
Input Encoding & Linear/Conv1d $C_{\mathrm{in}}\to32$ before Block 1 & Linear/Conv1d $C_{\mathrm{in}}\to5$ before Block 1 & Linear/Conv1d $C_{\mathrm{in}}\to8$ before Block 1 \\
\midrule
Positional Encoding & Relative PE $\delta(p_i{-}p_j)$ via small MLP; added to attention and feature bias & Relative PE retained; no learned absolute PE & Relative PE retained; no learned absolute PE \\
\midrule
Input Embedding & Block 1: Transformer block on 32-d features & Block 1: Transformer block on 5-d features & Block 1: Transformer block on 8-d features \\
\midrule
Feature Blocks and Depth & Blocks: $32$ (B1) $\to64$ (B2) $\to128$ (B3) $\to256$  (B4) $\to512$ (B5) & Blocks: $5$ (B1) $\to7$ (B2) $\to8$ & Blocks: $8$ (B1) $\to14$ (B2) $\to18$ \\
\midrule
Transformer Block & Block 1: $32\!\to\!32$; Attn: $32\!\to\!32$; Linear: $32\!\to\!32$; Residual$\times1$

Block 2: $64\!\to\!64$; Attn: $64\!\to\!64$; Linear: $64\!\to\!64$; Residual$\times1$

Block 3: $128\!\to\!128$; Attn: $128\!\to\!128$; Linear: $128\!\to\!128$; Residual$\times1$

Block 4: $256\!\to\!256$; Attn: $256\!\to\!256$; Linear: $256\!\to\!256$; Residual$\times1$

Block 5: $512\!\to\!512$; Attn: $512\!\to\!512$; Linear: $512\!\to\!512$; Residual$\times1$

& Block 1: $5\!\to\!7$; Attn: $7\!\to\!7$; Linear: $7\!\to\!28\!\to\!7$; Residual$\times2$

Block 2: $7\!\to\!8$; Attn: $8\!\to\!8$; Linear: $8\!\to\!32\!\to\!8$; Residual$\times2$
& Block 1: $8\!\to\!14$; Attn: $14\!\to\!14$; Linear: $14\!\to\!56\!\to\!14$; Residual$\times2$

Block 2: $14\!\to\!18$; Attn: $18\!\to\!18$; Linear: $18\!\to\!72\!\to\!18$; Residual$\times2$\\
\midrule
Sampling / Grouping & Hierarchical downsampling (e.g., FPS); local $k$-NN & No hierarchical downsampling; fixed small $k$-NN on full set & No hierarchical downsampling; fixed small $k$-NN on full set \\
\midrule
Classifier / Head & The paper does not provide classifier-layer details, and we could not verify an official classification implementation & Conv1d $8\to8\to K$; then mean over points & Conv1d $18\to18\to K$; then mean over points \\
\midrule
Activation / Normalization & The paper does not provide classifier-layer details, and we could not verify an official classification implementation & ReLU + LayerNorm & ReLU + LayerNorm \\
\midrule
Structural Differences  & Baseline architecture & Removes hierarchy; keeps explicit input projection and relative PE & Removes hierarchy; keeps explicit input projection and relative PE \\
\bottomrule
\end{tabular}
\end{table}

\begin{table}[t]
\centering
\footnotesize
\setlength{\tabcolsep}{6pt}
\renewcommand{\arraystretch}{1.15}
\caption{PointMLP Family: Original vs.\ Light vs.\ Mid.}
\label{tab:pointmlp}
\begin{tabular}{p{0.25\textwidth} p{0.24\textwidth} p{0.24\textwidth} p{0.24\textwidth}}
\toprule
\textbf{Component} & \textbf{PointMLP(elite)\_Original} & \textbf{PointMLP\_Light} & \textbf{PointMLP\_Mid} \\
\midrule
ResidualMLPBlock & Two FC--BN--Act; skip projection if in=out & Two FC--BN--Act; skip proj if in=out & Two FC--BN--Act; skip proj if in=out\\
\midrule
PointMLPStage & Pre residual (in$\to$in); Grouping + local MLP; Post residual (in$\to$out) & Pre residual(in$\to$in); Aggregation + MLP (no grouping); Post residual(in$\to$out) & Pre residual(in$\to$in); Aggregation + MLP (no grouping); Post residual(in$\to$out) \\
\midrule
Stem / Input & Linear $3\to32$; BN+ReLU & Linear $3\to6$; BN+ReLU & Linear $3\to8$; BN+ReLU \\
\midrule
Stage 1 & $32\to64$ & $6\to10$ & $8\to8$ \\
\midrule
Stage 2 & $64\to128$ & $10\to12$ & $8\to16$ \\
\midrule
Stage 3 & $128\to256$ & --- & $16\to32$ \\
\midrule
Stage 4 & $256\to256$ & --- & --- \\
\midrule
Encoding & Stage-wise encoders with FPS+$k$-NN; per-stage outputs: $32\to64$, $64\to128$, $128\to256$, $256\to256$ with down sampling & Once (outside stages): GAEncode with $k$-NN; $6\to6$ (no down sampling)& Once (outside stages): GAEncode with $k$-NN; $8\to8$ (no down sampling)\\
\midrule
Pooling & Max over 64 sampled points $\to256$ & Max $\to12$ & Max $\to32$ \\
\midrule
Classifier / Head & $256\to512\to256\to K$ & $12\to8\to K$ & $32\to16\to8\to K$ \\
\midrule
Activation / Normalization & ReLU + BatchNorm & ReLU + BatchNorm & ReLU + BatchNorm \\
\midrule
Structural differences & Baseline architecture & Hierarchy removed; widths reduced; single-scale & Hierarchy removed; widths reduced; single-scale \\
\bottomrule
\end{tabular}
\end{table}

\begin{table}[t]
\centering
\footnotesize
\setlength{\tabcolsep}{6pt}
\renewcommand{\arraystretch}{1.15}
\caption{PointMamba family: Original vs.\ Light vs.\ Mid.}
\label{tab:pointmamba}
\begin{tabular}{p{0.25\textwidth} p{0.24\textwidth} p{0.24\textwidth} p{0.24\textwidth}}
\toprule
\textbf{Component} & \textbf{PointMamba\_Original} & \textbf{PointMamba\_Light} & \textbf{PointMamba\_Mid} \\
\midrule
Input Encoding / Tokens
& Point cloud is partitioned into local point patches by FPS+$k$-NN; each patch is converted to relative coordinates and encoded by a PointNet-style MLP into a 384-dimensional patch token.
& No FPS or $k$-NN grouping; each 3D point is directly treated as a token and projected by a linear layer $3\!\to\!5$, yielding $N$ tokens per traversal (Hilbert/Trans-Hilbert), concatenated to $2N$ tokens width $=5$. 
& Same as Light, a linear point embedding $3\!\to\!5$ is used, producing $2N$ tokens. \\[2pt]
\midrule
SFC Ordering \& Order Indicator
& Two space-filling curves (Hilbert and Trans-Hilbert) and a learnable order-indicator module with per-order scale--shift parameters $(\gamma,\beta)$ are applied to input tokens and the two traversals are concatenated into a single sequence.
& Uses exactly the same pair of SFC traversals and the same order-indicator mechanism $(\gamma,\beta)$ as the original.
& Same SFC ordering and order-indicator mechanism as Light. \\[2pt]
\midrule
Mamba Blocks (Depth \& Width)
& Vanilla Mamba encoder with 12 Mamba blocks; each block operates on 384-dimensional tokens.
& Four compact Mamba blocks with $d_{\text{model}}{=}5$, inner width $d_{\text{inner}}{=}10$ (expand$=2$), state size $d_{\text{state}}{=}6$, and depthwise Conv1d kernel size $d_{\text{conv}}{=}2$.
& Nineteen compact Mamba blocks with $d_{\text{model}}{=}5$, $d_{\text{inner}}{=}10$ (expand$=2$), $d_{\text{state}}{=}6$, and depthwise Conv1d kernel size $d_{\text{conv}}{=}2$. \\[2pt]
\midrule
Selective SSM Implementation
& Selective SSM recurrence $(A,B,C,\Delta t,D)$ implemented via the official \texttt{mamba-ssm} library with fused CUDA kernels and other low-level optimizations.
& Implements the same selective SSM equations in plain PyTorch (state tensor $B\times d_{\text{inner}}\times d_{\text{state}}$), without fused kernels, targeting simple CPU-friendly execution.
& Same PyTorch Mamba implementation as Light; only the dimensionalities $(d_{\text{model}},d_{\text{inner}},d_{\text{state}},d_{\text{conv}})$ differ as above. \\[2pt]
\midrule
Pooling \& Classifier Head
& Global average pooling over the final token sequence, followed by a single linear classifier $384\!\to\!K$.
& LayerNorm on token features, then mean pooling over all $2N$ tokens, followed by linear classifier $5\!\to\!K$.
& LayerNorm on token features, then mean pooling over all $2N$ tokens, followed by linear classifier $5\!\to\!K$. \\[2pt]
\midrule
Activation / Normalization
& Pre-LayerNorm inside each Mamba block; SiLU activation after the depthwise Conv1d branch; sigmoid gate on the SSM output; residual connection $x{+}y$.
& Uses the same block-level structure as the original (pre-LN, SiLU after Conv1d, sigmoid gate, residual connection).
& Same as Light. \\[2pt]
\midrule
Structural Differences 
& Baseline architecture
& Operates directly on point tokens (no FPS, $k$-NN), with drastically reduced widths and depth to fit a $\sim$1.5K-parameter budget and to run with the plain PyTorch Mamba implementation instead of \texttt{mamba-ssm}.
& Same tokenizer and PyTorch Mamba design as Light, but deeper settings targeting a $\sim$7K-parameter budget. \\
\bottomrule
\end{tabular}
\end{table}

\begin{table}[t]
\caption{Mamba3D family: Original vs. Light vs. Mid.}
\label{tab:mamba3d_family}
\centering
\small
\begin{tabular}{p{0.19\textwidth}p{0.25\textwidth}p{0.25\textwidth}p{0.25\textwidth}}
\toprule
Component
& \textbf{Mamba3D\_Original}
& \textbf{Mamba3D\_Light}
& \textbf{Mamba3D\_Mid}
\\
\midrule
Input encoding
& Patch embeddings via FPS ($L$ centers) + $k$-NN (k points/patch) + light-PointNet, then add \texttt{[CLS]} and learnable positional encoding to obtain an $L+1$ token sequence.
& Per-point tokenization with a linear point embedding $3\!\to\!5$, plus a coordinate-based positional embedding added to each token; prepend a learnable \texttt{[CLS]} token (with its own positional parameter). No FPS/$k$-NN patching is used; tokens remain per-point.
& Same tokenizer and LNP as Light, with point embedding $3\!\to\!5$ and learnable \texttt{[CLS]} token.
\\
\addlinespace[0.2em]
\midrule
bi-SSM \& LNP blocks
& 12-layer encoder at width $384$; each layer applies LNP with k=4 and a bi-SSM (L+SSM and C-SSM) following the original setting.
& Two compact blocks (LNP + bi-SSM). L+SSM uses $(d_{\text{model}}, d_{\text{state}}, d_{\text{conv}}, \text{expand})=(5,6,2,2)$, while C-SSM uses $(d_{\text{model}}, d_{\text{state}}, d_{\text{conv}}, \text{expand})=(N\!+\!1,6,2,2)$ (with $N\in\{4,5,6\}$); LNP uses a global all-pairs neighborhood.
& Eight compact blocks with the same per-block settings as Light: L+SSM $(5,6,2,2)$ and C-SSM $(N\!+\!1,6,2,2)$ (with $N\in\{4,5,6\}$); LNP uses a global all-pairs neighborhood.
\\

\addlinespace[0.2em]
\midrule
Selective SSM implementation
& Uses the official \texttt{mamba-ssm} selective SSM module with fused CUDA kernels for both L+SSM and C$-$SSM.
& Reimplements the same selective SSM update in pure PyTorch for L+SSM and C$-$SSM with feature-channel flip; no fused CUDA kernels or low-level optimizations.
& Same PyTorch selective SSM implementation as Light; only depth differs.
\\
\addlinespace[0.2em]
\midrule
Pooling \& classifier head
& Classifies using the final \texttt{[CLS]} token (linear head on \texttt{[CLS]} after the encoder).
& LayerNorm + linear classifier $5\!\to\!K$ on the final \texttt{[CLS]} feature.
& Same as Light: LayerNorm + linear classifier $5\!\to\!K$ on the final \texttt{[CLS]} feature.
\\
\addlinespace[0.2em]
\midrule
Activation / normalization
& Layer normalization and SiLU activation inside each Mamba block with sigmoid gating on the SSM output and residual connections.
& Same block-level normalization and activations as Original (LayerNorm, SiLU, sigmoid gate, residual).
& Same as Light.
\\
\addlinespace[0.2em]
\midrule
Structural differences 
& Baseline architecture
& Keeps the original LNP+bi-SSM+\texttt{[CLS]} architecture but removes FPS/$k$-NN and shrinks embedding, state size and number of blocks to fit a $\sim$1.5K parameter budget; uses a compact PyTorch SSM without fused CUDA kernels.
& Same architecture as Light, with deeper depth, fitting a $\sim$7K parameter budget.
\\
\bottomrule
\end{tabular}
\end{table}

\clearpage
\begin{table}[t]
\caption{RP-EQGNN family: Original vs.\ Light vs.\ Mid.}
\label{tab:rpeqgnn_family}
\centering
\small
\begin{tabular}{p{0.19\textwidth}p{0.25\textwidth}p{0.25\textwidth}p{0.25\textwidth}}
\toprule
Component
& \textbf{RP-EQGNN\_Original}
& \textbf{RP-EQGNN\_Light}
& \textbf{RP-EQGNN\_Mid}
\\
\midrule
Coordinate encoder
& Geometric information encoding circuit (rotation encoding) maps node coordinates into a Hilbert space  (e.g., $\mathrm{RY}/\mathrm{RX}/\mathrm{RY}$ applied to $(x,y,z)$).
& Same as original.
& Same as original.
\\
\addlinespace[0.2em]
\midrule
Node-feature encoder
& The official implementation uses 11-dimensional per-node input features, injected through a feature-encoding circuit with $4\times\mathrm{RY}+7\times\mathrm{RX}$ rotations per qubit.
& Point cloud instantiation: we define a scalar node feature $f_i=\lVert p_i\rVert$ and apply $\mathrm{RY}$/$\mathrm{RX}$ rotations parameterized by $f_i$ with trainable scaling.
& Same as Light, with larger depth.
\\
\addlinespace[0.2em]
\midrule
Edge-conditioned entanglement
& Applies edge-conditioned two-qubit entangling gates (RXX), where a 1D scalar edge feature is used per edge to modulate the interaction strength.
& Complete-graph entanglement using coordinate-derived edge features: distance $d_{ij}=\lVert p_i-p_j\rVert$ and inter-vector angle $\alpha_{ij}$ (w.r.t.\ the origin) modulate two-qubit interactions (IsingXX).
& Same as Light, with larger depth.
\\
\addlinespace[0.2em]
\midrule
Measurement and classical head
& Produces outputs via a global symmetric quantum readout, specifically the expectation value of the observable $ Z^{\otimes N}$, followed by a small classical post-processing MLP($1\to3\to1$).
& Global symmetric readout yielding one scalar quantum feature per sample, followed by an MLP head ($1\!\to\!32\!\to\!32\!\to\!K$); uses the same readout observable $Z^{\otimes N}$ as the official implementation.
& Same readout; optionally a wider MLP head ($1\to32\to64\to 64\to32 \to K$) under a larger budget.
\\
\addlinespace[0.2em]
\midrule
Structural differences 
& Baseline architecture
& Adapts the official structure to point cloud classification by redefining node/edge features from coordinates, implemented in JAX+PennyLane; the same three components are repeated with depth $L{=}50$.
& Same as Light (JAX+PennyLane) with $L{=}50$, scaling capacity primarily via depth/width.

\\
\bottomrule
\end{tabular}
\end{table}
\clearpage

\begin{table}[t]
\centering
\footnotesize
\setlength{\tabcolsep}{6pt}
\renewcommand{\arraystretch}{1.15}
\caption{HyQuRP vs.\ Set-MLP (Light models).}
\label{tab:hyq_setmlp_light}
\begin{tabular}{p{0.26\textwidth} p{0.33\textwidth} p{0.33\textwidth}}
\toprule
\textbf{Component} & \textbf{HyQuRP\_Light} & \textbf{Set-MLP\_Light} \\
\midrule
High-level idea
& Hybrid quantum-classical model.
& Purely classical MLP with pooling, matching HyQuRP\_Light’s classical part. \\
\midrule
Input point shape
& Batches of $N$ 3D points ($N,3$). 
& Batches of $N$ 3D points ($N,3$). \\
\midrule
MLP input representation
& Quantum part produces two types of $\binom{N}{2}$ pairwise measurements ($\binom{N}{2},2$).
& For comparison with HyQuRP, a neural network transforms $(N,3) \to (N,2)$ \\
\midrule
Per-point / token MLP
& MyNN: shared $2\to4\to4$ MLP (with $\tanh$) on each Hamiltonian feature token.
& SimpleNN: shared $2\to4\to4$ MLP (with $\tanh$) on each per-point feature token. \\
\midrule
Pooling 
& Stats over $\binom{N}{2}$ tokens (mean / max / min / sum / std / var) on $4$-d features $\to$ $24$-d vector.
& Same symmetric stats over $N$ tokens on $4$-d features $\to$ $24$-d vector. \\
\midrule
Classifier / head
& $24\to24\to24\to K$ MLP (with $\tanh$) on the global descriptor.
& Same $24\to24\to24\to K$ MLP (with $\tanh$) head as HyQuRP\_Light. \\
\midrule
Explicit symmetry handling
& Rotation and permutation-invariant.
& Only permutation invariance from symmetric pooling. \\
\bottomrule
\end{tabular}
\end{table}

\begin{table}[t]
\centering
\footnotesize
\setlength{\tabcolsep}{6pt}
\renewcommand{\arraystretch}{1.15}
\caption{HyQuRP vs.\ Set-MLP (Mid models).}
\label{tab:hyq_setmlp_mid}
\begin{tabular}{p{0.26\textwidth} p{0.33\textwidth} p{0.33\textwidth}}
\toprule
\textbf{Component} & \textbf{HyQuRP\_Mid} & \textbf{Set-MLP\_Mid} \\
\midrule
High-level idea
& Same hybrid design as Light, but with a wider classical head for more capacity.
& Purely classical MLP with pooling, matching HyQuRP\_Mid’s classical part. \\
\midrule
Input point shape
& Same as HyQuRP\_Light: Batches of $N$ 3D points ($N,3$).
& Same as Set-MLP\_Light: Batches of $N$ 3D points ($N,3$). \\
\midrule
MLP input representation
& Same as HyQuRP\_Light.
& Same as Set-MLP\_Light. \\
\midrule
Per-point / token MLP
& MyNN: shared $2\to8\to16\to32$ MLP (with $\tanh$); higher per-token capacity.
& SimpleNN: shared $2\to8\to16\to32$ MLP (with $\tanh$); higher per-token capacity. \\
\midrule
Pooling 
& Stats over $\binom{N}{2}$ tokens on $32$-d features $\to$ $192$-d vector.
& Same stats over $N$ tokens on $32$-d features $\to$ $192$-d vector. \\
\midrule
Classifier / head
& $192\to32\to16\to8\to K$ MLP (with $\tanh$) on the global descriptor.
& $192\to32\to16\to8\to K$ MLP (with $\tanh$) head as HyQuRP\_Mid. \\
\midrule
Explicit symmetry handling
& Same as HyQuRP\_Light.
& Same as Set-MLP\_Light. \\
\bottomrule
\end{tabular}
\end{table}

\FloatBarrier

\FloatBarrier
\section{Additional Details}
\label{app:additional_details}

\subsection{Details of summary metrics}
\label{app:summary_metrics}

We follow the reporting convention in Table~\ref{tab:overall_results} and specify the exact aggregation rules used to compute $\pm\sigma$ and the average rank. $\sigma$ denotes the sample standard deviation across seeds, computed with Bessel's correction (division by $S-1$), where $S=7$. 
There were a small number of failed runs (3 out of 1,008) where the training loss did not decrease, likely due to loose hyperparameter tuning and instability of neighborhood construction in the sparse-point regime. For these runs, we retried training with the adjacent learning rates on the predefined grid 
$\{10^{-1},10^{-2},10^{-3},10^{-4},10^{-5}\}$, using the lower neighbor first and the higher neighbor second. With this fixed rule, we observed no further failures.
For each setting (dataset, point budget, capacity), we rank models by the mean test accuracy; ties are broken by smaller standard deviation. We then average the ranks across all settings. For the pairwise performance gaps reported in the subsection~\ref{subsec:res_ana}, we report the mean and sample standard deviation of seed-matched accuracy differences across seeds.

\subsection{Logit comparisons under rotations and permutations}
\label{app:logit_comparisons}
Let $z, z' \in \mathbb{R}^{C}$ denote the logit vectors produced by the same model
for an input $x$ and its transformed version $g\!\cdot\!x$ (e.g., rotation/permutation), respectively.
We define:
\begin{equation}
s_{\mathrm{cosine}}(z,z')
:= \frac{z^\top z'}{\|z\|_2 \, \|z'\|_2}.
\end{equation}
\begin{equation}
s_{\mathrm{norm}}(z,z')
:= \frac{\|z'\|_2}{\|z\|_2}.
\end{equation}
Note that $s_{\mathrm{norm}}(z,z')$ may exceed 1.

To probe logit-level invariance, we compare the logits for an input and those for its
rotated or permuted counterpart using the two quantities defined above. For selected baselines that do not enforce both symmetries by construction (MLP, Set-MLP, and PointNet), logits were extracted after
10 epochs of training to examine whether approximate invariance can be learned from
data. For VN-PointNet, TFN, and HyQuRP, the same quantities were evaluated without
training, since invariance is built into the architecture and does not depend on learned
weights. Accordingly, Table~\ref{tab:model_logit_comparison} reports logit-level invariance comparisons rather than predictive performance.

\par\medskip
\refstepcounter{table}\label{tab:model_logit_comparison}
\begin{center}
\setlength{\tabcolsep}{6pt}
\renewcommand{\arraystretch}{1.08}
\begin{tabular}{lcc}
\toprule
Model & cosine similarity & $\ell_2$-norm ratio \\
\midrule
MLP         & 0.37 & 0.66 \\
Set-MLP     & 0.69 & 1.18 \\
PointNet    & 0.98 & 0.63 \\
VN-PointNet & 1.00 & 1.00 \\
TFN         & 1.00 & 1.00 \\
HyQuRP      & 1.00 & 1.00 \\
\bottomrule
\end{tabular}

\vspace{0.5em}

\parbox{0.92\linewidth}{\small
\textbf{TABLE~\thetable.} \textbf{Logit Comparisons under Rotations and Permutations.}
The table reports the cosine similarity and the $\ell_2$-norm ratio between the logits
produced for an input and those produced for its transformed counterpart. Values closer
to 1 indicate stronger invariance.}
\end{center}
\par\medskip

\end{document}